\documentclass[11pt]{amsart}
\usepackage{amsthm}

\usepackage[all]{xy}
\usepackage{amssymb}
\usepackage{enumerate}
\usepackage{mathrsfs}
\usepackage{hyperref}
\usepackage{epsfig}
\usepackage{graphicx}
\usepackage{subfig}
\usepackage{float}
\usepackage{epigraph}

\numberwithin{equation}{section}

\newtheorem{thm}{Theorem}[section]

\newtheorem{lem}[thm]{Lemma}

\newtheorem{rem}{Remark}[section]
\newtheorem{example}[thm]{Example}

\renewcommand{\Re}{\operatorname{\rm Re}}
\renewcommand{\Im}{\operatorname{\rm Im}}

\newcommand{\beqast}{\begin{eqnarray*}}
\newcommand{\eqast}{\end{eqnarray*}}
\newcommand{\beqa}{\begin{eqnarray}}
\newcommand{\eqa}{\end{eqnarray}}

\newcommand{\bbe}{\begin{equation}}
\newcommand{\ee}{\end{equation}}

\newcommand{\dd}{\partial}

\newcommand{\mbr}{\medbreak}

\newcommand{\bfo}{{\bf 1}}

\newcommand{\bM}{{\mathbb M}}
\newcommand{\bE}{{\mathbb E}}
\newcommand{\bP}{{\mathbb P}}
\newcommand{\bQ}{{\mathbb Q}}

\newcommand{\bR}{{\mathbb R}}
\newcommand{\bC}{{\mathbb C}}
\newcommand{\bZ}{{\mathbb Z}}

\newcommand{\cF}{{\mathcal F}}
\newcommand{\cH}{{\mathcal H}}

\newcommand{\cE}{{\mathcal E}}
\newcommand{\cG}{{\mathcal G}}

\newcommand{\cS}{{\mathcal S}}

\newcommand{\cL}{{\mathcal L}}

\newcommand{\cM}{{\mathcal M}}
\newcommand{\cC}{{\mathcal C}}

\newcommand{\cZ}{{\mathcal Z}}

\newcommand{\dt}{{\Delta t}}

\newcommand{\hu}{\hat u}
\newcommand{\hv}{\hat v}

\newcommand{\hG}{\hat G}

\newcommand{\hV}{\hat V}

\newcommand{\eps}{\epsilon}
\newcommand{\de}{\delta}
\newcommand{\al}{\alpha}
\newcommand{\be}{\beta}
\newcommand{\bep}{\beta^+}
\newcommand{\bem}{\beta^-}

  \newcommand{\sg}{\sigma}
\newcommand{\sgm}{\sigma_-}
\newcommand{\sgp}{\sigma_+}

\newcommand{\De}{\Delta}
\newcommand{\la}{\lambda}
\newcommand{\lp}{\lambda_+}
\newcommand{\lm}{\lambda_-}
\newcommand{\La}{\Lambda}
\newcommand{\mum}{\mu_-}
\newcommand{\mup}{\mu_+}
\newcommand{\omp}{\om_+}
\newcommand{\omm}{\om_-}

\newcommand{\ka}{\kappa}

\newcommand{\om}{\omega}

\newcommand{\ze}{\zeta}

\newcommand{\ga}{\gamma}

\newcommand{\Ga}{\Gamma}

\newcommand{\tf}{\tilde f}

\newcommand{\tG}{{\tilde G}}

\newcommand{\tV}{\tilde V}

\newcommand{\bepq}{\be^+_q}
\newcommand{\bemq}{\be^-_q}{

\newcommand{\barX}{\overline{X}}
\newcommand{\uX}{\underline{X}}

\newcommand{\cEq}{\cE_q}
\newcommand{\cEpq}{\cE^+_q}
\newcommand{\cEmq}{\cE^-_q}
\newcommand{\phipq}{\phi^+_q}
\newcommand{\phimq}{\phi^-_q}

\begin{document}

\title
[Static and semi-static hedging]
{Static and semi-static hedging as contrarian or conformist bets}

\author[
Svetlana Boyarchenko and
Sergei Levendorski\u{i}]
{
Svetlana Boyarchenko and
Sergei Levendorski\u{i}}
%\address{S.B.: Department of Economics, The
%University of Texas at Austin, 1 University Station C3100, Austin,
%TX 78712--0301, {\tt leven@eco.utexas.edu}}
%}

\thanks{The authors are grateful to Nina Boyarchenko for the comments on the qualitative
effects of semi-static hedging described in the paper. The remaining errors are ours.\\
\emph{S.B.:} Department of Economics, The
University of Texas at Austin, 1 University Station C3100, Austin,
TX 78712--0301, {\tt sboyarch@eco.utexas.edu} \\
%\emph{M.B.:} Department of Mathematics, University of Michigan,
%530 Church Street, 2704 East Hall\\Ann Arbor, MI 48109-1043, USA. Email address: {\tt mityab@umich.edu}\\
\emph{S.L.:}
Calico Science Consulting. Austin, TX.
 Email address: {\tt
levendorskii@gmail.com}}
%\\
%The author is grateful to the participants of Mathematical Finance seminar at ETH Z\"urich, June 20, 2013, AMS Meeting, Philadelphia, October 12-13, 2013,
% Conference ``Mathematical finance and Partial Differential Equations 2013", Rutgers University, November 1, 2013,
% and, especially, to Paul Embrechts, Josef Teichmann and Martin Keller-Ressel, for discussions and useful comments and suggestions.
% The remarks and suggestions made by the anonymous referee on the first version of the paper are appreciated, and
% lead to important changes in the structure of the paper. The errors are mine.}
%(corresponding author).}

\begin{abstract}
In this paper, we argue that, once the costs of maintaining the hedging portfolio are properly taken into account, semi-static portfolios should more properly be thought of as separate classes of derivatives, with non-trivial, \textit{model-dependent} payoff structures. 
We derive new integral representations for payoffs of exotic European options in terms of payoffs of vanillas,
different from Carr-Madan representation, and suggest approximations of the idealized static hedging/replicating portfolio
using vanillas available in the market. We study the dependence of the hedging error on a model used for pricing and
show that the variance of the hedging errors of  static hedging portfolios can be sizably larger than the errors of variance-minimizing portfolios.
We explain why the exact semi-static hedging of barrier options is impossible for processes with jumps, and derive general formulas
for variance-minimizing semi-static portfolio. We show that hedging using vanillas  only leads to larger errors than hedging
using vanillas and first touch digitals. In all cases, efficient calculations of the weights of the hedging portfolios are
in the dual space using new efficient numerical methods for calculation of the Wiener-Hopf factors and Laplace-Fourier inversion.

\end{abstract}
\maketitle


\section{Introduction}\label{intro}
%By now, there is a large body of the literature devoted to static hedging/replication 
%of exotic European options, and semi-static hedging/replication of options with
%barrier features, and options of other types. See \cite{derman-static,CarrChou97,CarrEllisGupta98,
%carr-madan-statichedge,AnderAndrEliezer02,HirsaCourtadonMadan02,
%ADGS05,NalholmPaulsen06,CarrLee09,CarrWu14,kirkby-deng14} and the bibliographies therein. 
%In a number of studies, e.g., \cite{Tompkins02,NalholmPaulsen06,CarrWu14},  the advantages of static and semi-static hedging
%when compared to the delta-hedge are demonstrated. 

There is a large literature\footnote{See \cite{derman-static,CarrChou97,CarrEllisGupta98,
carr-madan-statichedge,AnderAndrEliezer02,HirsaCourtadonMadan02,Tompkins02,
ADGS05,NalholmPaulsen06,CarrLee09,CarrWu14,kirkby-deng14,CarrWu14} and the bibliographies therein.} 
studying static hedging and replication of exotic European options, and semi-static hedging and replication of barrier and other types of options. What this literature ignores, however, is the cost of maintaining the hedging position, which can drive the payoff of the overall portfolio negative. In this paper, we argue that, once the costs of maintaining the hedging portfolio are properly taken into account, semi-static portfolios should more properly be thought of as separate classes of derivatives, with non-trivial, \textit{model-dependent} payoff structures. Depending on the structure of the option being hedged {\em and the model}, the semi-static hedging portfolio may either function as a contrarian bet Ñ small losses with high probability and large gains with low probability Ñ or as a conformist bet Ñ small gains with high probability and large losses with low probability. 

We suggest new versions of static and semi-static hedging,
provide qualitative analysis of errors  of different static and semi-static procedures, explain why in the jump-diffusion case
the exact replication of barrier options by European options, hence, the model-independent replication, is impossible, 
and produce numerical examples to demonstrate how different sources of hedging errors depend on a model. 
In the main body of the paper, we consider European and down-and-in barrier options in L\'evy models, and then
indicate the directions in which the approach of the paper can be generalized and extended to cover options of other types,
in more complicated models. Pricing barrier options and the calculation of the variance of the hedging
portfolio at expiry are based on new efficient numerical procedures for calculation of the Wiener-Hopf factors and Laplace-Fourier inversion. These procedures
can be useful in other applications as well.

The underlying idea of the static hedge \cite{carr-madan-statichedge} of European options
with exotic payoffs
is simple. One replicates the payoff of an exotic European option by a linear combination payoffs of the underlying stock and vanillas,
and uses the portfolio of the stock and options to replicate or hedge the exotic option. 
In Section \ref{static_European},  we start with the derivation of  integral representations for an exact static hedging portfolio.
Contrary to \cite{carr-madan-statichedge}, we work in the dual space, and derive a representation
in terms of vanillas only; this representation if different from the one in \cite{carr-madan-statichedge}.
By construction, the portfolios we construct and the portfolio in \cite{carr-madan-statichedge} are model-independent,
which looks very attractive. However, 
%these ideal model-independent portfolios involve integrals of a continuum of vanillas. 
the continuum of vanillas does not exist, and even if it did, 
the integral portfolio would had been impossible to construct anyway.
Hence, one has to approximate each integral by a finite sum. 
The hedging error
of the approximation is inevitably model-dependent. We design simple constructions of approximate hedging portfolios and 
 study the dependence of the static hedging error on a model
using a portfolio of available vanillas. We derive an approximation in an almost $C(\bR)$-norm, and then 
calculate the weights of the variance-minimizing hedging portfolio. In both cases, the calculations
are in the dual space using the sinh-acceleration technique \cite{SINHregular}. 
%are calculated working in the dual space as well. Essentially, we need to price  options whose payoffs are 
%products of the payoffs of the options in the portfolio. Once the Fourier transforms of these products 
%is calculated, the same efficient sinh-acceleration technique can be applied.
We believe that both approximate hedging procedures have certain advantages as compared to the two-sep procedure
 in a recent paper \cite{kirkby14}, where, first, one uses the projection of the payoff of the security to be hedged and securities
 in the hedging portfolio on the space of model
 payoffs, then calculates the weights. This more complicated procedure does not help to decrease the hedging error, and, similarly
 to the approximate static hedging that we construct, cannot produce smaller variances than the variance minimizing hedging portfolio.

 In Section \ref{hedgingdownandin}, we outline the general structure
of the semi-static variance-minimizing hedging of barrier options; in the paper, we consider the down-and-out and down-and-in options.
The initial version of the semi-static hedging portfolio for barrier options was suggested in \cite{derman-static}:
put options with strikes equal to the barrier, with different expiry dates, are added to the portfolio in such a way
that the portfolio value is zero at the barrier. {\em Assuming that, at  the moment the barrier is breached, the underlying is
exactly at the barrier}, the weights of portfolio can be calculated backwards. It is clear that if the underlying can cross the barrier with a jump,
the procedure cannot be exact, and the implicit error is inevitably model-dependent. 
A different semi-static hedging of barrier options is developed in \cite{CarrChou97,CarrEllisGupta98,CarrLee09},
but the underlying assumption is the same as in \cite{derman-static}. 
%, the semi-static hedging of barrier options in the Black-Scholes model
%and more general diffusion models is developed. 
For a given barrier option, an exotic European payoff $\cG_{ex}$ is constructed so that,
at maturity or at the time of early expiry (the case of ``out" options) or activation (the case of ``in" options), the price of the 
hedging portfolio for barrier option
equals the price of the European option.  At the barrier is reached (the presumption is that it cannot be crossed by a jump), 
the portfolio is liquidated.
The European option being exotic, an approximate static hedging portfolio for
the latter is presumed to be used. Hence, in the presence of jumps, the hedging errors are model-dependent even if one believes that
an auxiliary exotic option can be hedged exactly using a portfolio of vanillas, and the question of the interaction of two types
of errors naturally arises. 
The  option with the payoff $\cG_{ex}$ is more exotic than the usual exotic options (the structure of the payoff
is more complicated), and the more exotic the option is, the larger the hedging errors are. Even in the case of diffusion models,
the errors can be  quite sizable,  
and the approximation is 
 justified under a certain rather restrictive symmetry condition
on the parameters of the model.
See \cite{NalholmPaulsen06}.

The paper \cite{kirkby-deng14} uses the approximate semi-static hedging of \cite{CarrLee09} and an approximation of the exotic European option which approximates
the barrier option; this leads to at least two sources of model-dependent errors, which can be large if the jump component is sizable;
in addition, the symmetry condition is more restrictive than in the case of diffusion models.  
In the introduction of \cite{kirkby-deng14}, it is claimed that Carr and Lee \cite{CarrLee09} rigorously justified
the semi-static procedure for jump-diffusion models; the picture is more complicated. In Section \ref{ss:semi-static},
we explain that the standard semi-static construction has numerous sources of errors,
and even an approximation can be justified under additional rather restrictive conditions only. In particular, in the presence of jumps,
the semi-static procedure is never exact.

The variance minimizing hedging portfolio has certain advantages. We can directly construct
the hedging portfolio using the securities traded in the market provided that a pricing model is chosen, and one can calculate
 the option prices $V_j$ in the portfolio and products of the prices as functions
of $(t, x)$,  $0\le t\le T$, where $T$ is time to maturity. Accurate and fast calculations are  possible for wide classes of options 
(barrier options, lookbacks, American options, Asians, etc.), and many popular pricing methods working in the state space can be applied.
However, to calculate the weights of the hedging portfolio, we need to calculate the expectations 
of the products of the discounted prices at time $\tau\wedge T$, where $\tau$ is the first entrance time
into the early exercise region.
%: $\bE^x\left[e^{-(\tau\wedge T) r}V_j(T-\tau , X_{\tau\wedge T}) e^{-(\tau\wedge T) r}V_\ell(T-\tau , X_{\tau\wedge T})\right]$. 
Hence, one needs to approximate the products of prices by functions which are amenable to application of
efficient option pricing techniques, which are, typically, based on the Laplace-Fourier transform.
In the paper, we suggest and use new efficient methods for the numerical Fourier-Laplace inversion and calculation
of the Wiener-Hopf factors; these methods are of a general interest.
 
We work in the dual space; calculations in the dual space are also necessary to accurately address the following practically important effect.
There is an additional source of errors of hedging portfolios consisting of vanillas only. In all popular models used in finance,
the prices of vanilla options are infinitely smooth before the maturity date and up to the boundary but prices of barrier options in L\'evy models are not smooth at the boundary,
the exceptions being double jump diffusion model, hyper-exponential jump diffusion model, and other models with rational characteristic functions.
For wide classes of purely jump models, it is proved in \cite{NG-MBS,BIL} that the price of an ``out" barrier option near the barrier
behaves as $c(T)|x-h|^\ka$, where $\ka\in [0,1)$ is independent of time to maturity $T$, $c(T)>0$, and $|x-h|$ is 
the log-distance from the barrier. For finite variation
processes with the drift pointing from the boundary, $\ka=0$, and the limit of the price at the barrier is positive. 
Similarly, the price of the first touch digital behaves as $1-c_1(T)|x-h|^\ga$. Even if  the diffusion component is present, the prices
of the barrier and first touch options are not differentiable at the barrier \cite{NG-MBS}, and if the diffusion component is small,
then essentially the same irregular behavior of the price will be observed outside a very small vicinity of the barrier.
Calculations in the state space are based on approximations by fairly regular functions, 
hence, cannot reproduce these effects sufficiently accurately. See examples in \cite{paired}, where it is demonstrated
that Carr's randomization method \cite{single}, which relies on the time randomization and interpolation in the state space,
underprices barrier options in a small vicinity of the barrier.
From the point of view of the qualitative composition of the hedging portfolio, one should expect 
that an accurate hedging of barrier options is impossible unless the corresponding first touch digitals are included.
 Fig.~\ref{Graph3V} clearly shows that the first touch digital is much closer to
the down-and-in option than a put option, and the first-touch options with the payoffs $(S/H)^\ga, \ga>0,$ would be even better
hedging instruments.
%{Graph3V.eps}{LpLmgraph}
\begin{figure}
\scalebox{0.5}
{\includegraphics{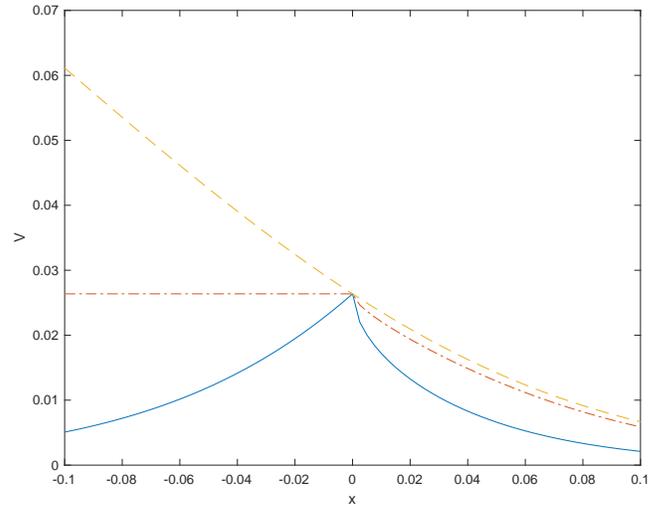}}
\caption{\small Solid line: the down-and in call option  used in the numerical example in Section
\ref{numer_din_example} as the function of $x=\ln(S/H)$. Dots-dashes: a multiple of the first touch digital. Dashes: a multiple of
the put option with strike $K=H$. Maturity is 0.1. }
\label{Graph3V}
\end{figure}

We calculate hedging portfolios consisting of vanillas only and of vanillas and the first-touch option
in in Sections
\ref{no-touch dual} and \ref{first-touch dual} using the Wiener-Hopf factorization technique.
We recall the latter in Section \ref{s:WHF},
and introduce the new efficient method for the calculation of the Wiener-Hopf factors based on the sinh-acceleration technique \cite{SINHregular}. 
%and demonstrate with numerical examples that the hedging errors of the former are much larger than the errors
%of the latter unless the spot is rather far from the barrier.
 The numerical examples for static hedging and calculation of the Wiener-Hopf factors and expectations
 related to barrier options are discussed at the end of the corresponding Sections;
a numerical example for hedging of barrier options is in Section \ref{numer_din_example}. 
In Section \ref{concl} we summarize
the results of the paper and outline natural extensions. The outline of Gaver-Stehfest-Wynn method, 
%calculation of the Fourier transforms of the payoff functions and their products, 
 and Tables and Figures (with the exception of Fig.~\ref{Graph3V} above and 
Fig.~\ref{ContoursSINH1}
 in
Section \ref{hedgingdownandin})
are relegated to Appendices.

\section{Static hedging of European options}\label{static_European}
\subsection{Static hedging in the ideal world}\label{stathedgeideal} 
Let $X$ be the L\'evy process, and let $G(X_T)=\cG(\ln X_T)$ be the payoff at maturity. 
 We assume that $\cG$ is continuously differentiable, the measure $d\cG'$ is a sum of a finite number
of atoms (equivalently, $\cG$ has a finite number of kinks), and, under additional weak regularity conditions, 
prove that if $\cG(S)$ polynomially decays as $S\to +0$ (call-like options), then
\bbe\label{reprCall}
\cG(S)=\int_0^{+\infty}(S-K)_+ d\cG'(K),
\ee
and if  $\cG(S)$ polynomially decays as $S\to +\infty$ (put-like options), then
\bbe\label{reprPut}
\cG(S)=\int_0^{+\infty}(K-S)_+ d\cG'(K).
\ee
\vskip0.1cm\noindent
{\sc Assumption $(\cG)$.}  $\cG$ is continuously differentiable function satisfying the following conditions
\begin{enumerate}[(i)]
%\item
%there exists $A>0$ s.t. either (1) $\cG(S)=0, \forall\ S\in (0,A)$ or (2) $\cG(S)=0, \forall\ S\ge A$;
\item
$d\cG^{\prime}$ is a signed measure, without the singular component;
\item
$\cG'$ has only a finite number of points of discontinuity; 
\item
the measure $d\cG'_{at}=\sum_{Y>0}(\cG'(Y+0)-\cG'(Y-0))\de_Y$ is finite;
\item
$\exists$ $\be\in (-\infty, -1)\cup (0,+\infty)$ s.t.  $S\mapsto S^{-\be-1}\cG(S)$,
$S\mapsto S^{-\be-1}\cG'(S)$ are of the class $L_1$.
\end{enumerate}
If $\be<-1$, define 
\bbe\label{cGreprcall}
\cG_1(S)=\cG(S)-\sum_{Y>0}(\cG'(Y+0)-\cG'(Y-0))(S-Y)_+,
\ee
and if $\be>0$, set
\bbe\label{cGreprput}
\cG_1(S)=\cG(S)-\sum_{Y>0}(\cG'(Y+0)-\cG'(Y-0))(Y-S)_+.
\ee 
Clearly, the measure $d\cG_1$ is absolutely continuous, and $d\cG'=d\cG'_{at}+d(\cG_1)'$.
\begin{thm}\label{thm:stathedgeEuro}
Under Assumption $(\cG)$, if $\be>1$ (resp., $\be<-1$), \eqref{reprCall} (resp., \eqref{reprPut}) holds.
\end{thm}
\begin{proof}
In view of \eqref{cGreprcall}-\eqref{cGreprput}, it suffices to consider the case $d\cG'=d\cG'_1$.
Set $G_1(x)=\cG_1(e^x)$. %, and notice that $d(\cG_1)'(e^x)=e^x d(G_1)'(e^x)$.  
%The following steps are for Case (1);
%in Case (2), the proof is similar. 

%Take $\om<-\max\{1,\be\}$, 
Consider the case $\be>0$. 
Set   $k=\ln K$,  and calculate the Fourier transform of the RHS of \eqref{reprPut}
w.r.t. $x:=\ln S$, for $\xi$ on the line $\{\Im\xi=\be\}$. 
%the half-plane $\Im\xi<-\max\{1,\be\}$. 
Using Fubini's theorem, we obtain 
\beqast
\int_{-\infty}^{+\infty} e^{-ix\xi}\int_0^{+\infty}(K-S)_+ d(\cG_1)'(K)dx
%\\&=&\int_{-\infty}^{+\infty} e^{-ix\xi}\int_{-\infty}^{\infty}(e^k-e^x)_+ d(\cG_1)'(K)dx\\
&=&\int_{0}^{\infty}\int_{-\infty}^{\infty}e^{-ix\xi}(K-e^x)_+ dx\,d(\cG_1)'(K)\\
&=&\int_{0}^{\infty} \frac{K^{1-i\xi}}{i\xi(i\xi-1)}d(\cG_1)'(K)\\&=&
\int_{\De G'(k)=0} \frac{e^{-i\xi k}}{i\xi(i\xi-1)}K^2 \frac{d^2 \cG_1}{dK^2}(K)dk,
\eqast
where $\De G'(k)=G'(k+0)-G'(k-0)$.
Using $K\frac{d}{dK}=\frac{d}{dk}$ and $K^2\frac{d^2}{dK^2}=\frac{d^2}{dk^2}-\frac{d}{dk}$, then integrating by parts
and taking into account that 
%1) $G_1(k)=0, k<\ln A$, and 2) 
$e^{-ik\xi}G'_1(k)$ and $e^{-ik\xi}G_1(k)$ tend to 0 as $k\to\pm \infty$, we continue
\beqast
=\int_{\De G'(k)=0} \frac{e^{-ik\xi}}{i\xi(i\xi-1)}d(G'_1(k)-G_1(k))
&=&\int_\bR \frac{e^{-ik\xi}}{i\xi-1}(G'_1(k)-G_1(k))dk\\
&=&\int_\bR \frac{e^{-ik\xi}}{i\xi-1}dG_1(k)-\int_\bR \frac{e^{-ik\xi}}{i\xi-1}G_1(k)dk\\
&=&\int_\bR e^{-ik\xi}G_1(k)dk.
\eqast
Thus, in the case $\cG=\cG_1$, the Fourier transforms of the LHS and RHS of \eqref{reprPut} coincide on
 the line $\Im\xi=\be$,
which proves \eqref{reprPut}.
If $\be<-1$, then, in the proof above, we replace $(K-S)_+$ with $(S-K)_+$, and modify all the steps accordingly.
\end{proof}

\begin{rem} {\rm If $\cG$ has the compact support and no atoms, then both representations, in terms of puts and calls,
are valid, {\em with integration w.r.t. the same measure. } This (mildly surprising) fact can be verified using  
the put-call parity and the following calculation
\beqast
\int_0^{+\infty}(S-K)\cG^{\prime\prime}(K)dK&=&\cG'(K) (S-K)\vert_0^{+\infty}+\int_0^{+\infty} \cG'(K)dk\\
&=&
(\cG'(K) (S-K)+\cG(K))\vert_0^{+\infty}=0-0=0.
\eqast
}
 \end{rem}
 \begin{example}
{\rm
The payoff function of the powered call of order $\be>1$ with the strike $K_0$ is
$\cG(S)=((S-K_0)_+)^\be$. %, and the Fourier transform is $\hG(\be, K_0;\xi)= -\be\frac{K_0^{1-i\xi/\be}}{\xi(\xi+i\be)}$.
Since $\be>1$, there are no kinks, $d\cG'(K)$ has no atoms, and
%For $S>K_0$, we calculate $\cG'(S)=\be (S-K_0)^{\be-1}$, $\cG^{\prime\prime}(x)=\be(\be-1) (S-K_0)^{\be-2}e^{2x}$.
%$G^{\prime\prime}(x)-G'(x)=\be(be-1) (e^x-K_0)^{\be-2}e^{2x}.$
%Hence, 
\[
((S-K_0)_+)^\be=\be(\be-1)\int_{K_0}^{+\infty} dK\, (K-K_0)^{\be-2} (S-K)_+.
\]
}
\end{example}

%\subsubsection{A model example 2: hedging of power calls}\label{ss:powerhedgeideal}
\begin{example}
{\rm 
The payoff function of the power call is $\cG(S)=(S^\be-K_0)_+$, where $\be>0$, $K_0>0$.
The measure $d(\cG')(K)$ has an atom $\be\de_{K_0}$, and, on $(K_0,+\infty)$, $\cG_{KK}(K)=\be(\be-1)K^{\be-2}$.
Hence, 
\[
(S^\be-K_0)_+=\be(S-K_0)_++\be(\be-1)\int_{K_0}^{+\infty} dK\,  K^{\be-2}(S-K)_+.
\]
}
\end{example}

%\subsubsection{A model example 3: 
%hedging of the European option in the semi-static procedure for down-and-in call options}%\label{ss:semistatichedgeideal}

\begin{example}
{\rm 
Consider the down-and-in call option with barrier $H$ and strike $K_0>H$. A popular semi-static replicating portfolio
for this option is a European option with the payoff $\cG_{ex}(S)=((S-K_0)_++(S/H)^{\be}(H^2/S-K_0)_+)\bfo_{(0,H]}$
(see \eqref{stHedgeEur} and Section \ref{SemiStaticLevy}). Since $K_0>H$, the payoff simplifies 
$\cG_{ex}(S)=(S/H)^{\be}(H^2/S-K_0)_+)\bfo_{(0,H]}(x)=H^{-\be}K_0\cG_{ex,1}(S)$,
 where $\cG_1(S)=S^{\be-1}(K_1-S)_+$, and $K_1=H^2/K_0<H$. Clearly, it suffices to construct
 the hedging portfolio for the option with the payoff function $\cG_1$, which vanishes above $K_1$. At $K_1$, $\cG_1$ has a kink, and $G'_1(K_1-0)=-K_1^{\be-1}$. On $(0,K_1)$, $\cG_1$ is infinitely smooth, and $\cG'_1(S)=K_1(\be-1)S^{\be-2}-\be S^{\be-1}$, $\cG^{\prime\prime}(S)=
 K_1(\be-1)(\be-2)S^{\be-3}-\be(\be-1)S^{\be-2}=(\be-1)S^{\be-2}((\be-2)K_1/S-\be)$. 
 Hence, 
 \[
\cG_{ex}(S)=H^{-\be}K_0\left(K_1^{\be-1}(K_1-S)_++(\be-1)\int_0^{K_1}dK\,K^{\be-2}((\be-2)K_1/K-\be)(K-S)_+\right).
\]

}
\end{example}

\subsection{Approximate static hedging}\label{appstatic}
In the real world, only finite number of options are available, hence, one has to approximate the measure $d\cG'(K)$ using
an atomic measure, typically, with a not very large number of atoms. 
For instance, it is documented in \cite{CarrWu14} that static hedging with 3-5 options produces good results.
Hence, the static hedging will be approximate.
Furthermore, as seen from time 0, the hedging error will depend on the choice of the model used; although the idealized static hedge
is model independent, the approximate one is model-dependent, and
the quality of the approximation depends (naturally) on the choice of the approximation procedure.
We will approximate the payoff of an exotic option by linear combinations of payoffs of vanillas, in the norm
of a Sobolev space with an exponential weight.
To be more specific, we minimize the difference between the payoff of an exotic option and the portfolio payoff
$G(x):=\cG(e^x)$
in the norm of the Sobolev space $H^{s,\om}(\bR)$ of order $s$, with an appropriate
exponential weight $e^{\om x}$ (a.k.a. dampening factor). The Plancherel theorem allows us to do the calculations
in the dual space. The integrals are calculated accurately and very fast using the sinh-acceleration techniques \cite{SINHregular}.
If $s>1/2$, $H^{s}(\bR)$ is continuously embedded into $C(\bR)$, hence, we can estimate the error in the $C$-norm (with the
corresponding weight), which
is natural for the approximate static hedge: if the error 0 is not achievable, we control the maximal error.
We study the dependence of the variance of the hedging error on the model and $s,\om$, and demonstrate that the variances of
errors in cases
$s=1/2$ and $s>1/2$ close to $1/2$ are comparable (and essentially independent of $\om$ in a reasonable range), the differences being smaller than the differences between the 
variances of errors of approximate static hedging portfolios and variance-minimizing portfolio. 

Let $\om, s\in \bR$. The Sobolev space $H^{s,\om}(\bR)$ of order $s$, with weight $e^{\om x}$, is the space of the generalized functions $u$ 
such that $u_\om:=e^{\om \cdot}u\in H^s(\bR)$. The scalar product in $H^{s,\om}(\bR)$ is defined by $(u,v)_{s;\om}=(u_\om, v_\om)_{H^s(\bR)}$. Thus,
\[
(u,v)_{H^{s,\om}(\bR)}=\int_{\Im\xi=\om}(1+(\xi-i\om)^2)^{s}\hu(\xi)\overline{\hv(\xi)}d\xi.
\]
By one of the Sobolev embedding theorems (see, e.g., Theorem 4.3 in Eskin (1973)), if $s>1/2$, $H^{s}(\bR)$ is continuously embedded into $C_0(\bR)$, 
the space of uniformly bounded continuous functions vanishing at infinity, with $L_\infty$-norm.
Hence, for any $\om\in \bR$ and $s>1/2$, an approximation in the $H^{s,\om}$-topology gives a uniform approximation over any fixed compact $K$.

Consider an exotic option whose payoff vanishes below $K$, which we normalize to 1. 
For practical purposes, we may assume that the strikes of European options used for hedging are close to 1, and the spot is close to 1; 
hence, the log-spot $x$ is close to 0, and if $\om$ is not large in modulus, the differences among the hedging weights for different omegas are not large. Likewise,
if $\hu(\xi)$ decay fairly fast at infinity, the norms of $u$ in $H^{s,\om}(\bR)$ will be close if $s\in [1/2, s_0]$ and $s_0$ is close to $1/2$. %Hence, we will use an ad hoc approximation
%in $H^{1/2,\om}$, where $\om\le \be$ if $\be<-1$ and $\om\ge \be$ if $\be>0$ is fixed, and is not too close to $\be$. 

Assume that $\be<-1$. Thus, we have a call-like options, which is hedged using a portfolio of
call options. We fix $\om\le \be$ as discussed above, $s\ge 1/2$,  and
the set of call options with the payoff functions $G_j:=G(K_j;\cdot)$. Set $G_0=G$. We look for the set of weights $n=(n_1,\ldots, n_N)$ (numbers of call options in the portfolio) which minimizes
$
StHG(n):=-G_0+\sum_{j=1}^N n_jG_j
$
in the $H^{s;\om}(\bR)$ norm. Denote $G^{s;\om}_{jk}=(G_j,G_k)_{s;\om}$; 
these scalar products can be easily calculated with 
a sufficiently high precision since the integrands in the formula for $G^{s;\om}_{jk}$ decay as $|\xi|^{-4+s}$.
Furthermore, if $\hG_0(\xi)$ is of the form $e^{-ik_0\xi}\hG_{00}(\xi)$, where $\hG_{00}(\xi)$ is a rational function, and $k_0\in\bR$,
then the integrands in the formulas for $(G_j,G_k)_{s;\om}$ are of the form 
$
e^{-ik_{jk}\xi}\hG_{jk,0}(\xi)(1+(\xi-i\om)^2)^s$, where $k_{jk}\in\bR$ and
$\hG_{jk,0}(\xi)$ are rational functions. Hence, the scalar products can be calculated with almost machine precision and very fast using the sinh-acceleration
technique \cite{SINHregular}.  After an appropriate change of variables of the form
$\xi=i\om_1+b\sinh(i\om'+y)$, the simplified trapezoid rule with  a dozen of terms typically suffices to satisfy the error tolerance of the order of
$10^{-10}$ and less.

The minimizer $n^{s;\om}$ of
\[
F^{s;\om}(n):=\|StHG(n)\|^2=
G^{s;\om}_{00}-\sum_{j=1}^Nn_jsG^{s;\om}_{0j}+\sum_{j=1}^Nn_j^2G^{s;\om}_{jj}+\sum_{j\ne k, 1\le j,k\le N}^Nn_jn_kG^{s;\om}_{jk}
\]
is the solution of the equation $\nabla F^{s;\om}(n)=0$, equivalently,
%\[
%-G^{s;\om}_{0j}+2G^{s;\om}_{jj}w_j+2\sum_{k\ne j, 1\le k\le N}^Nw_kG^{s;\om}_{jk}=0,\quad j=1,2,\ldots, N.
%\]
\bbe\label{statwsom}
n^{s;\om}=\frac{1}{2}A(G;s,\om)^{-1}G^{s;\om},
\ee
where $G^{s;\om}=[G^{s;\om}_{01},\ldots, G^{s;\om}_{0N}]'$ is a vector-column, and 
$A(G;s,\om)=[G^{s;\om}_{jk}]_{j,k=1,\ldots, N}$.

\subsection{Variance minimizing hedging portfolio}\label{euroVarMinport}
%Given the weight $w$,
The hedging error is the random variable
\[
Err(n;x,X_T)=e^{-rT}\left(-G_0(X_T\ |\ X_0=x)+\sum_{j=1}^N n_jG_j(X_T\ |\ X_0=x)\right).
\]
Set $V_j(T,x)=e^{-rT}\bE^x[G_j(X_T)]$  (the expectations are under an EMM $\bQ$ chosen for pricing).
Assuming that $\hG_j, j=0,1,\ldots, N$, are calculated, calculation of the mean hedging error
\beqast
\bE\left[Err(n;x,X_T)\right]
%&=&e^{-rT}\bE^x\left[G_0(X_T)-\sum_{j=1}^N w^{s;\om}_jG_{j}(X_T)\right]\\
&=&-V_0(T,x)+\sum_{j=1}^N n_jV_j(T,x)
\eqast
is reducible to the Fourier inversion. As it is explained in \cite{iFT,pitfalls,SINHregular}, if (1) $\hG_j$ is of the form
$\hG_j(\xi)=e^{-ik_j\xi}\hG_{j0}(\xi)$, where $k_j\in\bR$ and $\hG_{j0}$ is a rational functions, and (2)  $\psi$ is of the form
$\psi(\xi)=-i\mu\xi+\psi^0(\xi)$, where $\mu\in\bR$ and $\Re\psi^0(\xi)\to +\infty$ as $\xi\to \infty$ remaining in a cone around the real axis,
then  it is advantageous to represent $V_j(T,x)$ in the form 
\bbe\label{Vj}
V_j(T;x)=\frac{1}{2\pi}\int_{\Im\xi=\om}e^{ix'_j\xi-T(r+\psi^0(\xi))}\hG_{j0}(\xi)d\xi,
\ee
where the set of admissible $\om\in\bR$ depends on $\hG_{j0}$, and $x'_j=x+\mu T-k_j$. Then 
we use a conformal deformation of the contour of integration
in \eqref{Vj} and the corresponding change of variables, and apply the simplified trapezoid rule. The most efficient change of variables (the sinh-acceleration)
suggested in \cite{SINHregular} is of the form $\xi=i\om_1+b\sinh(i\om'+y)$, where $\om'$ is of the same sign as $x'$; the upper bound on admissible 
$|\om'|$ depends on $\psi^0$ and $\hG_{j0}$.
 
The variance can be calculated using the equality
\[
e^{-2rT}\bE[(G-\bE[G])^2]=e^{-2rT}(\bE[G^2]-\bE[G]^2).
\]
To calculate $\bE^x\left[Err(n;x, X_T)^2\right]$, we need to calculate 
$V_{j\ell}(T,x):=e^{-2rT}\bE^{x}\left[G_j(S_T)G_\ell(S_T)\right], j,\ell=0,1,\ldots, N$, which is the ``price" of the European option
with the payoff $G_j(S_T)G_k(S_T)$ at maturity $T$. We calculate the Fourier
transform $\widehat{G_jG_\ell}(\xi)$ of the product $G_jG_\ell$, multiply by the characteristic function $e^{-T\psi(\xi)}$,
and apply the inverse Fourier transform. For typical exotic options and vanillas,  $\widehat{G_jG_\ell}(\xi)$ is of the form 
$e^{-ik_{j\ell}\xi}\hG_{j\ell;0}(\xi)$, where $\hG_{j\ell;0}(\xi)$ is a rational function, and $k_{j\ell}\in\bR$. Hence,
\bbe\label{Vjell}
V_{j\ell}(T;X)=\frac{1}{2\pi}\int_{\Im\xi=\om}e^{ix'_{j\ell}\xi-T(2r+\psi^0(\xi))}\hG_{j\ell;0}(\xi)d\xi,
\ee
where $x'=x+\mu T- k_{j\ell}$. The integral on the RHS of \eqref{Vjell}
 can be calculated accurately and fast  using the sinh-acceleration
technique \cite{SINHregular}. The numbers of options in the variance minimizing portfolio are given by
\bbe\label{minvarwsom}
n(T;x)=\frac{1}{2}A(T;x)^{-1}V^0(T;x),
\ee
where $V^{0}(T;x)=([V_{0j}(T;x)-V_0(T,x)V_j(T,x)]_{j=1}^{N})'$ is a vector-column, and 
$A(T;x)=[V_{j\ell}(T;x)-V_{\ell}(T;x)V_{\ell}(T;x)]_{j,\ell=1,\ldots, N}$, is a matrix.

\subsection{An example: static hedging and variance minimizing hedging of the exotic 
option with the payoff function $\cG_{ex}$ given by \eqref{stHedgeEur}.
The case $\cG(S)=(S-K)_+$.}\label{example_ex_down_and_in}
Let $k:=\ln K>h:=\ln H$. Then $G(x)\bfo_{(-\infty,h]}(x)=0$, and $2h-k=h-(k-h)<h$. 
%In Section \ref{calcFTGex},
%we calculate
Direct calculations show that 
 the Fourier transform of $G^0(x)=\cG_{ex}(e^x)$ is well-defined in the half-plane $\{\Im\xi>1-\be\}$ by
 \bbe\label{hG0}
 \hG^{0}(\xi)
 %&=&\int_{-\infty}^{2h-k}e^{-ix\xi+\be(x-h)}(e^{2h-x}-e^{k})dx\\
% &=& e^{(2-\be)h}\frac{e^{(2h-k)(-i\xi+\be-1)}}{-i\xi+\be-1}-e^{k-\be h}\frac{e^{(2h-k)(-i\xi+\be)}}{-i\xi+\be}\\
 %&=&\frac{e^{\be(h-k) +k-i\xi(2h-k)}}{(-i\xi+\be-1)(-i\xi+\be)}=\frac{e^{(\be-2)(h-k) +(2h-k)(1-i\xi)}}{(-i\xi+\be-1)(-i\xi+\be)}\\
% &=&\left(\frac{H}{K}\right)^{\be-2}\left(\frac{H^2}{K}\right)^{1-i\xi}\frac{1}{(-i\xi+\be-1)(-i\xi+\be)}\\
 =H^{-\be} K^{1-\be}\frac{e^{-i\xi(2h-k)}}{(-i\xi+\be-1)(-i\xi+\be)}.
 \ee
 In the hedging portfolio, we use put options with strikes $K_j\le K_0=H^2/K$, $j=1,2,\ldots, N$. Set $G^j(x)=(K_j-e^x)_+$;
 then $\hG^j(\xi)=K_j^{1-i\xi}/(i\xi(i\xi-1))$ is well-defined in the half-plane $\{\Im\xi>0\}$. 
 
 \subsubsection{Construction of an approximate
 static hedging portfolio} We take $\om>(1-\be)_+$. Typically, $\be<0$, hence, $\om>1$. We have
 \beqast
 (G^0,G^0)_{s,\om}&=&H^{-2\be}K^{2-2\be}\int_{\bR}\frac{e^{(-i\xi+\om)(2h-k)}}{(-i\xi+\om+\be-1)(-i\xi+\om+\be)}
 \\&& \hskip2.2cm\times \frac{e^{(i\xi+\om)(2h-k)}}{(i\xi+\om+\be-1)(i\xi+\om+\be)}(1+\xi^2)^s d\xi\\
 &=& H^{4\om-2\be}K^{2(1-\be-\om)}\int_\bR\frac{(1+\xi^2)^sd\xi}{(\xi^2+(\om+\be-1)^2)(\xi^2+(\om+\be)^2)}.
 \eqast
 where $k=\ln K$. We can calculate the integral accurately and fast making the simplest sinh-change of variables
 $\xi=b\sinh y$, and applying the simplified trapezoid rule. See \cite{SINHregular} for explicit recommendations for
 the choice of $b$ and the parameters $\ze, N$ of the simplified trapezoid rule.
 
 Next, for $j=1,2,\ldots, N$, we set $k_j=\ln K_j$ and calculate
 \beqast
&& (G^0,G^j)_{s,\om}\\&=&H^{-\be}K^{1-\be}K_j\int_{\bR}\frac{e^{(-i\xi+\om)(2h-k)}}{(-i\xi+\om+\be-1)(-i\xi+\om+\be)}
% \\&&\hskip2.2cm\times 
\frac{e^{(i\xi+\om)k_j}}{(-i\xi-\om-1)(-i\xi-\om)}(1+\xi^2)^s d\xi\\
 &=& H^{2\om-\be}K^{1-\be-\om}K_j^{1+\om}\int_\bR\frac{e^{-i\xi(2h-k-k_j)}(1+\xi^2)^s d\xi}
 {(-i\xi+\om+\be-1)(-i\xi+\om+\be)(i\xi+\om+1)(i\xi+\om)}.
\eqast
 If $k_j= 2h-k$, we make the sinh-change of variables of the same form as above (the choice of the parameters $b$, $\ze$, $N$
 is slightly different). If $2h-k-k_j<0$ (resp., $>0$), then it is advantageous to deform the contour of integration so that the wings of the deformed
 contour point up (resp., down). Hence, we make the change of variables $\xi=i\om_1+b\sinh(i\om'+y)$, where $\om'\in (0,\pi/2)$
  (resp., $\om\in(-\pi/2,0)$) and set $\om_1=-b\sin\om'$.
 The parameters $b$, $\ze$, $N$ are chosen as explained in \cite{SINHregular}.

 Finally, for $j,\ell=1,2,\ldots, N$, and $\om>0$,
  \beqast
   (G^j,G^\ell)_{s,\om}&=&(K_jK_k)^{1+\om}\int_\bR \frac{e^{-i\xi(k_j-k_\ell)}(1+\xi^2)^s d\xi}{(\xi^2+\om^2)(\xi^2+(\om+1)^2)}.
   \eqast
   If $k_j>k_\ell$, we deform the wings of the contour down, equivalently, use the sinh-acceleration with $\om'\in (-\pi/2,0)$.
   If  $k_j<k_\ell$, we use $\om'\in (0,\pi/2)$. Finally, if $k_j=k_\ell$, we may use any $\om'\in [-\pi/2,\pi/2]$; the choice $\om'=0$
   is the best one.
   
   After the scalar products are calculated, we apply \eqref{statwsom} to find the approximate static hedging portfolio.
   Below, we will use a modification of this scheme when the hedging portfolio has
   the fixed amount $H^{-\be}K_0K_1^{\be-1}=(H/K_0)^{\be-2}$ of put options with strike $K_1$,
   and the weights of the other put options in the hedging portfolio are calculated minimizing the hedging error.
   
  \subsubsection{Construction of the variance-minimizing hedging portfolio}
  Calculating the integral  on the RHS
  of  \eqref{Vj}, we use \eqref{hG0} for $j=0$; for $j=1,2,\ldots, N$, 
  $\hG^j(\xi)=K_j^{1-i\xi}/(i\xi(i\xi-1)$. To calculate the integral on the RHS of  \eqref{Vjell}, we need to calculate the Fourier transforms
  of the products of the payoff functions.  The straightforward calculations give
  \begin{enumerate}[(1)]
  \item
  %In Section \ref{calcFTGex2}, we calculate 
  $\widehat{(G^0)^2}(\xi)=e^{-i\xi(2h-k)}\hG_{00}(\xi),
$ for $\xi$ in the half-plane $\{\Im\xi>2(1-\be)_+$, where %the rational function $\hG_{00}$ is given by \eqref{hG00};
 \bbe\label{hG00}
 \hG_{00}(\xi)=\frac{2H^{2\be}K^{2(1-\be)}}{(-i\xi+2\be-2)(-i\xi+2\be-1)(-i\xi+2\be)};
 \ee
  \item
%In Section \ref{hatcalcCexCj}, we prove that 
for $j=1,2,\ldots, N$ and $\xi$ in the half-plane $\Im\xi>(1-\be)_+$,
 \[
\widehat{G^0G^j}(\xi)=e^{-ik_j\xi}\hG_{0j}(\xi),
\]
where %the rational function $\hG_{0j}(\xi)$ is given by \eqref{hG0j};
\bbe\label{hG0j}
\hG_{0j}(\xi)=\frac{(K_j/H)^\be (H^2/K)}{-i\xi+\be}\left[\frac{H^2/K}{-i\xi+\be-1}-\frac{K_j}{-i\xi+\be+1}\right].
\ee
  \item
  %In Section \ref{Gprodput}, we prove that 
  if 
$K_j\le K_\ell$, then the Fourier transform of $G^{jk}(x)=(K_j-e^x)_+(K_\ell-e^x)_+$ 
\bbe\label{prodPuts}
\hG^{j\ell}(\xi)=K_j\frac{e^{-ik_j\xi}}{i\xi-1}\left(\frac{K_\ell}{i\xi}-\frac{K_j}{i\xi-2}\right)
\ee
is well-defined in the half-plane $\{\Im\xi>0\}$. 
  \end{enumerate}
  
  \subsubsection{Numerical experiments}
In Tables collected in Section \ref{Eurotables}, we study the relative performance
of the static hedging and variance-minimizing hedging.  We consider the exotic European option with the payoff $\cG_{ex}=(S/H)^\be(H^2/S-K_0)_+$;
the hedging portfolios consist of put options with the strikes $K_j=H^2/K_0-(j-1)0.02, j=1,\ldots, \#K$, where $\#K=3,5$.
For different variants
of hedging, we list numbers $n_j$ of the options with strikes $K_j$ in the hedging portfolio.
Static portfolios are constructed minimizing the hedging error in the $H^{s;\om}$ norm;
  the results are
 essentially the same for $\om=2(1-\be)_++0.1, \om=2(1-\be)_++0.2$, and weakly depend on $s=0.5, 0.55$.
 The static portfolios are independent of time to maturity $T$ and the process but $nStd$ does depend on both as well
 as on the spot $S$.
We study the dependence of the (normalized 
by the price $V_{ex}(T,x)$ of the exotic option) standard deviations $nStd$ of the static hedging portfolio
and variance minimizing portfolio on the process, time to maturity $T$ and 
$x:=\log(S/K_1)\in [-0.03,0.03]$.
For the static hedging portfolios, for each process and time to maturity,
 we show the range of $nStd$ as the function of $x\in [-0.03,0.03]$.
 In the case of the variance-minimizing hedging, $n_j$ 
 depend on $x$ by construction, and we show $n_j$ and $nStd$ for each $x$ in a table.
 We consider two variants of the variance minimizing portfolios:  $VM1$ means that $n_1$ (the same as for static portfolios)
 is fixed, $VM2$ means that all $n_j$ may vary.
 
 In all cases, $H=1$, $K_0=1.02$,
the underlying $S_t=e^{X_t}$ pays no dividends, $X$ is KoBoL; in two tables,
the BM with the embedded KoBoL component.  
The second instantaneous moment is $m_2=0.1$ or $0.15$, and $c$  is determined by $m_2, \lp,\lm$ and $\sg$. 
The riskless rate is found from the EMM condition
$\psi(-i)+r=0$;  $\be$  found from $\psi(\xi)=\psi(-\xi-i\be)$. If $X$ is KoBoL, then $\be$ exists 
only if 
$\mu=0$, and then $\be=-\lm-\lp$. If the BM component is present, then 
we choose $\sg$ and $\mu$ so that $\be=-\lp-\lm=-\mu/(2\sg^2)$.  The results for two cases when such a $\be$ exists 
are presented
in Tables~\ref{StVar1_1.2slow}-\ref{BMKBL_1.95}; in Tables~\ref{StVar1_1.95} and \ref{BMKBL_1.95}, KoBoL is close to BM ($\nu=1.95$), and in Tables~\ref{StVar1_1.2slow}, \ref{StVar1_1.2}, \ref{BMKBL_1.2},
 close to NIG ($\nu=1.2$).
The BM component is non-trivial in Tables~\ref{BMKBL_1.2} and \ref{BMKBL_1.95}.
The reader may notice that the parameter sets are not very natural; the reason is that it is rather difficult to
find natural parameter sets which ensure that (1) $\be$ satisfying $\psi(\xi)=\psi(-\xi-i\be), \forall\ \xi\in\bR,$ exists;
(2) the EMM condition holds. In Tables~\ref{AsymKBL_1.2} and \ref{AsymKBL_1.95}, we present results for a sizably asymmetric KoBoL with $\mu\neq 0$, and consider the same exotic option with $\be=-\lm-\lp$. Naturally, this exotic option cannot be even formally used to hedge the barrier option
but can serve for the purpose of the comparison of two variants of hedging of exotic European options.

Tables illustrate the following general observations.
\begin{enumerate}[(1)]
\item
If the number of vanillas in a static hedging portfolio is sufficiently large,
the portfolio  provides uniform (approximate) hedging over wide stretches of spots and strikes.
Hence, if the jump density decays slowly, one expects that,  far from the spot, the static hedging portfolio will
outperform the variance minimizing portfolio. Table~\ref{StVar1_1.2slow} demonstrate that even in cases
when the rate of decay of the jump density is only moderately small, and the process is not very far from the BM,
the variance of the static portfolio differs from the variance of the variance minimizing portfolios (constructed
separately for each spot from a moderate range, and using the information about the characteristics of the process)
by several percent only; if the jump density decays slower and/or process is farther from the BM,
the relative difference is smaller.  Hence, if the rate of decay of the jump density is
not large and the density os approximately symmetric, the static portfolio is competitive for hedging risks of small fluctuations.
  It is clear that the hedging performance of the static portfolio in the tails must be better still.
\item
However, if the jump density decays moderately fast, then the variance of the static portfolio can be sizably larger than
the one of the variance minimizing portfolios (Tables~\ref{StVar1_1.2}-\ref{StVar1_1.95}), and
if a moderate BM component is added, then the relative difference becomes large (Tables~\ref{BMKBL_1.2}-\ref{BMKBL_1.95}).
\item
In Tables~\ref{StVar1_1.2slow}-\ref{BMKBL_1.95}, the jump density is approximately symmetric. In 
Tables~\ref{AsymKBL_1.2}-\ref{AsymKBL_1.95}, we consider the pure jump process with a
moderately asymmetric density of jumps. In this case, the variance of the static portfolio is much larger than
the variance of variance minimizing portfolios.
\item
The quality of variance minimizing portfolios VM1 and VM2 is essentially the same in almost all cases
when 5 vanillas are used although the portfolio weights can be rather different. Hence,
as a rule of thumb, one can recommend to use vanillas associated with the atomic part of the measure
in the integral representation of the ideal static portfolio - provided these vanillas are available in the market.
\end{enumerate}
The implication of observations (1)-(2) for semi-static hedging of barrier options is as follows. If the 
the variance of the BM component makes a non-negligible contribution to the instantaneous variance 
of the process,
 the ideal semi-static hedging using a continuum of options improves
but the quality of an approximation of the integral of options by a finite sum decreases. Hence,
one should expect that the variance minimizing hedging of barrier options would be significantly better
than an approximate semi-static hedging in all cases.

%\subsection{Calculation of expectations of products of payoffs at maturity, in L\'evy models}
\section{Hedging down-and-in and down-and-out options}\label{hedgingdownandin}

\subsection{Semi-static hedging}\label{ss:semi-static}
Carr and Lee \cite{CarrLee09}  formulate several equivalent conditions on a positive martingale $M$ under the reference measure 
$\bP$, call the class of these martingales PCS processes, and design semi-static replication strategies for various types of barrier options. 
Since the proof of Theorem 5.10 in \cite{CarrLee09} strongly relies on the assumption
that at the random time $\tau$ when the barrier $H$ is breached, the underlying is exactly at the barrier,
in Remark 5.11 in \cite{CarrLee09}, the authors
state that these strategies replicate the options in question for all PCS processes,  including those with jumps, provided that the jumps cannot cross the barrier.
However, one of the equivalent conditions which define PCS processes is
the equality of the distributions of $M_T/M_0$ 
under $\bP$ and $M_0/M_T$ under $\bM$, where $d\bM/d\bP= M_T/M$. This condition implies that $M$ has positive jumps if and only if $M$ has negative jumps.
Hence, if jumps in $M$ does not cross the barrier, equivalently, there are no jumps in the direction of the barrier, then there are no jumps in the opposite direction,
and, therefore, $M$ has no jumps. Further relaxing PCS conditions, Carr and Lee \cite{CarrLee09}  generalize to various asymmetric dynamics, 
but the property that jumps in one direction are impossible
means that the results for semi-static replication of barrier options under these asymmetric dynamics are valid only if there are no jumps.
 Carr and Lee \cite{CarrLee09} give additional conditions which will ensure the super-replication property
of the semi-static portfolio for ``in" options; but the corresponding portfolio for ``out" options  recommended in \cite{CarrLee09} will
under-replicate the option. 
%To summarize, 
%\footnote{In the case of an exponential L\'evy model with the characteristic exponent $\psi$,
%the condition is: $\exists \be\in \bR$ s.t. $\psi(\xi)=\psi(-\xi-i\be)$, which makes the scope of (approximate) hedging rather small.}.

For an additional clarification of these issues, in Section \ref{SemiStaticLevy}, we derive the generalized symmetry condition
for the case of a L\'evy process $X$ in terms of the characteristic function $\psi$: $\exists \be \in\bR$ s.t. 
$\psi(\xi)=\psi(-\xi-i\be)$ for all $\xi$ in the domain of $\psi$, and show that
this condition implies that either $X$ is the Brownian motion (BM) and the riskless rate equals the dividend rate or there are jumps
in both directions, and asymmetry of the jump component is uniquely defined by the volatility $\sg^2$ and drift $\mu$. Furthermore, 
if $\sg=0$, then $\mu=0$ as well. For the case of the down-and-in option with the payoff $G(X_T)=\cG(e^{X_T})$ and barrier $H$, we rederive 
the formula for the payoff of the  exotic European option, which, in the presence of jumps,
replicates the barrier option only approximately:
\bbe\label{stHedgeEur}
\cG_{ex}(S_T)=( \cG(S_T)+(S/H)^\be \cG(H^2/S))\bfo_{S_T\le H},
\ee
The numerical examples above show that the variance of the hedging error of the static portfolio for the exotic option
with the payoff \eqref{stHedgeEur} is close
to the variance of the variance-minimizing portfolio if the BM component is 0, the jump density does not decrease fast and is approximately symmetric;
if the BM component is sizable and/or the density of jumps is either asymmetric or fast decaying,
then the variance of the static portfolio is significantly larger than the variance of the variance-minimizing portfolio.
Hence, the static portfolio is a good (even best) choice in cases when the idealized semi-static replicated exotic option is
a bad approximation fo the barrier option.

\subsection{General scheme of variance minimizing hedging}\label{down-and-in_or_out: general scheme}
We consider one-factor models. The underlying is $e^{X_t}$, there is no dividends, and the riskless rate $r$ is constant.
Let $V^0(t,X_t)$ be the price  of the contingent claim to be hedged, of maturity $T$, under an EMM $\bQ$ chosen for pricing.
Let $V^j(t,X_t)$, $j=1,2,\ldots, N$, be the prices  of the options used for hedging. We assume that the latter options do not expire
before $\tau\wedge T$, where $\tau=\tau^-_h$ is the first entrance time into the activation region $U$ of the down-and-in option
(in the early expiry region of the down-and-out option). 
As in the papers on the semi-static hedging, we assume that, at time $\tau\wedge T$, the hedging portfolio is liquidated. 
 Let $(-1,n_1,\ldots, n_N)$ be the vector of numbers of securities in the hedging portfolio. The portfolio at
 the liquidation date $\tau\wedge T$ is the random variable
 \[
 P_0(\tau\wedge T, X_{\tau\wedge T})=-V^0(\tau\wedge T, X_{\tau\wedge T})+\sum_{j=1}^N n_j V^j(\tau\wedge T, X_{\tau\wedge T}),
 \]
and the discounted portfolio at the liquidation date is $P(\tau\wedge T, X_{\tau\wedge T})=e^{-r\tau\wedge T}P_0(\tau\wedge T, X_{\tau\wedge T})$.
One can consider the variance minimization problem for either $P_0(\tau\wedge T, X_{\tau\wedge T})$ or $P(\tau\wedge T, X_{\tau\wedge T})$,
and we can calculate the variance under either the EMM $\bQ$ used for pricing or the historic measure $\bP$.
 %Still another version is the minimization of
%\[
%Var(w):=\bE\left[e^{-r\tau\wedge T}\left(P_0(\tau\wedge T, X_{\tau\wedge T})-\bE\left[P_0(\tau\wedge T, X_{\tau\wedge T})\right]\right)^2\right]
%\]
%(under $\bQ$ or $\bP$). 
We consider the minimization of the variance of $P(\tau\wedge T, X_{\tau\wedge T})$ under $\bQ$.
%, and, later,  indicate changes needed to consider
%the other versions of the variance minimization problem.  

Let $x=X_0$, and, for $j,k=0,1,\ldots, N$, set $C^0_{j\ell}(x)=V^j(0,x)V^\ell(0,x)$ and
\bbe\label{Cjell}
C_{j\ell}(x)=\bE^x\left[e^{-2r\tau\wedge T}V^j(\tau\wedge T, X_{\tau\wedge T})V^\ell(\tau\wedge T, X_{\tau\wedge T})\right].
\ee
Using $\bE[(U-\bE[U])^2]=\bE[U^2]-\bE[U]^2$, we represent the variance of $P(\tau\wedge T, X_{\tau\wedge T})$ in the form
\beqa\label{defVar}
Var_P(x)&=&C_{00}(x)-C^0_{00}(x)-2n_j\sum_{j=1}^N (C_{0j}(x)-C^0_{0j}(x))\\\nonumber
&&+\sum_{j=1}^Nw_j^2 (C_{jj}(x)-C^0_{jj}(x))+\sum_{j\neq \ell, j,\ell=1,2,\ldots, N}n_jn_\ell (C_{j\ell}(x)-C^0_{j\ell}(x)),
\eqa
and find the minimizing $n=n(x)$ as
\bbe\label{wmin}
n(x)=A(x)^{-1}B(x),
\ee
where
$B(x)=[C_{0j}(x)-C^0_{0j}(x)]'_{j=1,\ldots, N}$ is a column vector, and $A=[C_{j\ell}(x)-C^0_{j\ell}(x)]_{j,\ell=1,\ldots, N}$
is an invertible square matrix if the random variables $V_j(\tau\wedge T, X_{\tau\wedge T}), j=1,2,\ldots, N,$ are uncorrelated.

To calculate $C^0_{j\ell}(x)$, it suffices to calculate $V^j(0,x)$ and $V^\ell(0,x)$. We  decompose 
$V^j(0,x)$ into the sum of the first-touch option $V^j_{ft}(x)$ with the payoff
$V^j(\tau,X_\tau)\bfo_{\tau\le T}$, and no-touch option $V^j_{nt}(x)$ with the payoff
$V^j(T,X_T)\bfo_{\tau> T}$. Given a model for $X$, we can calculate the prices of the no-touch and first touch options. 
 Similarly, we
can decompose $C_{j\ell}(x)$ into the sum of the first-touch option $C_{ft;j\ell}(x)$ with the payoff
$V^j(\tau,X_\tau)V^\ell(\tau,X_\tau)\bfo_{\tau\le T}$, and no-touch option $C_{nt;j\ell}(x)$ with the payoff
$V^j(T,X_T)V^\ell(T,X_T)\bfo_{\tau> T}$. 

The no-touch options can be efficiently calculated if
the Fourier transforms of the payoff functions $G^j$ and $G^jG^\ell$ of options $V^j$ and $V^\ell$ can be explicitly calculated;
then several methods based of the Fourier inversion can be applied.

%In the present paper, we apply three methods. 
First, one can apply 
apply Carr's randomization method developed in \cite{NG-MBS,single,MSdouble,BLdouble} for option pricing in L\'evy models.
%, regime-switching L\'evy
%models and stochastic volatility models and models with stochastic interest rates. 
A simple generalization is necessary 
in the case of no-touch options because, in the setting of the present paper, the payoff functions depend on $(t,X_t)$ and not on
$X_t$ only as in  \cite{NG-MBS,single,MSdouble, BLdouble}.  Apart from the calculation 
of the Wiener-Hopf factors, which, for a general L\'evy process, must be done
in the dual space,
%\footnote{In Section \ref{WHFSinh}, we design new efficient procedures for the calculation
%of the Wiener-Hopf factors.}, 
the rest of calculations in \cite{NG-MBS,single,MSdouble,BLdouble} are made in the state space. 
The calculations in the state space \cite{single,MSdouble,BLdouble}
  can be efficient for small $\dt$ and $\nu$ if the tails of the L\'evy density decay sufficiently fast. 
  In Section \ref{WHFSinh}, we design new efficient procedures for the calculation
of the Wiener-Hopf factors. % which can be applied wherever the Wiener-Hopf factors are needed.
These procedures are of a general interest. 
  
The second method is 
a more efficient version of Carr's randomization. The
calculations at each step of the backward procedure bar the last one are made remaining in the dual space. These steps are of the same form
 as in the Hilbert transform method \cite{feng-linetsky08} for barrier options with discrete monitoring, with a different operator used at each time step.
 In \cite{feng-linetsky08}, the operator is $(1-e^{-r\dt})^{1}(I-e^{\dt(-r+L)}),$ where $L$ is the infinitesimal generator of $X$ under the probability measure used;
  in the Carr's randomization setting, the operator is $q^{-1}(q-L)$, where $q=r+1/\dt$, and $\dt$ is the time step. 
  If $\dt$ is not small and/or the order of the process $\nu$ is close to two\footnote{Recall
  for stable L\'evy processes, the order of the process as defined in
  \cite{NG-MBS} for processes with exponentially decaying jump densities,  is the Blumenthal-Getoor index},
  then $(1-e^{-r\dt})^{1}(I-e^{\dt(-r+L)})$ can be efficiently realized using the fast Hilbert transform \cite{feng-linetsky08}.
  Otherwise too long grids may be necessary. 
 An efficient numerical realization of  $q^{-1}(q-L)$ in the dual space requires much longer grids than an efficient numerical
 realization of $(1-e^{-r\dt})^{1}(I-e^{\dt(-r+L)})$ if the fast Hilbert transform is used. In the result, this straightforward scheme can be
 very inefficient. Instead, we can apply the Double Spiral method \cite{AsianGammaSIAMFM} calculating the Fourier transform of the option price at two contours, at each time step. In \cite{AsianGammaSIAMFM} discretly sampled Asian options were considered, and
 a complicated structure of functions arising at each step of the backward induction procedure required the usage of the flat contours of
 integration. In application to barrier options, the contours in the Double Spiral method can be efficiently deformed and 
  an efficient sinh-acceleration technique developed in \cite{SINHregular} applied.
 Namely, we can use the change of variables of the form $\xi=i\om_1+b\sinh(i\om+y)$ and the simplified trapezoid rule. In the result,
 an accurate numerical calculation of integrals at each time step needs summation of 2-3 dozen of terms in the simplified trapezoid rule.
 We leave the design of explicit procedures for hedging using both versions of Carr's randomization to the future. 
 
 In the present paper, 
 %The last method that 
 we directly apply the general formulas for the double Fourier/Laplace inversion. These formulas
 are the same as the ones in \cite{paraLaplace} in the case of no-touch options, with the following improvement:
  instead of the fractional-parabolic changes of variables, the sinh-acceleration is used.  
 In the case of the first touch options, an additional generalization is needed
 because, contrary to the cases considered in \cite{paraLaplace}, 
 the payoff depends on $(t, X_t)$ rather than on $X_t$ only.  
 %The additional integration (calculation of the Fourier transform of the products $V^j(t,x)V^\ell(t,x)$ w.r.t. $x$) causes an 
 %additional complication; the sinh-acceleration allows one
 %to make calculations fairly fast in this more complicated case as well.

One can use other methods that use approximations in the state space. Any such method has several sources of errors,
which are not easy to control. Even in the case of pricing European and barrier options, serious errors may result (see 
\cite{iFT,one-sidedCDS,pitfalls} for examples),
and, typically, very long and fine grids in the state space are needed. 
The recommendations in \cite{kirkby14a,kirkby16,kirkby15SIAMFM} for the choice of the truncation parameter 
   rely on the ad-hoc recommendation for the truncation parameter used in  a series of papers
 \cite{COS,COS2,COS3}. As examples in \cite{iFT,MarcoDiscBarr}
demonstrate, this ad-hoc recommendation  can be unreliable even if applied only once.
In the hedging framework suggested in the present paper, the truncation needs to be applied many times,
for each $t_j$ used in the time-discretization of the initial problem,
hence, the error control becomes almost impossible.

\subsection{Conditions on processes and payoff functions}\label{condXGNt}
We consider the down-and-out case; $H=e^h$ is the barrier, $T$ is the maturity date,
and $G(X_T)$ is the payoff at maturity. The most efficient realizations of the pricing/hedging formulas
are possible if the characteristic exponent admits analytic continuation to a union of a strip and cone
and behaves sufficiently regularly at infinity. For the general definition of the corresponding class of
L\'evy processes (called {\em SINH-regular}) and applications to pricing European options
in L\'evy models and affine models, see \cite{SINHregular}. In the present paper, for simplicity, we assume
that the characteristic exponent admits analytic continuation to the complex plane with two cuts.

\vskip0.1cm\noindent
{\em Assumption  $(X)$}. \begin{enumerate}[1.]
\item $X$ is a L\'evy process with the characteristic exponent 
$\psi$ admitting analytic continuation
to the complex plane with the cuts $i(-\infty, \lm]$, $i[\lp,+\infty)$;
\item
$\psi$ admits the representation $\psi(\xi)=-i\mu\xi+\psi^0(\xi)$, where $\mu\in\bR$, and $\psi^0(\xi)$
has the following asymptotics as $\xi\to \infty$: for any 
$\varphi\in (-\pi/2, \pi/2)$:
\bbe\label{aspsiinf}
\psi^0(\rho e^{i\varphi})=c^0_\infty e^{i\nu\varphi}\rho^\nu+O(\rho^{\nu_1}),\ \rho\to+\infty,
\ee
where $\nu\in (0,2]$, $\nu_1<\nu$ are independent of $\varphi$, and $c^0_\infty>0$. 
%is a continuous function satisfying
Then
\bbe\label{cinf}
\Re c^0_\infty e^{i\nu\varphi}>0, \quad |\varphi|<\pi/(2\nu).
\ee
\end{enumerate}
The condition \eqref{cinf} implies that if $\nu\in (0,1)$, then $|\varphi|\ge \pi/2$ can be admissible. For standard classes of L\'evy processes
used in finance, this is  possible if we consider analytic continuation to an appropriate Riemann surface. See \cite{iFT,paired} for details.
We will not use analytic continuation to Riemann surfaces in the present paper.
\begin{example}{\rm
%\subsection{Examples and some generalizations}\label{genexSINH}
\begin{enumerate}[(1)]
\item
 Essentially all L\'evy processes
used in quantitative finance are elliptic SINH-regular L\'evy processes:
 Brownian motion (BM); Merton model \cite{merton-model}; NIG (normal inverse Gaussian model) \cite{B-N}; hyperbolic
 processes \cite{EK}; double-exponential jump-diffusion model \cite{lipton-risk,lipton-columbia,kou,KW1,KW2};
 its generalization: hyper-exponential jump-diffusion model,
  introduced in \cite{ amer-put-levy-maphysto,lipton-risk} and studied in detail in  \cite{amer-put-levy-maphysto, amer-put-levy};
 the majority of processes of the $\beta$-class \cite{beta}; the generalized
 Koponen's family \cite{KoBoL} and its subclass KoBoL
 \cite{NG-MBS}. A subclass of KoBoL
(known as the CGMY model - see \cite{CGMY}) is given by the characteristic exponent
\begin{equation}\label{kbl}
\psi(\xi)=-i\mu\xi+c\Gamma(-\nu)[\lp^\nu-(\lp+i\xi)^\nu+(-\lm)^\nu-(-\lm-i\xi)^\nu],
\end{equation}
where $\nu\in (0,2), \nu\neq 1$ (in the case $\nu=1$, the analytical expression is different: see \cite{KoBoL,NG-MBS}).
%Thus, KoBoL is SINH-regular of type $((\lm,\lp); \cC, \cC_+)$ and order $\nu$, where $\cC=\cC_{-\ga,\ga}$, $\ga\ge \pi/2$, and
%$\cC_+=\cC_{-\ga', \ga'}$, where  $\ga'=\pi/(2\nu)$ (see (3) below for the meaning of $\ga'>\pi/2$).
% BM, DEJD and HEJD are of order $\nu=2$, and NIG is of order $\nu=1$.

 The characteristic exponents of NTS processes  constructed in \cite{B-N-L} are given by
 \begin{equation}\label{NTS2}
\psi(\xi)=-i\mu\xi+\de[(\al^2+(\xi+i\be)^2)^{\nu/2}-(\al^2-\be^2)^{\nu/2}],
\end{equation}
where $\nu\in (0,2)$, $\de>0$, $|\be|<\al$.
% This is a process of type $((\al+\be,\al-\be); \cC, \cC_+)$ of order $\nu$,
%where $\cC$ and $\cC_+$ are the same as for KoBoL of the same order.

%\item
% In order to consider Variance Gamma processes (VG) \cite{MM91}, Definitions \ref{defSINHLevy}-\ref{defSINHLevyell} must be generalized
%replacing the function $\rho\mapsto \rho^\nu$ with a strictly increasing function
%$w:\bR_{+}\to\bR_{+}$ satisfying $w(+\infty)=+\infty$. We say: $X$ is an (elliptic) SINH-regular L\'evy process of type
%$(S,\cC,\cC_+, w)$. For Variance Gamma processes, $w(\rho)=\ln(1+\rho)$.

\item
  For KoBoL, VG  and NTS, $\psi^0$ admits analytic
continuation to an appropriate Riemann surface.
%, and $\cC$ can be defined as an appropriate subset of $\cR$.
%Formally, in (1), $\ga>\pi/2$ is admissible with the understanding that seemingly overlapping parts of $\cC_{-\ga,\ga}$ lie on
%different sheets of $\cR$.
%See \cite{iFT,paraLaplace,paired},
%where advantages of $\cC_+\subset \cR$ were utilized to increase the speed. 
This extension can be useful
when the SINH-acceleration is applied to calculate the Wiener-Hopf factors, and less so for pricing European options.

\item
The asymptotic coefficient $c_\infty(\mathrm{arg}\,\xi_0)$ is
 \begin{enumerate}[(i)]
 \item
 if $X$ is BM, DEJD and HEJD, $c^0_\infty=\sg^2/2$;%, hence, $\cC_+=\cC_{-\pi/4,\pi/4}$;
 \item
  if $X$ is given by \eqref{kbl}, $c_\infty^0=-2c\Ga(-\nu)\cos(\pi\nu/2)$;
  %then, for $\varphi\in [-\pi/(2\nu),\pi/(2\nu)]$,
 %\bbe\label{asKoBoL}
 %c_\infty(\varphi)=-2c\Ga(-\nu)\cos(\pi\nu/2)e^{i\nu\varphi},
 %\ee
 %hence, $\cC_+=\cC_{-\pi/(2\nu),\pi/(2\nu)}$;
 \item
  if $X$ is given by \eqref{NTS2}, $c_\infty^0=\de$.
  %then, for $\varphi\in [-\pi/(2\nu),\pi/(2\nu)]$,
 % \bbe\label{asNTS}
% c_\infty(i\varphi)=\de e^{i\nu\varphi};
% \ee
% hence, $\cC_+=\cC_{-\pi/(2\nu),\pi/(2\nu)}$.
\end{enumerate}
\item In \cite{KoBoL}, we constructed more general classes of L\'evy processes, with the characteristic exponents of the form
\begin{equation}\label{kblgen}
\psi(\xi)=-i\mu\xi+c_+\Gamma(-\nu_+)[\lp^{\nu_+}-(\lp+i\xi)^{\nu_+}]+c_-\Gamma(-\nu_-)[(-\lm)^{\nu_-}-(-\lm-i\xi)^{\nu_-}],
\end{equation}
where $c_\pm\ge 0$, $c_++c_->0$, $\lm<0<\lp, \nu_\pm\in (0,2), \nu_\pm\neq 1$, with  modifications in the case $\nu_+=1$ and/or $\nu_-=1$.
For these processes, the domains of analyticity and bounds are more involved. In particular, in general, the coni are not symmetric w.r.t. the real axis.
\end{enumerate}

Note that in the strongly asymmetric version of KoBoL, in the spectrally one-sided case in particular,
the condition $c^0_\infty>0$ does not hold, and should be replaced with $\mathrm{arg}\, c^0_\infty\in (-\pi/2,\pi/2)$.
}
\end{example}
\vskip0.1cm\noindent
{\em Assumption $(G)$}. \begin{enumerate}[1.]
\item $G$ is a measurable function admitting a bound 
\bbe\label{Gbebound}
|G(x)|\le C(1+e^{\be x}),
\ee
where $C>0$ and $\be\in [0,-\lm)$ are independent of $x\in\bR$;
\item
either (i) $G(x)=e^{\be x}$, $x\in\bR$, or
(ii) the Fourier transform of $G$ 
\[
\hG(\xi)=\int_\bR e^{-ix\xi}G(x)dx
\]
is well-defined on the half-space $\{\Im\xi<-\be\}$, and admits the representation $\hG(\xi)=e^{-ia\xi}\hG_0(\xi)$,
where $a\in \bR$ and $\hG_0(\xi)$ is a rational function decaying at infinity.
\end{enumerate}
Note that only the values $G(x), x>h$,  matter, hence, we may replace $G$ with $\bfo_{(h,+\infty)}G$, and then
there is no need to consider the case (i) separately.

%For $\om'<\om$, denote by $\cZ(\hG_0; \om',\om)$ the set of poles of $\hG_0$ in the strip $S_{(\om',\om)}:=\{\Im\xi\in (\om',\om)\}$.
\begin{example} {\rm
\begin{enumerate}[(a)]
\item
$G=\bfo_{(-\infty,h]}$: the payoff function of the no-touch option and the square of this payoff;
\item
$G(x)=e^x$: the value of the underlying at the maturity date $T$ and $X_T=x$, and the product of the latter
and the payoff of the no-touch digital;
\item
$G(x)=e^{2x}$: the square of the underlying at the maturity date $T$ and $X_T=x$;
\item
$G(x)=(e^x-K)_+$: the payoff function of the call option. $\hG(\xi)=e^{-ik\xi}\hG_0(\xi)$, where $k=\ln K$, and
$\hG_0=K/(i\xi(i\xi-1)$ has two simple poles at 0 and $-i$.
%\item
%$G(x)=((e^x-K)_+)^2$: the square of the payoff function of the call option (needed to calculate the variance
%of the hedging portfolio) and the payoff function of the power option;
\item
$G(x)=(e^x-K_j)_+(e^x-K_\ell)_+$: the product of the payoff functions of two call options. With $K_j=K_\ell$,
we have the payoff function of a powered call option. $\hG_0$ has three simple poles at 0, $-i$ and $-2i$.
 If $K_j\ge K_\ell$, $\hG(\xi)$ is given by \eqref{prodPuts}.
\item
$G(x)=(K-e^x)_+$: the payoff function of the put option; $\hG(\xi)=e^{-ik\xi}\hG_0(\xi)$, where $k=\ln K$,
$\hG_0=K/(i\xi(i\xi-1)$ has two simple poles at 0 and $-i$.
\item
$G(x)=(e^x-K_j)_+(K_\ell-e^x)_+$: the product of the payoff functions of the call and put options. With $K_j=K_\ell$,
we have the payoff function of a powered put option. $\hG_0$ has three simple poles at 0, $-i$ and $-2i$.
 If $K_j\le K_\ell$, $\hG(\xi)$ is given by \eqref{prodPuts}.
\end{enumerate}
}
\end{example}
\begin{rem}{\rm 
%\begin{enumerate}[(i)]
%\item
 In (a) and (f) (after the multiplication by $\bfo_{(h,+\infty)}$), $G$ satisfies condition 2 with $\be=0$,
in (b) and (d) - with $\be=1$, and in the other cases - with $\be=2$.
%\item We will calculate the expectation $V(G;T,x)=\bE^x\left[e^{-rT}G(X_T)\bfo_{\uX_T>h}\right]$
%letting $r=0$.  If the expectation of a discounted payoff (or the product of two discounted payoffs) is calculated,
%the result must be multiplied by $e^{-rT}$ (resp., by $e^{-2rT}$).
%\item
%By linearity, we can reduce the calculations to cases $G(x)=e^{\be x}$ and  $G(x)=e^{\ga x}(e^x-K)_+$.
%However, if $\hG$ can be explicitly calculated and shown to satisfy condition (ii), then the separation of the payoff
%in a sum of simpler functions is inefficient.
%\end{enumerate}
}
\end{rem}

\subsection{More general payoff functions and embedded options}\label{embedd}
In the case of embedded options, the simple structure of $\hG$ formalized in Assumption $(G)$ is impossible.
The following generalization allows us to consider payoffs that are prices of vanillas and some exotic options.
For $\ga\in (0,\pi/2)$, define cones in the complex plane $\cC^+_\ga=\{z\in \bC\ |\ \mathrm{arg}\,z\in (-\ga,\ga)\}$,
$\cC_\ga=\cC^+_\ga\cup(-\cC^+_\ga)$. For $\mum<\mup$, set $S_{(\mum,\mup)}:=\{\xi\ |\ \Im\xi\in (\mum,\mup)\}$.
\vskip0.1cm\noindent
{\em Assumption $(G_{emb})$}. \begin{enumerate}[1.]
\item $G$ is a measurable function admitting the bound \eqref{Gbebound};
\item
there exist $a\in \bR$, $\de>0$ and $\ga\in (0,\pi/2)$ such that 
$\hG(\xi)=e^{-ia\xi}\hG_0(\xi)$, where  $\hG_0(\xi)$ is meromorphic in
$S_{(\lm,\lp)}\cup \cC_\ga$, has a finite number of poles and decays as  $|\xi|^{-\de}$ or faster as $\xi\to\infty$
remaining in $S_{(\lm,\lp)}\cup \cC_\ga$.
\end{enumerate}

%\subsection{General pricing formulas}\label{genpricingNT}

\section{Wiener-Hopf factorization}\label{s:WHF}
\subsection{Wiener-Hopf factorization: basic facts used and derived in \cite{barrier-RLPE,NG-MBS,single,paraLaplace,paired}}
Several equivalent versions of general pricing formulas for no-touch and first touch options were derived in \cite{barrier-RLPE,NG-MBS,single,paraLaplace,paired},
in terms of the Wiener-Hopf factors. In this subsection, we list the notation and facts which we use in the present paper. 
\subsubsection{Three forms of the Wiener-Hopf factorization}\label{3forms}
Let $X$ be a L\'evy process with the characteristic exponent $\psi$. 
The supremum and infimum process are defined by $\barX_t=\sup_{0\le s\le t}X_s$ and $\uX_t=\inf_{0\le s\le t}X_s$,
respectively. Let $q>0$ and let $T_q$ be an exponentially distributed random variable of mean $1/q$, independent of $X$.
Introduce functions 
\begin{equation}\label{e:def-phi-pm}
\phi^+_q(\xi)=\bE\left[ e^{i\xi\barX_{T_q}} \right], \qquad
\phi^-_q(\xi)=\bE\left[ e^{i\xi\uX_{T_q}} \right],
\end{equation}
and normalized resolvents (the {\em EPV operators} under $X$, $\barX$ and $\uX$, respectively) 
\beqa\label{cEq}
\cEq u(x)&=&\bE\left[u(X_{T_q})\right]=\bE^x\left[\int_0^{+\infty} qe^{-qt}u(X_t)dt\right],\\\label{cEpq}
\cEpq u(x)&=&\bE\left[u(\barX_{T_q})\right]=\bE^x\left[\int_0^{+\infty} qe^{-qt}u(\barX_t)dt\right],\\\label{cEmq}
\cEmq u(x)&=&\bE\left[u(\barX_{T_q})\right]=\bE^x\left[\int_0^{+\infty} qe^{-qt}u(\uX_t)dt\right].
\eqa
%Assume that $\psi(\xi)$ admits analytic continuation to a strip
%$S_{(\lm,\lp)}:=\{\Im\xi \in (\lm,\lp)\}$ around the real axis, where $\lm<0\le \lp$.
If $\bE\left[e^{\be X_{T_q}}\right]<+\infty$ for $\be\in [\bemq, \bepq]$, %where $-\lm\le \bemq< 0<\bepq\le \lp$, 
then the operators $\cEq, \cEpq, \cEmq$ are well-defined
in spaces of measurable functions bounded by $C\left(e^{\bemq x}+e^{\bepq x}\right)$, where $C$ depends on a function but not on $x$, and
the Wiener-Hopf factors are well-defined and bounded on the closed strip $\{\Re\xi \in [\bem,\bep]\}$ and analytic in the open strip. 

The  Wiener-Hopf factorization formula in
the form used in probability \cite[p. 98]{RW}
is
\begin{equation}\label{whfprob}
\bE \left[ e^{X_{T_q}} \right]=\bE \left[ e^{\barX_{T_q}}
\right]\cdot \bE \bigl[ e^{\uX_{T_q}} \bigr],
\end{equation}
and its operator analog 
%\begin{equation}\label{WHFoper}
$\cE_q=\cE^-_q\cE^+_q=\cE^+_q\cE^-_q$
%\end{equation}
is proved similarly to \eqref{whfprob} (see, e.g., \cite{single,BLdouble}).  Finally,
introducing the notation
%\begin{equation}\label{e:def-phi-pm}
$ \phi^+_q(\xi)=\bE\bigl[ e^{i\xi\overline{X}_{T_q}} \bigr], 
\phi^-_q(\xi)=\bE\bigl[ e^{i\xi\underline{X}_{T_q}} \bigr] $
%\end{equation}
and noticing that
%\begin{equation}\label{e:E-xi}
$\bE\left[e^{i X_{T_q}\xi}\right] = \frac{q}{q+\psi(\xi)}$, %typo removed extra closed bracket
%\end{equation}
we can write \eqref{whfprob} in the form
\begin{equation}\label{whfanal}
\frac{q}{q+\psi(\xi)}=\phi^+_q(\xi)\phi^-_q(\xi).
\end{equation}
Equation \eqref{whfanal} is a special case of the factorization of
functions on the real line into a product of two functions analytic
in the upper and lower open half planes and admitting the continuous
continuation up to the real line. This is the initial factorization %typo factrization --> factorization
formula discovered by Wiener and Hopf  \cite{wh1931} in 1931, for
functions of a much more general form than in the LHS of
\eqref{whfanal}.

%MSdouble,BIL,hejd,single,BLdouble,NG-MBS,barrier-RLPE,iFT,paraLaplace,SINHregular,paired

\subsubsection{ Explicit formulas for and properties of  the Wiener-Hopf factors}\label{explWHF}
Let $X$ be a L\'evy process with characteristic exponent admitting analytic continuation to a strip $\Im\xi\in (\lm,\lp)$, $\lm<0<\lp$,
 around the real axis, and let $q>0$. Then (see, e.g., \cite{BLSIAM02,NG-MBS})

\mbr\noindent
(I)  there exist
$\sg_-(q)<0<\sg_+(q)$ such that
\begin{equation}\label{crucial}
q+\psi(\eta)\not\in (-\infty,0],\quad \Im\eta\in (\sg_-(q),\sg_+(q));
\end{equation}
\mbr\noindent
(II) the Wiener-Hopf factor $\phi^+_q(\xi)$ admits analytic continuation
to the half-plane $\Im\xi>\sg_-(q)$, and can be calculated as follows: for any $\om_-\in (\sg_-(q), \Im\xi)$,
\begin{equation}\label{phip1}
\phi^+_q(\xi)=\exp\left[\frac{1}{2\pi i}\int_{\Im\eta=\om_-}\frac{\xi \ln (q+\psi(\eta))}{\eta(\xi-\eta)}d\eta\right];
\end{equation}
\mbr\noindent
(III) the Wiener-Hopf factor $\phi^-_q(\xi)$ admits analytic continuation
to the half-plane $\Im\xi<\sg_+(q)$, and can be calculated as follows: for any $\om_+\in (\Im\xi, \sg_+(q))$,
\begin{equation}\label{phim1}
\phi^-_q(\xi)=\exp\left[-\frac{1}{2\pi i}\int_{\Im\eta=\om_+}\frac{\xi \ln (q+\psi(\eta))}{\eta(\xi-\eta)}d\eta\right].
\end{equation}
Note that one may use $1+\psi(\eta)/q$ instead of $q+\psi(\eta)$; the integrals \eqref{phip1}-\eqref{phim1} do no change.
Naturally, in this case, we require 
\begin{equation}\label{crucial2}
1+\psi(\eta)/q\not\in (-\infty,0],\quad \Im\eta\in (\sg_-(q),\sg_+(q)).
\end{equation}
Under additional conditions on $\psi$, there exist more efficient formulas
for the Wiener-Hopf factors. See Section \ref{WHFadd}.

\subsubsection{Analytic continuation of the Wiener-Hopf factors w.r.t. $\xi$, for $q$ fixed}\label{analcont} 
Under Assumption $(X)$,
$\psi(\eta)$ is analytic in the complex plane with two
cuts $i(-\infty, \lm]$ and $i[\lp,+\infty)$, and, for any $[\omm,\omp]\subset (\lm,\lp)$, $\Re\psi(\eta)\to+\infty$ as
$\eta\to \infty$ remaining in the strip. Hence, for $[\omm,\omp]\subset (\lm,\lp)$, there exists $\sg>0$ s.t. if $\Re q\ge \sg$,
then \eqref{crucial} holds. It follows that, for $q$ in the half-plane $\{\Re q\ge \sg\}$,  
\begin{enumerate}[(1)]
\item
\eqref{phip1} defines $\phipq(\xi)$ on the half-plane 
$\{\Im\xi>\omm\}$, and analytic continuation of $\phimq(\xi)$ to the strip $S_{(\omm,\omp)}$ can be defined by
\bbe\label{contphimqxi}
\phimq(\xi)=\frac{q}{(q+\psi(\xi))\phipq(\xi)};
\ee
\item
\eqref{phim1} defines $\phimq(\xi)$ on the half-plane 
$\{\Im\xi<\omp\}$, and analytic continuation of $\phipq(\xi)$ to the strip $S_{(\omm,\omp)}$ can be defined by
\bbe\label{contphipqxi}
\phipq(\xi)=\frac{q}{(q+\psi(\xi))\phimq(\xi)}.
\ee
\end{enumerate}
\begin{rem}{\rm It follows that $\phipq(\xi)$ (resp., $\phimq(\xi)$) admits analytic continuation w.r.t. $\xi$ to $\bC\setminus i(-\infty,\omm]$
(resp., $\bC\setminus i[\omp,+\infty)$).
}\end{rem}
For $\om_1,\om\in \bR$ and $b>0$, introduce the function
\[
\bC\ni y\mapsto \chi(\om_1,\om,b; y)=i\om_1+b\sinh(i\om+b)\in\bC,
\]
and the curve $\cL(\om_1,\om,b)= \chi(\om_1,\om,b; \bR)$, the image of the real line under $\chi(\om_1,\om,b; \cdot)$.
If $\om=0 (>0, <0)$, the curve is flat (wings point upward, wings point downward, respectively). For simplicity, we will use $\om\in [-\pi/2,\pi/2]$ only;
if $\psi$ admits analytic continuation to an appropriate Riemann surface, then $\om\not\in [-\pi/2,\pi/2]$ can be used, and then
the curve lies on this surface. See \cite{iFT,paraLaplace,paired, SINHregular} for details, and Fig.~\ref{ContoursSINH1} for an illustration.

Depending on the situation, we will need the Wiener-Hopf factors either on a curve  $\cL(\om_1,\om,b)$ with the wings pointing up ($\om>0$),
or on a curve  $\cL(\om'_1,\om',b')$ with the wings pointing down ($\om<0$). 
In the former case, we deform the contour of integration (in the formula for the Wiener-Hopf factors that we use)
so that the wings of the deformed contour $\cL(\om'_1,\om', b')$ point down ($\om'<0$), 
and in the latter case - up ($\om'>0$). This straightforward requirement is easy to  satisfy, as well as the second requirement:
the curves do not intersect. Indeed, if $\om\in (0,\pi/2]$ and $\om'\in [-\pi,2,0)$,
the curves do not intersect if and only if the point of intersection of the former with the imaginary axis is above
the point of the intersection of the latter with the same axis:
\bbe\label{no-intersect1}
\om_1+b\sin \om>\om'_1+b'\sin\om'.
\ee
The last condition on $\cL(\om'_1,\om', b')$ is that for $q$ of interest, the function
$\cL(\om'_1,\om', b')\ni \eta\mapsto 1+\psi(\eta)/q\in\bC$ (or $\eta\mapsto q+\psi(\eta)$)
is well-defined, and,  in the process of the deformation
of the initial line of integration into $\cL(\om'_1,\om', b')$,  image does not intersect
$(-\infty,0]$. If the parameters of the curve are fixed, this requirement is satisfied  if $\Re q\ge \sg$ and $\sg>0$ is sufficiently large.
For details, see \cite{paired}, where a different family of deformations (farctional-parabolic ones) was used.
In cases of the sinh-acceleration and fractional-parabolic deformation, at infinity,
the curves stabilize to rays, hence, the analysis  in
\cite{paired} can be used to derive the conditions on the deformation parameters if $q>0$.

If the Gaver-Stehfest method is applied, then we need to use the Wiener-Hopf factorization
technique for $q>0$ only, and the analysis in \cite{paired} suffices. An alternative method used in \cite{paraLaplace}
is to deform the contour of integration in the Bromwich intergral; in this case, the deformation of the latter
and the deformation of the contours in the formulas for the Wiener-Hopf factors must be in a certain agreement.
We outline the restrictions on the parameters of the deformations in Section \ref{LaplSinh}.

For positive $q$, the maximal (in absolute value) $\sg_\pm(q)$ are easy to find for all popular classes of L\'evy processes
used in finance. As it is proved in \cite{NG-MBS} (see also \cite{paired}), the equation $q+\psi(\xi)=0$ has
either 0 or 1 or two roots in the complex plane with the two cuts $i(-\infty,\lm]$ and $i[\lp,+\infty)$. 
Each root is purely imaginary, and of the form $-i\be^\pm_q$, where $-\bemq\in (0,\lp)$
and $-\bepq\in (\lm,0)$. If the root $-i\be^\pm_q$ exists, then $\sg_\mp(q)=-\be^\pm_q$ otherwise
$\sg_\pm(q)=\la_\pm$.

\begin{figure}
\scalebox{0.6}
{\includegraphics{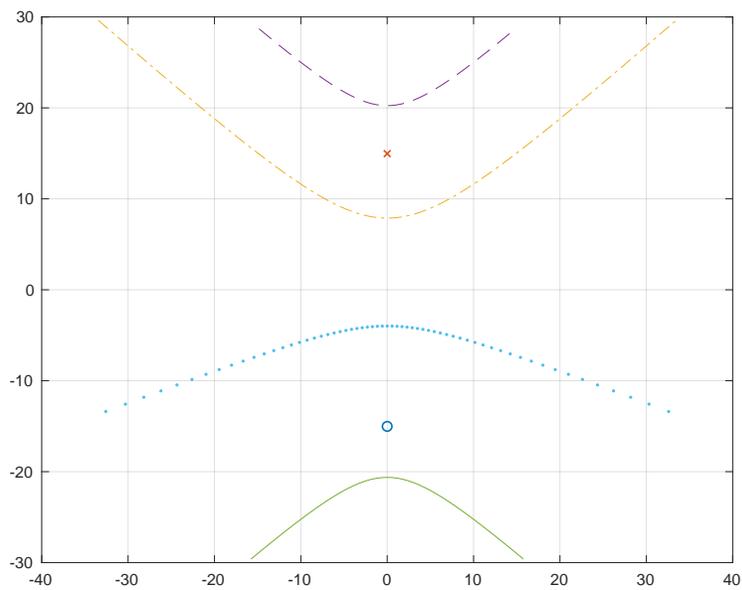}}
\caption{\small Examples of contours for KoBoL model. Parameters: $\nu=1.4, c=0.1466, \lm=-30, \lp=25$.
Cross and circle: $-i\be^\mp_q$.
%,the roots of the equation $q+\psi(\xi)=0$; $\sg^\pm(q):=-\be^\mp_q$.
Dashes: $\cL(15.0,0.71, 8.1)$. Dot-dashes:  $\cL(0,0.76, 12.2)$. Dots: $\cL(-1,-0.35, 8.6)$. Solid line:
$\cL(-15,-0.71, 12.1)$ }
\label{ContoursSINH1}
\end{figure}

% if \eqref{defPBMnu}-\eqref{defPBMnu2} are used, then an additional condition should be imposed:
%in the process of the deformation, the roots $-i\be^\pm_q$ are not crossed.

\subsection{Calculation of the Wiener-Hopf factors using the sinh-acceleration}\label{WHFSinh}
If for the Laplace inversion the Gaver-Stehfest method or other methods utilizing only positive $q$ is used,
%and $\psi$ is analytic in the complex plane with the cuts $i(-\infty,\lm]$, $i[\lp,+\infty)$
%(as it is the case for KoBoL, NIG, VG or many other L\'evy  processes used in finance),
 then
we can take any $\om\in [-\pi/2,\pi/2]$ and $\om'=-\om$. If $\psi$ admits analytic continuation 
to an appropriate Riemann surface (as it is the case for the some classes of L\'evy processes),
 then $\om\not\in [-\pi/2,\pi/2]$ can be used. The curve lies in the complex plane but the ``conical" region around the curve, in the $y$-coordinate,
 is a subset of the Riemann surface.
See \cite{iFT,paraLaplace,paired, SINHregular} for details.
In \cite{iFT,paraLaplace,paired}, fractional-parabolic deformations of the contour of integration  into a Riemann surface significantly increased
the rate of the decay of the integrand as compared to deformations into the complex plane with the two cuts; the number of terms
in the simplified trapezoid rule decreased by a factor of 10-1000 and more (see \cite{pitfalls} for a detailed analysis).
If the sinh-acceleration is used, the gain is not large if any: the number of terms in the simplified trapezoid rule decreases by a factor of 1.5-2, at best,
but the analytic expressions that one has to evaluate become more involved. Hence, in the present paper,
we use $\om,\om'\in [-\pi/2,\pi/2]$, of opposite signs. The sinh-acceleration has another advantage as compared to
the fractional-parabolic deformations. As it is shown in \cite{paired},
 if $q>0$ is small (which is the case for some terms in the Gaver-Stehfest formula if $T$ is large), then one of  
the $\sg_\pm(q)$ in Subsection \ref{explWHF} or both are small in absolute value, and then the strip of analyticity of the integrand is too narrow.
 Hence, if the fractional-parabolic change of variables is applied, the size $\ze$ of the mesh in the simplified trapezoid rule necessary 
 to satisfy the desired error tolerance $\eps$
  is too small, and the number of terms $N$ too large. A rescaling in the dual space can increase the width of the strip
  of analyticity but then the product $\La:=N\ze$ must increase, and the decrease in $N$ is insignificant. 
 If the sinh-acceleration is used, then the rescaling (using an appropriately small $b$) does not increase $\La$ significantly.
 Roughly speaking, in the recommendation for the choice $\La$, $1/\eps$ should be replaced with $1/\eps+a\ln(1/b)$, where
 $a$ is a moderate constant independent of $\eps, b$ and other parameters. As numerical examples in \cite{NewFamilies}
 demonstrate, this kind of rescaling is efficient even if the initial strip of analyticity is of the width $10^{-6}$ and less. 
 Hence, in this paper, we will use the Gaver-Stehfest method with the Rho-Wynn acceleration.
 % but produce numerical comparison
% with the results obtained when the sinh-acceleration in the Laplace inversion formula is used.
For each $q$, we use the following versions of \eqref{phip1}-\eqref{phim1}:
\begin{enumerate}[(i)]
\item
for $\xi\in \cL(\om_1,\om,b)$,
\begin{equation}\label{phip1sinh}
\phi^+_q(\xi)=\exp\left[\frac{b'}{2\pi i}\int_\bR\frac{\xi \ln (q+\psi(\chi(\om'_1,\om',b',y))}{\eta(\xi-\chi(\om'_1,\om',b',y))} b'\cosh(i\om'+y)dy\right];
\end{equation}
\item
for $\xi\in \cL(\om'_1,\om',b')$,
\begin{equation}\label{phim1sinh}
\phi^-_q(\xi)=\exp\left[-\frac{b}{2\pi i}\int_\bR\frac{\xi \ln (q+\psi(\chi(\om_1,\om,b,y))}{\eta(\xi-\chi(\om_1,\om,b,y))} b'\cosh(i\om+y)dy\right].
\end{equation}
\end{enumerate}
Each integral is evaluated applying the simplified trapezoid rule.

\subsection{Numerical examples}\label{numerWHF} In Table~\ref{TableWHF}, we apply the above scheme to calculate $\phi^\pm_q(\xi)$
in KoBoL model of order $\nu=1.2$. If the factors are calculated at 30 points, then approximately 1.5 msec is needed
to satisfy the error tolerance $\eps=10^{-15}$, and 1.0-1.2 msec are needed to satisfy the error tolerance $\eps=10^{-10}$.
The number of terms is in the range 350-385 in the first case, and 159-175 in the second case. 
To satisfy the error tolerance of the order of $10^{-20}$, about 500 terms would suffice but, naturally, high precision arithmetics
will be needed.
 
\section{Calculation of no-touch options and expectations of no-touch products}\label{no-touch dual}

\subsection{General formulas for no-touch options}\label{gen_f_no_touch}
In \cite{single,paraLaplace,paired}, it is proved that the Laplace transform $\tV_1(G;T,x)$
of $V_1(G;T,x)=e^{rT}V(G;T,x)$
 w.r.t. $T$ is given by
\bbe\label{LTnt_main}
\tV_1(G;q,x)=q^{-1}(\cEmq \bfo_{(h,+\infty)}\cEpq G)(x).
\ee
The result is proved under conditions more general than Assumptions $(X)$ and $(G_{emb})$.

Let $\cF$ denote the operator of the Fourier transform, and set $\Pi^+_h:=\cF \bfo_{(h,+\infty)} \cF^{-1}$. 
The operator $\Pi^+$ arises  
systematically in the theory of the Wiener-Hopf factorization and boundary value problems. 
See \cite{eskin} for the general setting in applications
to multi-dimensional problems. Using $\cF$ and $\Pi^+_h$ and taking into account that
$\cE^\pm_q=\cF^{-1}\phi^\pm_q\cF$, we rewrite \eqref{LTnt_main} as
\bbe\label{LTnt_main2}
\tV_1(G;q,x)=q^{-1}(\cF^{-1}\phimq \Pi^+_h\phipq \cF G)(x).
\ee

\begin{lem}\label{l:Pip1}
Let $f$ be an analytic function in a strip $S_{(\sgm,\sgp)}$, where $\sgm<\sgp$, that decays as $|\xi|^{-s}, s>1,$ as $\xi\to\infty$ remaining in
the strip.

Then $(\Pi^+_hf)(\xi)$ is analytic in the half-plane $\{\xi\ |\ \Im\xi<\sgp\}$, and
can be defined by any of the following three formulas:
\begin{enumerate}[a)]
\item
for any $\omp\in (\Im\xi, \sgp)$,
\begin{equation}\label{fPimax}
 \Pi^+_hf(\xi)=\frac{1}{2\pi i}\int_{\Im\eta=\om_+}\frac{e^{ih(\eta-\xi)}f(\eta)}{\xi-\eta}d\eta;
 \end{equation}
 \item
 for any $\omm\in (-\infty,\Im\xi)$,
 \begin{equation}\label{fPimin}
 \Pi^+_hf(\xi)=f(\xi)+\frac{1}{2\pi i}\int_{\Im\eta=\om_-}\frac{e^{ih(\eta-\xi)}f(\eta)}{\xi-\eta}d\eta;
 \end{equation}
 \item
 \begin{equation}\label{fPiH}
 \Pi^+_hf(\xi)=\frac{1}{2}f(\xi)+\frac{1}{2\pi i}{\rm v.p.}\int_{\Im\eta=\Im\xi}\frac{e^{ih(\eta-\xi)}f(\eta)}{\xi-\eta}d\eta,
 \end{equation}
 where ${\rm v.p.}$ denotes the Cauchy principal value of the integral.
 \end{enumerate}
 \end{lem}
\begin{proof} a) 
Applying the definition of $\Pi^+_h$ and Fubini's theorem,  we obtain
\begin{eqnarray}\nonumber
\Pi^+_hf(\xi)&=&(\cF\bfo_{(h,+\infty)}\cF^{-1}f)(\xi)%\\\nonumber&=&
=\int_h^{+\infty} e^{-ix\xi}(2\pi)^{-1}\int_{\Im\eta=\omp} e^{ix\eta}f(\eta)d\eta dx\\\label{Pif}
&=& \int_{\Im\eta=\omp} f(\eta) (2\pi)^{-1}\int_h^{+\infty} e^{i(\eta-\xi)x}dx d\eta%\\&=&
=\frac{1}{2\pi i}\int_{\Im\eta=\omp}\frac{e^{ih(\eta-\xi)}f(\eta)}{\xi-\eta}d\eta,
%\\\nonumber
%&=&\frac{1}{2\pi i}\int_{\cL^+_{\al;\omp,\be}}\frac{\hW_s(\eta)}{\xi-\eta}d\eta.
\end{eqnarray}
which proves \eqref{fPimax}.

b) We  push the line of integration in \eqref{Pif} down. On crossing the simple
pole at $\eta=\xi$, we apply the Cauchy residue theorem, and obtain \eqref{fPimin}.

c) Let $\om:=\Im\xi$ and $a=\Re\xi$. We deform the contour $\Im\xi=\om$ into
\[
\cL_\eps=i\om+((-\infty, a-\eps)\cup (a+\eps,+\infty))\cup\eps\{e^{i\varphi}\ |\ 0\le \varphi\le \pi\}
\]
and then pass to the limit $\eps\downarrow 0$. The result is \eqref{fPiH}.

\end{proof}
\begin{rem}{\rm \begin{enumerate}[a)]
\item
Let $\om=h=0$. Then \eqref{fPiH}
can be written in the form 
$\Pi^+_0=\frac{1}{2}I+\frac{1}{2i}\cH,$
where $I$ is the identity operator, and $\cH$ is the Hilbert transform.
Therefore,
a realization of $\Pi^+_0$ is, essentially, equivalent to a realization of $\cH$.
\item
 Similar formulas are valid in the multi-dimensional case, and for $\Pi^+_h$ acting in appropriate spaces of generalized functions.
 In particular, the condition on the rate of decay of $f$ at infinity can be relaxed.
 Under appropriate regularity conditions
on $f$, equation \eqref{fPiH} can be proved
for $f$ defined on the line $i\om+\bR$, the result being a function on the same line that admits analytic continuation
to the half-plane below this line.
 See \cite{eskin}.
 \end{enumerate}
 }\end{rem}

The following theorem is a part of the proof of the general theorem for pricing no-touch options in \cite{barrier-RLPE,NG-MBS};
we will outline the proof based on \eqref{LTnt_main2}.
\begin{thm}\label{thm:LTFT}
Let Assumptions $(X)$ and $(G_{emb})$ hold, and let $\omm\in (\lm,-\be)$ and 
$\omp\in (\omm, -\be)$. Then there exists $\sg>0$
such that for all $q$ in the half-plane $\{\Re q\ge \sg\}$, and all $x>h$,
\beqa\label{eq:LTFT}
\tV_1(G;q,x)&=&\frac{1}{2\pi q}\int_{\Im\xi=\omp}d\xi\,e^{ix\xi}
\phimq(\xi)\frac{1}{2\pi}\int_{\eta=\omm}d\eta\,\frac{e^{ih(\eta-\xi)}\phipq(\eta)\hG(\eta)}{i(\xi-\eta)}\\\label{eq:LTFTpr}
&&+\frac{1}{2\pi}\int_{\Im\xi=\omp}d\xi\,\frac{e^{ix\xi}\hG(\xi)}{q+\psi(\xi)}.
 \eqa
 \end{thm}
 \begin{proof} We can choose $\sg$ so that, for a fixed $q$ in the half-plane $\{\Re q\ge \sg\}$,
  $\phi^\pm_q(\eta)$ and $\hG(\eta)$ are analytic in the strip 
 $\eta\in S_{[\omm,\omp]}$. Hence, the product $\phipq(\eta)\hG(\eta)$ is analytic in the same strip, and, by Lemma 
 \ref{l:Pip1},  the function $(\Pi_h\phipq\hG)(\xi)$ is analytic in the half-plane $\{\Im\xi\in (-\infty, \omp)\}$.
 Applying \eqref {fPimin}, we obtain \eqref{eq:LTFT}-\eqref{eq:LTFTpr}. In more detail, 
 on the RHS of \eqref{eq:LTFTpr}, we obtain, first,
 \[
 \frac{1}{2\pi q}\int_{\Im\xi=\omp}d\xi\,e^{ix\xi}\phimq(\xi)\phipq(\xi)\hG(\xi),\]
 and then apply the Wiener-Hopf factorization formula \eqref{whfanal}.
 The integrals on the RHS of \eqref{eq:LTFT}-\eqref{eq:LTFTpr}  absolutely converge (or absolutely converge
 after integration by parts in the oscillatory integrals w.r.t. $\xi$) and their sum equals $\tV(G;q,x)$. See \cite{barrier-RLPE,NG-MBS,BIL}.
 
 \end{proof}
  Denote by $\cM$ the set of $q$'s
in the Gaver-Stehfest formula and by $Q(q)$ the weights. We have
an approximation
\bbe\label{GSt}
V(G,T;x)=e^{-rT}\sum_{q\in \cM} Q(q)e^{qT}\tV_1(G;q,x).
\ee
It remains to design an efficient numerical procedure for the evaluation of $\tV_1(G;q,x)$, $q>0$.
 
 \subsection{Sinh-acceleration in the integrals on the RHS of \eqref{eq:LTFT}-\eqref{eq:LTFTpr}}\label{sinh_no_touch_general}

 First, we rewrite 
  \eqref{eq:LTFT}-\eqref{eq:LTFTpr} as follows: for $x>h$,
 \beqa\label{eq:LTFT2}
\tV(G;q,x)&=&\frac{1}{2\pi q}\int_{\Im\xi=\omp}d\xi\,e^{i\xi(x-h)}
\phimq(\xi)\frac{1}{2\pi}\int_{\Im\eta=\omm}d\eta\,\frac{e^{-i\eta(a-h)}\phipq(\eta)\hG_0(\eta)}{i(\xi-\eta)}\\\label{eq:LTFT2pr}
&&+\frac{1}{2\pi}\int_{\Im\xi=\omp}d\xi\,\frac{e^{i(x-a)\xi}\hG_0(\xi)}{q+\psi(\xi)},
 \eqa
 where $a\ge h$ and $\hG_0$ satisfy the conditions in Assumption $(G_{emb})$.
 
 As $\eta\to \infty$, $\hG_0(\eta)\to 0$. If the order $\nu\in [1,2]$ or $\mu=0$, then $\phi^\pm_q(\xi)\to 0$ as $\xi\to \infty$ in the domain of analyticity; if $\nu<1$ and $\mu\neq 0$, then one of the Wiener-Hopf factors stabilizes to constant at infinity, and
 the other one decays as $1/|\xi|$. Since $x-h>0$,  it is advantageous to deform the outer contour 
 on the RHS of \eqref{eq:LTFT2} so that the wings
 of the deformed contour point upward: $\cL^+:=\cL(\om_1,\om,b)= \chi(\om_1,\om,b; \bR)$, where $\om>0$. At this stage, we
 deforms the contour so that no pole of the integrand (if it exists) is crossed, and consider the cases of crossing later.
 In all cases, in the process of deformations, the curves must remain in $S_{(\lm,\lp)}\cup \cC_\ga$, where $\cC_\ga$ is the cone in 
 Assumption $(G_{emb})$, and $S_{(\lm,\lp)}=\{\xi\ | \Im\xi\in (\lm,\lp)\}$. If Assumption $(G)$ holds, then we may take any
 $\ga\in (0,\pi/2)$; in our numerical experiments, we will take $\ga=\pi/4$.
 
 The type of the deformation of the inner contour depends on the sign of $a-h$. If $a\ge h$, which, in the case of puts and calls means that 
 the strike is above or at the barrier, then we may deform the contour downward. The deformed contour is of the form
 $\cL^-:=\cL(\om'_1,\om',b')= \chi(\om'_1,\om',b'; \bR)$, where $\om'<0$; as in the case of the deformation
 of the outer contour, we choose the parameters of the deformation so that, in the process of deformation, no pole of the inner integrand is 
 crossed, and consider the cases of crossing later.
 The parameters of the both contours are chosen so that
 $\cL^+$ is strictly above $\cL^-$. 
 %If $\om\in (0,\pi/2)$ and $\om'\in (-\pi/2,0)$, the necessary and sufficient condition is 
% $(\omm<)\om'_1+b'\sin\om'<\om_1+b\sin \om(<\omp)$. 
 The case $a<h$ (the strike is below the barrier) is reducible to the case $G(x)=\bfo_{[h,+\infty)}e^{\be x}$. 
 We have $\hG(\xi)=e^{\be h-ih\xi}\phipq(-i\be)/(i\xi-\be)$, hence, Assumption $(G)$ is valid with $a=h$.  
If $G$ is the value function of an embedded option, then $a<h$ is possible (e.g., the strike of the embedded European
 put or call is below the barrier). We consider this case separately, in Section \ref{strikebelowbarrier}.

 The type of deformation of the contour of integration on the RHS of \eqref{eq:LTFT2pr} depends on the sign
 of $x-a$. If $x-a>0$, we use $\cL^+$, and if $x-a<0$, then $\cL^-$. If $x-a=0$, then either deformation can be used.
 We conclude that, for the majority of applications when embedded options are not involved, we may use  deformed contours of the form $\cL^+$ in the outer integral and
 of the form $\cL^-$ in the inner integral, and, for
 $x>h$, write
 \beqa\label{eq:LTFT3}
&&\tV_1(G;q,x)\\\nonumber&=&\frac{1}{2\pi q}\int_{\cL^+}d\xi\,e^{i\xi(x-h)}
\phimq(\xi)\frac{1}{2\pi}\int_{\cL^-}d\eta\,\frac{e^{-i\eta(a-h)}\phipq(\eta)\hG_0(\eta)}{i(\xi-\eta)}+\frac{1}{2\pi}\int_{\cL'}d\xi\,\frac{e^{i(x-a)\xi}\hG_0(\xi)}{q+\psi(\xi)},
 \eqa
 where the type of the deformed contour $\cL'$ depends on the sign of $x-a$:
 $\cL'=\cL^+$ if $x-a\ge 0$, and $\cL'=\cL^-$ if $x-a\le 0$.
 To evaluate the integrals numerically, we make the changes of variables $\xi=i\om_1+b\sinh(i\om+y)$ and
$\eta=i\om'_1+b'\sinh (i\om'+y')$, and apply the simplified trapezoid rule w.r.t. $y$ and $y'$.

\subsection{Crossing poles}\label{cross_poles}
For simplicity, we consider the case $q>0$.  Hence, the results in this Subsection can be applied only if the Gaver-Stehfest method
or other similar methods are applied.\footnote{ Under additional assumptions on the process, similar constructions can be applied
if the contour deformation in the Bromwich integral is applied. See, e.g., \cite{UltraFast}, where the fractional-parabolic deformation was applied.} Assume that both solutions  $-i\be^\pm_q$ exist.

 \subsubsection{Case $h\le a\le x$, $x>h$} If
 $G(x)=(e^x-e^a)_+$ or $G(x)=(e^a-e^x)_+$, the condition $x\ge a$ means that the corresponding European call is 
 ITM or ATM, and put OTM or ATM.
  Recall that the initial lines of integration and the deformed contours are in the strip of analyticity of the integrand,
  around the real axis.
 It follows that the intersection of $\cL^+$ (resp., $\cL^-$)
 with the imaginary line $i\bR$ is below $-i\bemq$ (resp., above $-i\bemq$). On the RHS of \eqref{eq:LTFT3},
 we move $\cL^+$ up (resp., $\cL^-$ down),
 cross the simple pole at $-i\bemq$ (resp., at $-i\bepq$), and stop before the cut $i[\lp,+\infty)$ 
 (resp., $i(-\infty,\lm]$) is reached. On crossing the pole, we apply the residue theorem
 and the equalities
 \bbe\label{respsi}
 \lim_{\xi\to-i\be^\pm_q}q^{-1}\phi^\pm_q(\xi)(\xi+i\be^\pm_q)=\frac{1}{\psi'(-i\be^\pm_q)\phi^\mp_q(-i\be^\pm_q)}, 
 \ee
 which follow from \eqref{contphimqxi}-\eqref{contphipqxi}. Denote the new contours $\cL^{++}$ and $\cL^{--}$.
 In the last integral on the RHS of \eqref{eq:LTFT3},
 we deform $\cL'$ into $\cL^{++}$  (or a contour with the properties of $\cL^{++}$). In the process of the deformation,
 we may have to cross the poles of $\hG_0$ that are above the line $\{\Im\xi=\omm\}$ but below the contour $\cL^{++}$.
 Let $\cZ(\hG_0;\omm,\cL^{++})$ the set of these poles, and
 let  the poles be simple and different from $-i\bemq$; then
 \beqa\label{eq:LTFT4}
\tV_1(G;q,x)&=&\frac{1}{2\pi q}\int_{\cL^{++}}d\xi\,e^{i\xi(x-h)}
\phimq(\xi)\frac{1}{2\pi}\int_{\cL^{--}}d\eta\,\frac{e^{-i\eta(a-h)}\phipq(\eta)\hG_0(\eta)}{i(\xi-\eta)}\\\nonumber%\label{eq:LTFT42}
&&-\frac{e^{\bemq(x-h)}}{\phipq(-i\bemq)\psi'(-i\bemq)}
\frac{1}{2\pi}\int_{\cL^{--}}d\eta\,\frac{e^{-i\eta(a-h)}\phipq(\eta)\hG_0(\eta)}{\eta+i\bemq}\\\nonumber
&&-\frac{e^{-\bepq(a-h)}\hG_0(-i\bepq)}{\phimq(-i\bepq)\psi'(-i\bepq)}
\frac{1}{2\pi }\int_{\cL^{++}}d\xi\,e^{i\xi(x-h)}\frac{\phimq(\xi)}{\xi+i\bepq}
\\\nonumber
&&-\frac{e^{\bemq(x-h)}}{\phipq(-i\bemq)\psi'(-i\bemq)}
\frac{e^{-\bepq(a-h)}\hG_0(-i\bepq)}{\phimq(-i\bepq)\psi'(-i\bepq)(\bepq-\bemq)}\\\nonumber
&&+\frac{1}{2\pi}\int_{\cL^{++}}d\xi\,\frac{e^{i(x-a)\xi}\hG_0(\xi)}{q+\psi(\xi)} 
 +\frac{ie^{\bemq(x-a)}\hG_0(-i\bemq)}{\psi'(-i\bemq)}\\\nonumber
 &&+\sum_{z\in \cZ(\hG_0;\omm,\cL^{++})}\frac{ie^{-iz(x-a)}}{q+\psi(z)}\lim_{z'\to z}\hG_0(z')(z'-z).
 \eqa
 Note that if $-i\be^\pm_q$ do not exist or not crossed in the process of the deformations, than
 \eqref{eq:LTFT4} is valid after all terms with $\be^\pm_q$ are removed. In applications to pricing call options,
 $\hG(\eta)=e^{-ia\eta}\hG_0(\eta)$, where $a=\ln K$, and
% \bbe\label{hGcall}
$\hG_0(\eta)=-\frac{Ke^{-rT}}{\eta(\eta+i)}$
%\ee
has two simple poles at $-i,0$, it is advantageous to cross these two poles even if the poles $-i\be^\pm_q$ are not crossed:
 \beqa\label{eq:LTFT4call}
\tV_{call}(q,x)&=&-\frac{Ke^{-rT}}{(2\pi)^2 q}\int_{\cL^{+}}d\xi\,e^{i\xi(x-h)}
\phimq(\xi)\int_{\cL^{-}}d\eta\,\frac{e^{-i\eta(a-h)}\phipq(\eta)}{\eta(\eta+i)i(\xi-\eta)}\\\nonumber%\label{eq:LTFT42}
&&-\frac{Ke^{-rT}}{2\pi}\int_{\cL^{+}}d\xi\,\frac{e^{i(x-a)\xi}}{(q+\psi(\xi))\xi(\xi+i)}+
Ke^{-rT}\left(\frac{e^{x-a}}{q+\psi(-i)}-\frac{1}{q}\right).
\eqa
 \subsubsection{Case $h<x< a$} In typical examples 
 $G(x)=(e^x-e^a)_+$ or $G(x)=(e^a-e^x)_+$, the corresponding European call is OTM, and put ITM.
 In the case $x\ge a$, the last two terms on the RHS of \eqref{eq:LTFT4}
 appear as we push up the line of integration in the last integral on the RHS of
 \eqref{eq:LTFT3}.  Now $x\ge a$, hence, we deform this line down, and obtain \eqref{eq:LTFT4} with the following modification:
 the last three terms on the RHS are replaced with the sum
 \bbe\label{eq:LTFT4ad}
 \frac{1}{2\pi}\int_{\cL^{--}}d\xi\,\frac{e^{i(x-a)\xi}\hG_0(\xi)}{q+\psi(\xi)} 
 -\frac{ie^{\bepq(x-a)}\hG_0(-i\bepq)}{\psi'(-i\bepq)}
\ee

 \subsubsection{Approximate formulas in the case of a wide strip of analyticity} If $-\lm,\lp$ are large,  $q$ is not large, so that
 $-i\be^\pm_q$ are not large in absolute value, and $x-h$ and $x-a$ are nor very small
 in absolute value, then we can choose the deformations $\cL^{++}$ and $\cL^{--}$
 so that the integrals over $\cL^{++}$ and $\cL^{--}$ are small, and can be omitted (we do not copy-paste the result in order to save space); a small error results.
 The resulting formulas are given by simple analytical expressions in terms of $\be^\pm_q$ and
 $\phimq(-i\bepq), \phipq(-i\bemq), \psi'(-i\be^\pm_q)$ (see
 \cite{paired}). E.g., if $h<x<a$, then
 \beqa\label{asympGSt}
 &&\tV_1(G;q,x)\\\nonumber&=&-\frac{e^{\bemq(x-h)}}{\phipq(-i\bemq)\psi'(-i\bemq)}
\frac{e^{-\bepq(a-h)}\hG_0(-i\bepq)}{\phimq(-i\bepq)\psi'(-i\bepq)(\bepq-\bemq)}
-\frac{ie^{\bepq(x-a)}\hG_0(-i\bepq)}{\psi'(-i\bepq)}
\eqa

\subsection{Case of $G$ the value function of an embedded put option, with the strike below the barrier}\label{strikebelowbarrier}
Let $a<h$; then $x-a>0$. We deform both contours on the RHS of \eqref{eq:LTFT2} up. First, 
 we deform the outer contour into $\cL^+=\cL(\om_1,\om,b)$,
 and then the inner one, into $\cL^+_2=\cL(\om_{12},\om_2,b_2)$. We choose the parameters of the latter$(\om_{12},\om_2,b_2)$
  so that $\cL^+_2$ is strictly below $\cL^+_2$, and, furthermore,
 so that the angles between the asymptotes of the two contours are positive: $0<\om_2<\om$. Instead of 
 \eqref{eq:LTFT3}, we have 
 \beqa\label{eq:LTFTaleh}
&&\tV_1(G;q,x)\\\nonumber
&=&\frac{1}{2\pi q}\int_{\cL^+}d\xi\,e^{i\xi(x-h)}
\phimq(\xi)\frac{1}{2\pi}\int_{\cL^+_2}d\eta\,\frac{e^{-i\eta(a-h)}\phipq(\eta)\hG_0(\eta)}{i(\xi-\eta)}+\frac{1}{2\pi}\int_{\cL^+}d\xi\,\frac{e^{i(x-a)\xi}\hG_0(\xi)}{q+\psi(\xi)}.
 \eqa
The crossing of possible poles can be done similarly to the case $a\ge h$.

\subsection{Numerical examples}\label{examples:no-touch}
\subsubsection{Pricing no-touch options and first touch digitals, down case (Table~\ref{TableNT})} 
Let $V_{nt}(T,x)$ and $V_{ft}(T,x)$ be the prices of the no-touch and first touch digital option, respectively.
Then, for any $\om\in (\sg_-(q),0)$ and $\om_+\in (0,\sg_+(q))$, $\tV_{nt}(G;q,x)$  is given by 
\beqa\nonumber
\tV_{nt}(G;q,x)&=&\frac{e^{-rT}}{2\pi q} \int_{\Im\xi=\om}e^{i(x-h)\xi}\frac{\phimq(\xi)}{-i\xi}d\xi
\\\label{Vntq}&=&
\frac{e^{-rT}}{q}+\frac{e^{-rT}}{2\pi q} \int_{\Im\xi=\om_+}e^{i(x-h)\xi}\frac{\phimq(\xi)}{-i\xi}d\xi.
\eqa
A similar formula for  $V_{ft}(T,x)$ (see \cite{paired}) is
\bbe\label{Vftq}
\tV_{ft}(G;q,x)=-\frac{e^{-rT}}{2\pi (q-r)}\int_{\Im\xi=\om_+}e^{i(x-h)\xi}\frac{\phimq(\xi)}{-i\xi}d\xi,
\ee
where $\om_+\in (0, \sg_+(q-r))$. The integrals on the RHSs of \eqref{Vntq} and \eqref{Vftq} differ by scalar factors
in front of the integrals, hence, both options can be evaluated simultataneously. If $-\bemq$ exists and $\lp+\bemq>0$ is not small,
it may be advantageous  to move the line of integration up and cross the simple pole at $-i\bemq$.

In both cases,  the line of integration is deformed into a contour $\cL_+:=\cL(\om_{1+}, \om_+, b_+)$ in the upper half-plane,
with wings pointing up. To evaluate $\phimq(\xi), \xi\in \cL_+$, we use \eqref{phip1} to calculate $\phipq(\xi)$
(the line of integration is deformed into a contour in the lower half-plane with the wings pointing down), and then
\eqref{contphimqxi}. We calculate prices of no-touch options and first touch digitals in the same KoBoL model as in Table~\ref{StVar1_1.2}. 
Our numerical experiments
show that, for integrals for the Wiener-Hopf factor and the Fourier inversion, the relative error of the order of
E-8  and better (for prices) can be satisfied using the general recommendations of 
\cite{SINHregular}
 for the error tolerance $\eps=10^{-10}$; if 
$\eps=10^{-15}$ is used, the pricing errors decrease insignificantly. Hence, in this example, 
the errors of GWR method are of the order of $10^{-10}-10^{-8}$. The CPU time (for 7 spots) is less than 10 msec for the
error tolerance $\eps=10^{-6}$, and less than 25 msec for the
error tolerance $\eps=10^{-10}$.
The bulk of the CPU time is spent on the calculation of the Wiener-Hopf factor at the points of the chosen grids;
this calculation can be easily paralellized, and the total CPU time significantly decreased. The CPU
time can be decreased further
using more efficient representations for the Wiener-Hopf factors. 

\subsubsection{Pricing down-and-out call option (Tables~\ref{TableDOCall}-\ref{TableDOCallClose})} We consider the same model, and the call options
with the strikes $K=1.04,1.1$ In both cases, the barrier $H=1$, and $T=0.1, 0.5$. In \eqref{eq:LTFT4call},
we make the corresponding changes of variables and apply the simplified trapezoid rule.
In Table~\ref{TableDOCall}, we see that even in cases when the spot is close to the barrier, the error tolerance of the order of
E-05 and smaller can be 
satisfied at a moderate CPU-time cost: about 0.1 msec for the calculation at 7 spots. 

In Table~\ref{TableDOCallClose}, we show the prices $V_{call}(H,K; T, S)$ for $S$ very close to the barrier $H$ ($x:=\ln(S/H)\in [0.0005, 0.0035] $
and the prices divided by $x^{\nu/2}$. We see that the ratio is approximately constant which agrees with the
asymptotics of the price of the down-and-out option 
\[
V_{call}(H,K; T, S)\sim c(K,T)x^{\nu/2}, \quad x\downarrow 0.
\]

\section{Pricing first-touch options and expectations of first-touch products}\label{first-touch dual}
Conditions  on processes and payoff functions are the same as in Section \ref{condXGNt}.
We consider the down-and-in case; $H=e^h$ is the barrier, $T$ is the maturity date,
and $G(\tau, X_\tau)$ is the payoff at the first entrance time $\tau$ into $(-\infty,h]$. %Denote by $\tau=T-t$ time to maturity.
We need to calculate $V(G;T;x)=\bE^x\left[G(\tau,X_\tau)\bfo_{\tau<T}\right]$. 

\subsection{The simplest case}\label{first-touch-simplest}
Let $G(\tau,X_\tau)=e^{-r\tau+\be X_\tau}$, where $\be\in [0,-\lm)$. For $\be=0$, $V(G;T;x)$ is
the time-0 price of the  first-touch digital, for $\be=1$ of the stock which is due at time $\tau$ if $\tau<T$.
The case of the down-and-in forwards obtains by linearity. If we need to calculate the expectation of the product of two payoffs of this form,
 $\be$ may assume value $\be=2$; in this case, we need to replace $r$ with $2r$. 
 %With different betas, we can consider
%the cases of power forwards. 

Since $\be\in [0, -\lm)$, $\phimq(-i\be)=\bE\left[e^{\be\barX_{T_q}}\right]$ is 
finite for any $q$ in the right half-plane. Furthermore, for any $\sg>0$, there exists $\om>0$ such that, for any $q$ in the half-plane $\Re q\ge \sg$,
$\phimq(\xi)$ admits analytic continuation to the half-plane $\{\Im\xi<\om\}$. For such $\sg$ and $\om$, the
general formula for the down-and-out options derived in \cite{barrier-RLPE,single} (see also \cite{paraLaplace,paired})
is applicable: for $x>h$,
\bbe\label{FTV}
V(G;T;x)=\frac{e^{\be h-rT}}{2\pi i}\int_{\Re q=\sg}dq\, e^{qT}\phimq(-i\be)^{-1}\frac{1}{2\pi}\int_{\Im\xi=\om}e^{i(x-h)\xi}\frac{\phimq(\xi)}{\be-i\xi}d\xi.
\ee
Assuming that we use the Gaver-Stehfest method to evaluate the Bromwich integral (the outer integral on the RHS of \eqref{FTV}), 
we need to calculate the inner integral for $q>0$.
%\[
%I(\be)=\frac{1}{2\pi}\int_{\Im\xi=\om}e^{i(x-h)\xi}\frac{\phimq(\xi)}{\be-i\xi}d\xi.
%\]
As in the case of no-touch options, we deform the contour upward and cross the simple pole at $-i\bemq$ if $\bemq$ exists
and $-\bemq$ is not close to $\lp$:%. The result is
\[
\frac{1}{2\pi}\int_{\Im\xi=\om}e^{i(x-h)\xi}\frac{\phimq(\xi)}{\be-i\xi}d\xi=\frac{1}{2\pi}\int_{\cL^{++}}e^{i(x-h)\xi}\frac{\phimq(\xi)}{\be-i\xi}d\xi-\frac{qe^{\bemq(x-h)}}{\phipq(-i\bemq)\psi'(-i\bemq)}.
\]
If $x-h>0$ is not very small, $q$ not too large and $\lp$ and $-\lm$ are large, then a good approximation can be obtained using
the last term on the RHS above and omitting the integral over $\cL^{++}$.

\subsection{General case}\label{call_first_touch}
Even in a simple case of the down-and-in option which, at time $\tau$, becomes the call option with strike $K$ and time $T-\tau$
to maturity,   $G(\tau, X_\tau)=e^{-r\tau}V_{\rm call}(K;T-\tau,X_\tau)$, the payoff at time $\tau$, is  more involved than
the payoff functions
in  \cite{barrier-RLPE,single}  because the pricing formulas in  \cite{barrier-RLPE,single}  were derived
under assumption that the dependence  of $G(\tau, X_\tau)$ on $\tau$ is of the simplest form $e^{-r\tau}G_1(X_\tau)$.
If we calculate the expectation of the product of a discounted down-and-in call option $e^{-r\tau}V_{\rm call}(K;T-\tau,X_\tau)$
and $e^{-r\tau+\be X_\tau}$, then
$G(\tau, X_\tau)=e^{-2r\tau+\be X_\tau}V_{\rm call}(K;T-\tau,X_\tau)$. In the case of the expectation of the product
of discounted prices of European options, the structure of $G(\tau, X_\tau)$ is even more involved.

First, we 
calculate the expectation $V(G;T;x)=\bE^x[\bfo_{\tau<T}G(\tau, X_\tau)]$ in the general form, and then
make further steps for the special cases mentioned above.
We repeat the main steps of the initial proof in \cite{barrier-RLPE} omitting the technical details of the justification
of the application of Fubini's theorem; they are the same as in \cite{barrier-RLPE,BIL}.  Let $L=-\psi(D)$
be the infinitesimal generator of $X$. Recall that the pseudo-differential operator $\psi(D)$ with the symbol
$\psi$ is the composition of the Fourier transform, multiplication operator by the function $\psi$, and inverse Fourier transform.
If $\hu$ is well-defined and analytic in a strip, $\psi$ admits analytic continuation to the same strip,
 and the product $\psi(\xi)\hu(\xi)$ decays sufficiently fast as $\xi\to\infty$ remaining in the strip, an equivalent definition is 
$\widehat{(\psi(D)u)}(\xi)=\psi(\xi) \hu(\xi)$, for $\xi$ in the strip. For details, see \cite{eskin,NG-MBS,barrier-RLPE}.

The function $V_1(G;t,x):=V(G;T-t;x)$
 is the bounded sufficiently regular solution of the boundary value problem
\beqa\label{BVftdown}
(\dd_t+\psi(D))V_1(G;t,x)&=&0,\quad t>0, x>h;\\\label{BVdownbc1}
V_1(G;t,x)&=& G(t,x), \quad t>0, x\le h;\\\label{BVdownbc2}
V_1(G;0,x)&=&0,\quad x\in\bR.
\eqa
Making the Laplace transform w.r.t. $t$, we obtain that if $\sg>0$ is sufficiently large, then, for all $q$ in the half-plane
$\{\Re\,q\ge \sg\}$, $\tV_1(G;q,x)$ solves the boundary problem
\beqa\label{BVftdownL}
(q+\psi(D))\tV_1(G;q,x)&=&0,\quad  x>h;\\\label{BVdownbc1L}
\tV_1(G;q,x)&=& \tG(q,x), \quad x\le h,
\eqa
in the class of sufficiently regular bounded functions. If $\{\tV_1(G;q,\cdot)\}_{\Re q\ge \sg}$ is
the (sufficiently regular) solution of the family of boundary problems \eqref{BVftdownL}-\eqref{BVdownbc1L}
 on $\bR$, then $V_1(G;t,x)$ can be found using the Laplace inversion
formula. Finally, $V(G;t,x)=V_1(G;T-t,x)$.

The family of problems \eqref{BVftdownL}-\eqref{BVdownbc1L}
 is similar to the one in \cite{barrier-RLPE}; the only difference is a more involved dependence of $\tG$ 
on $q$ (in \cite{barrier-RLPE}, $\tG(q,x)=G(x)/(q-r)$).
Hence, we can apply the Wiener-Hopf factorization technique as in \cite{barrier-RLPE} and obtain
\bbe\label{FTsolL}
\tV_1(G;q,\cdot)=\cEmq\bfo_{(-\infty,h]}(\cEmq)^{-1}\tG(q,\cdot).
\ee
Let $\om\in (0,\lp)$. If $\sg>0$ is sufficiently large, then, for $q$ in the half-plane $\{\Re q\ge \sg\}$, $\phimq(\xi)$ is analytic in
the half-plane $\{\Im\xi\le \om\}$. For $\xi$ in this half-plane,
the double Laplace-Fourier transform of $V_1(G;t,x)$ w.r.t. $(q,x)$ is given by
\bbe\label{FTsolLF}
\hV_1(G;q,\xi)=\phimq(\xi)\widehat{\bfo_{(-\infty,h]}W(q,\cdot)}(\xi),
\ee
where $\widehat{\bfo_{(-\infty,h]}W(q,\cdot)}(\xi)$ is the Fourier transform of $\bfo_{(-\infty,h]}W(q,\cdot)$,
 $W(q,\cdot)=(\cEmq)^{-1}\tG(q,\cdot)$.
We take $\om'<-\be$, and, similarly to the proof of Lemma \ref{l:Pip1},  calculate
\beqast
\widehat{\bfo_{(-\infty,h]}W(q,\cdot)}(\xi)&=&
\int_{-\infty}^h dx\, e^{-ix\xi}(2\pi)^{-1}\int_{\Im\eta=\om'}d\eta\, e^{ix\eta}\phimq(\eta)^{-1}\hG(q,\eta) \\
&=& \int_{\Im\eta=\om'}\phimq(\eta)^{-1}\hG(q,\eta) (2\pi)^{-1}\int_{-\infty}^he^{i(\eta-\xi)x}dx d\eta.
\eqast
The result is: for $(q,\xi)$ s.t. $\Re q\ge \sg$ and $\Im\xi>\om'$, 
\bbe\label{FTsolLFW}
\widehat{\bfo_{(-\infty,h]}W(q,\cdot)}(\xi)=\frac{1}{2\pi i}\int_{\Im\eta=\om'}\frac{e^{ih(\eta-\xi)}\phimq(\eta)^{-1}\hG(q,\eta)}{\eta-\xi}d\eta.
\ee
Applying the inverse Fourier transform, we obtain
\beqa\label{tV1}
\tV_1(G;q,x)=\frac{1}{2\pi}\int_{\Im\xi=\om}d\xi\,e^{i(x-h)\xi}\phimq(\xi)
\frac{1}{2\pi i}\int_{\Im\eta=\om'}\frac{e^{ih\eta}\hG(q,\eta)}{\phimq(\eta)(\eta-\xi)}d\eta.
\eqa
Since $x-h>0$, we deform the outer contour upward, the new contour being of the type $\cL^+$
(meaning: of the form  $\cL(\om_1,\om,b)$, where $\om>0$):
\beqa\label{tV12}
\tV_1(G;q,x)=\frac{1}{2\pi}\int_{\cL^+}d\xi\,e^{i(x-h)\xi}\phimq(\xi)
\frac{1}{2\pi i}\int_{\Im\eta=\om'}\frac{e^{ih\eta}\hG(q,\eta)}{\phimq(\eta)(\eta-\xi)}d\eta.
\eqa
If $-i\bemq$ exists and $-\bemq$ is not close to $\lp$, we push the contour up, cross the simple pole at $\xi=-i\bemq$,
and obtain
\beqa\label{tV13}
\tV_1(G;q,x)&=&\frac{1}{2\pi}\int_{\cL^{++}}d\xi\,e^{i(x-h)\xi}\phimq(\xi)
\frac{1}{2\pi i}\int_{\Im\eta=\om'}\frac{e^{ih\eta}\hG(q,\eta)}{\phimq(\eta)(\eta-\xi)}d\eta\\\label{tV13pr}
&&
+\frac{qe^{\bemq(x-h)}}{\phipq(-i\bemq)\psi'(-i\phimq)}
\frac{1}{2\pi }\int_{\Im\eta=\om'}\frac{e^{ih\eta}\hG(q,\eta)}{\phimq(\eta)(\eta+i\bemq)}.
\eqa
%Here and below, $\cL^{++}$ denotes a contour of type $\cL^+$ which is above $-i\bemq$.
An admissible type of the deformation  of the inner integral on the RHS of \eqref{tV13} and the integral on the RHS of
\eqref{tV13pr} depends on the properties of $e^{ih\eta}\hG(q,\eta)$. If $G$ is the price of a vanilla option
or the product of prices of two vanilla options, then an admissible deformation is determined by the relative position of the barrier
and the strikes of the options involved. Hence, we are forced to consider several cases.
%\subsection{Calculation of  $\hG(q,\eta)$}\label{calchGeta}
%: the cases of a European call option and the product of the latter and the discounted $e^{\be X_\tau}$}
\subsection{Down-and-in call and put options}\label{oneEuro}
\subsubsection{Call option, the strike is at or above the barrier} \label{oneEuroabovecall}
Consider first the down-and-in call option.
Since the strike is at or above the barrier, $a:=\ln K\ge h=\ln H$.
 We must have $\lm<-1$. Let $\om'\in (\lm,-1)$ and $\sg>0$ be such that $\Re(q+\psi(\xi))>0$ if $\Re q\ge \sg$ and $\Im\xi\in [\om',-1)$.
 Then 
the double Laplace-Fourier transform w.r.t. $(t,x)$ of the discounted price $G(t,x)=e^{-r(T-t)}V_{\mathrm{call}}(T,K;T-t, x)$ 
 is well-defined in the region $\{(q,\eta)\ |\ \Re q\ge \sg, \Im\eta\in [\om',-1)\}$, and it is given by
\[
\hG(q,\eta)=e^{-rT}\int_0^{+\infty}e^{-qt}\frac{K^{1-i\eta}e^{-t\psi(\eta)}}{i\eta(i\eta-1)}dt.
\]
Integrating, we obtain $\hG(q,\xi)=e^{-ia\xi}\hG_0(q,\xi)$, where $a=\ln K$, and 
\bbe\label{hGLFcall}
\hG_0(q,\eta)=\frac{Ke^{-rT}}{(q+\psi(\eta))i\eta(i\eta-1)}=-\frac{Ke^{-rT}}{(q+\psi(\eta))\eta(\eta+i)}.
\ee
Under condition $(X)$, for each $q>0$, $\hG_0(q,\cdot)$ is a meromorphic function 
in the complex plane with two cuts $i(-\infty,\lm]$, $i[\lp,+\infty)$ and simple  poles at $0$, $-i$
(and $-i\be^\pm_q$, if the latter exist).
Since $h<a$, we may deform the inner contour of integration on the RHS of \eqref{tV13} and the contour
of integration on the RHS of \eqref{tV13pr} down. On the strength of \eqref{whfanal}, 
\[
\frac{\hG_0(q,\eta)}{\phimq(\eta)}=-\frac{Ke^{-rT}}{\phimq(\eta)(q+\psi(\eta))\eta(\eta+i)}=-\frac{Ke^{-rT}\phipq(\eta)}{q\eta(\eta+i)}.
\]
Hence,   \eqref{tV12} becomes 
\beqa\label{tV121}
\tV_1(G;q,x)=-\frac{Ke^{-rT}}{2\pi q}\int_{\cL^{+}}d\xi\,e^{i(x-h)\xi}\phimq(\xi)
\frac{1}{2\pi i}\int_{\Im\eta=\om'}\frac{e^{i(h-a)\eta}\phipq(\eta)}{\eta(\eta+i)(\eta-\xi)}d\eta,
\eqa
and \eqref{tV13}-\eqref{tV13pr} can be rewritten as 
\beqa\label{tV13a}
\tV_1(G;q,x)
%\\\nonumber
&=&-\frac{Ke^{-rT}}{2\pi q}\int_{\cL^{++}}d\xi\,e^{i(x-h)\xi}\phimq(\xi)
\frac{1}{2\pi i}\int_{\Im\eta=\om'}\frac{e^{i(h-a)\eta}\phipq(\eta)}{\eta(\eta+i)(\eta-\xi)}d\eta\\\nonumber
&&
-\frac{e^{\bemq(x-h)}}{\phipq(-i\bemq)\psi'(-i\phimq)}
\frac{Ke^{-rT}}{2\pi }\int_{\Im\eta=\om'}\frac{e^{i(h-a)\eta}\phipq(\eta)}{\eta(\eta+i)(\eta+i\bemq)}d\eta.
\eqa
If $-i\bepq$ exists, we can cross the pole at $-i\bepq$, and obtain
\beqa\label{tV14}
\tV_1(G;q,x)&=&\frac{Ke^{-rT}}{2\pi}\int_{\cL^{++}}d\xi\,e^{i(x-h)\xi}\phimq(\xi)
\left(-\frac{1}{2\pi i}\int_{\cL^{--}}\frac{e^{i(h-a)\eta}\phipq(\eta)}{q\eta(\eta+i)(\eta-\xi)}d\eta\right.\\\nonumber
&&\left. \hskip3cm
+\frac{e^{\bepq(h-a)}}{\phimq(-i\bepq)\psi'(-i\bepq)\bepq(\bepq-1)(i\bepq+\xi)}\right)
\\\nonumber
&&
-\frac{e^{\bemq(x-h)}}{\phipq(-i\bemq)\psi'(-i\phimq)}
\frac{Ke^{-rT}}{2\pi }\int_{\cL^{--}}\frac{e^{i(h-a)\eta}\phipq(\eta)}{q\eta(\eta+i)(\eta+i\bepq)}d\eta\\\nonumber
&&+\frac{e^{\bemq(x-h)}}{\phipq(-i\bemq)\psi'(-i\bemq)}
\frac{Ke^{-rT}e^{\bepq(h-a)}}{\phimq(-i\bepq)\psi'(-i\bepq)\bepq(\bepq-1)(\bepq-\bemq)}
\eqa
\subsubsection{Put option, the strike is at or above the barrier} \label{oneEuroaboveput}
In the case of the put option, we have \eqref{tV13a} with $\om'\in (0,\om)$. 
Hence, deforming the contours of integrations w.r.t. $\eta$ down into $\cL^{--}$, we cross not only a simple pole at
$-i\bepq$ but simple poles $0, -i$ as well.  Hence, 
to the RHS of \eqref{tV14}, we need to add
\beqa\label{tV14putadd}
\tV_{1,add}(G;q,x)&=&-\frac{Ke^{-rT}}{2\pi}\int_{\cL^{++}}d\xi\,e^{i(x-h)\xi}\phimq(\xi)
\left(\frac{1}{iq\xi}+\frac{e^{h-a}\phipq(-i)}{q(1-i\xi)}\right)
\\\nonumber
&&-\frac{Ke^{-rT}qe^{\bemq(x-h)}}{\phipq(-i\bemq)\psi'(-i\bemq)}\left(\frac{1}{q\bemq}
+\frac{e^{h-a}\phipq(-i)}{q(1-\bemq)}\right)
\eqa

\subsubsection{Call option, the strike is below the barrier}\label{Eurocallstrikebelow}
We start with the contour $\{\Im\eta=\om'\}$, where $\om'\in (-\bepq,-1)$. Since $a<h$, we need to deform  
 the inner contour of integration on the RHS of \eqref{tV13} and the contour
of integration in \eqref{tV13pr} up. Since $\phipq(\eta)$ is analytic and bounded in the half-plane
$\{\Im\eta\ge \om'\}$, the inner integrand on the RHS of \eqref{tV13a} has
three simple poles at $\eta=-i,0,\xi$, and the 1D integrand on the RHS of \eqref{tV13a} has three simple poles at $-i,0,-i\bemq$;
after the poles are crossed, we can move the line of integration up to infinity, and show that the integrals
after crossing are zeroes. Hence, we obtain
\beqa\label{tV15}
\tV_1(G;q,x)&=&-\frac{Ke^{-rT}}{2\pi}\int_{\cL^{++}}\frac{e^{i(x-a)\xi}d\xi}{(q+\psi(\xi))
\xi(\xi+i)}
\\\nonumber&&+\frac{Ke^{-rT}e^{\bemq(x-a)}}{\psi'(-i\phimq)
q\bemq(\bemq-1)}-\tV_{1,add}(G;q,x).
\eqa

\begin{rem}{\rm If $-i\be^\pm_q$ either do not exist or not crossed, in all cases above the contours
of the types $\cL^\pm$ are used, and, in all formulas above, all the terms that contain $\be^\pm_q$ should
be omitted.
}\end{rem}
\subsubsection{Put option, the strike is below the barrier}\label{Europutstrikebelow}
Evidently, the price is the same as the one of the European put.

\subsubsection{The case of the product of discounted a European call or put option  
and $e^{X_\tau}$}\label{expXprodEuro}
 In the case of the call option, we have $G(t,x)=e^{-2rT}e^{x+2rt} V_{\mathrm{call}}(T,K;T-t,x)$. Assuming that 
$\lm<-2$, we take $\om'\in (\lm+1,-1)$, and write
\beqast
G(t,x)&=&e^{-2rT}e^{x+rt}\frac{1}{2\pi}\int_{\Im\xi=\om}e^{ix\eta-t\psi(\eta)}\frac{K^{1-i\eta}}{i\eta(i\eta-1)}d\eta\\
&=&e^{-2rT}e^{rt}\frac{1}{2\pi}\int_{\Im\xi=\om}e^{ix(\eta-i)-t\psi(\eta)}\frac{K^{1-i\eta}}{i\eta(i\eta-1)}d\eta\\
&=&e^{-2rT}e^{rt}\frac{1}{2\pi}\int_{\Im\xi=\om-1}e^{ix \eta-t\psi(\eta+i)}\frac{K^{2-i\eta}}{(i\eta-2)(i\eta-1)}d\eta.
\eqast
Let $\lm<-2$ and choose $\om\in (\lm,-2)$ and $\sg>0$ so that $\Re(q-r+\psi(\xi))>0$ if $\Re q\ge \sg$ and $\Im\xi\in [\om,-2)$.
 Then 
the double Laplace-Fourier transform of  $G(t,x)$
w.r.t. $(t,x)$ is well-defined in the region $\{(q,\eta)\ |\ \Re q\ge \sg, \Im\eta\in [\om,-2)\}$, and it is given by
\[
\hG(q,\eta)=e^{-2rT}\int_0^{+\infty}e^{-(q-r)t}\frac{K^{2-i\eta}e^{-t\psi(\eta+i)}}{(i\eta-2)(i\eta-1)}dt.
\]
Thus, $\hG(q,\eta)=e^{-ia\eta}\hG_0(q,\eta)$, where $a=\ln K$ and
\[
\hG_0(q,\eta)=\frac{K^2e^{-2rT}}{(q-r+\psi(\eta+i))(i\eta-2)(i\eta-1)}.
\]
If the strike is at or above the barrier,
the remaining steps  are essentially the same as in Subsections \ref{oneEuroabovecall}-\ref{oneEuroaboveput}. The poles of the integrands are at $-2i$ and $-i$
rather than $-i$ and $0$, and the contour $\cL^-$ must be below $-2i$ rather than $-i$.
Also, due to a different factor $q-r+\psi(\eta+i)$ in the denominator, we have an additional restriction on $\sg$ and $\cL^{--}$.
Furthermore, instead of the equality $\phimq(\eta)(q+\psi(\eta))=q/\phipq(\eta)$ we have a more complicated equality
\[
\phimq(\eta)(q-r+\psi(\eta))=q(q-r+\psi(\eta))/(\phipq(\eta)(q+\psi(\eta)),
\]
hence, some of the poles and the corresponding residue terms are different.

In the case of the put, the calculations are the same only $\om'\in (0,\om)$ must be chosen at the first step,
and in the process of the deformation of the contours of integration w.r.t. $\eta$ down, tho poles at $\eta=-i,-2i$
are crossed rather than at $\eta=0,-i$.

 If the strike is below the barrier,
the above argument is modified similarly to the modification in Subsection \ref{Eurocallstrikebelow}.

 \subsection{The case of the product of two European call or put options}\label{twoEurcall}
 \subsubsection{General formulas}\label{twocallprelim} 
 {\sc In the case of calls}, we have 
 \[
 G(t,x)=e^{-2rT}e^{2rt} V_{\mathrm{call}}(T,K_1;T-t,x)V_{\mathrm{call}}(T,K_2;T-t,x),\] where $K_j\ge H$.
 Set $a_j=\ln K_j$.
 Assuming that 
$\lm<-2$, we take $\om_1,\om_2\in (\lm+1,-1)$, and calculate, first, the Laplace transform of
\beqast
G_1(t,x)&:=&e^{2rt} V_{\mathrm{call}}(T,K_1;T-t,x)V_{\mathrm{call}}(T,K_2;T-t,x)\\
&=&\frac{1}{(2\pi)^2}\int_{\Im\eta_1=\om_1}d\eta_1 e^{ix\eta_1-t\psi(\eta_1)}\frac{K_1^{1-i\eta_1}}{i\eta_1(i\eta_1-1)}
\int_{\Im\eta_2=\om_2}d\eta_2 e^{ix\eta_2-t\psi(\eta_2)}\frac{K_2^{1-i\eta_1}}{i\eta_2(i\eta_2-1)}.
\eqast 
 Using the Fubini theorem, we have
\beqast
\tilde G_1(q,x)&=&\frac{K_1K_2}{(2\pi)^2}\int_{\Im\eta_1=\om_1}\int_{\Im\eta_2=\om_2}
 \frac{e^{i(\eta_1+\eta_2)x}e^{-ia_1\eta_1-ia_2\eta_2}d\eta_1\, d\eta_2}
{(q+\psi(\eta_1)+\psi(\eta_2))i\eta_1(i\eta_1-1)i\eta_2(i\eta_2-1)}.
\eqast
Since $(\cEmq)^{-1}e^{i(\eta_1+\eta_2)x}=\phimq(\eta_1+\eta_2)^{-1}e^{i(\eta_1+\eta_2)x}$, 
we derive
\beqast
(\cEmq)^{-1}\tilde G_1(q,x)&=&\frac{K_1K_2}{(2\pi)^2}\int_{\Im\eta_1=\om_1}\int_{\Im\eta_2=\om_2}
 \frac{e^{i(\eta_1+\eta_2)x}e^{-ia_1\eta_1-ia_2\eta_2}d\eta_1\, d\eta_2}
{\phimq(\eta_1+\eta_2)(q+\psi(\eta_1)+\psi(\eta_2))\eta_1(\eta_1+i)\eta_2(\eta_2+i)}.
\eqast
Next, the Fourier transform of $\bfo_{(-\infty,h]}(x)e^{i\be x}$ is well-defined in the half-plane $\Im\xi>-\be$
by $e^{(\be-i\xi)}/(\be-i\xi)$, therefore, taking $\om\in (0,-\bemq)$, we can represent
\[
\tV_1(G;q,x):=\cEmq \bfo_{(-\infty,h]}(\cEmq)^{-1}\tilde G_1(q,x)
\]
 in the form
\[
\tV_1(G;q,x)=
\frac{K_1K_2}{(2\pi)^3}\int_{\Im\xi=\om}d\xi\,e^{i(x-h)\xi}\phimq(\xi)\int_{\Im\eta=\om_1}
\int_{\Im\eta=\om_2}\frac{e^{-i(a_1-h)\eta_1-i(a_2-h)\eta_2}d\eta_1d\eta_2}{\Phi(q,\xi;\eta_1,\eta_2)\eta_1(\eta_1+i)\eta_2(\eta_2+i)},
\]
where
$
\Phi(q,\xi;\eta_1,\eta_2)=\phimq(\eta_1+\eta_2)i(\eta_1+\eta_2-\xi)(q+\psi(\eta_1)+\psi(\eta_2)).
$

Since $x-h>0$, we deform the outer contour into a contour of the form $\cL^{++}$, crossing the pole at $\xi=-i\bemq$ if it exists:
\beqa\label{FTsolLFCC2b}
&&\tV_1(G;q,x)\\\nonumber&=&
\frac{K_1K_2}{(2\pi)^3}\int_{\cL^{++}}d\xi\,e^{i(x-h)\xi}\phimq(\xi)\int_{\Im\eta_1=\om_1}
\int_{\Im\eta_2=\om_2}
\frac{e^{-i(a_1-h)\eta_1-i(a_2-h)\eta_2}d\eta_1d\eta_2}{\Phi(q,\xi;\eta_1,\eta_2)\eta_1(\eta_1+i)\eta_2(\eta_2+i)},
\\\nonumber&&
+\frac{K_1K_2}{(2\pi)^2}\frac{iqe^{\bemq(x-h)}}{\phipq(-i\bemq)\psi'(-i\bemq)}\int_{\Im\eta_1=\om_1}
\int_{\Im\eta_2=\om_2}
\frac{e^{-i(a_1-h)\eta_1-i(a_2-h)\eta_2}d\eta_1d\eta_2}{\Phi(q,-i\bemq;\eta_1,\eta_2)\eta_1(\eta_1+i)\eta_2(\eta_2+i)}.
\eqa
\vskip0.1cm\noindent
{\sc If the first option is a call and the other is a put},
then $\om_1\in (-\bepq,-1)$ and $\om_2\in (0,-\bemq)$ should satisfy $\om_{1+}-b_+\sin(\om_+)-(\om_1+\om_2)>0$,
where $\om_{1+},\om_+, b_+$ are the parameters that define the curve $\cL_{++}$ (recall that the lowest point of the
latter is 
$i(\om_{1+}-b_+\sin(\om_+))$).
\vskip0.1cm\noindent
{\sc In the case of two put options}, the calculations are the same but we take $\om_1,\om_2\in (0,
-\bemq)$, and \eqref{FTsolLFCC2b} can be justified if $\om_{1+}-b_+\sin(\om_+)-(\om_1+\om_2)>0$.
\subsubsection{Reductions}\label{reductions} All three cases are reducible one to another. We start
with the product of two call options. 
We move the lines of integration $\{\Im\eta_j=\om_j\}$, $\om_j\in (-\bepq,-1)$, $j=1,2$, up, and, on crossing the
poles, apply the residue theorem. Let $\om'_1, \om'_2>0$,
$\om'_1+\om'_2<\Im\xi$. Then, moving the innermost line of integration up, we obtain
 \beqa\label{call_put}
&&\frac{1}{(2\pi)^2} \int_{\Im\eta_1=\om_1}
\int_{\Im\eta_2=\om_2}
\frac{e^{-i(a_1-h)\eta_1-i(a_2-h)\eta_2}d\eta_1d\eta_2}{\Phi(q,\xi;\eta_1,\eta_2)\eta_1(\eta_1+i)\eta_2(\eta_2+i)}\\\nonumber
&=&
\frac{1}{(2\pi)^2}\int_{\Im\eta_1=\om_1}
\int_{\Im\eta_2=\om'_2}
\frac{e^{-i(a_1-h)\eta_1-i(a_2-h)\eta_2}d\eta_1d\eta_2}{\Phi(q,\xi;\eta_1,\eta_2)\eta_1(\eta_1+i)\eta_2(\eta_2+i)}\\\nonumber
&&+\frac{1}{2\pi} \int_{\Im\eta_1=\om_1}\frac{e^{-i(a_1-h)\eta_1}d\eta_1}{\Phi(q,\xi;\eta_1,0)\eta_1(\eta_1+i)}
-\frac{e^{h-a_2}}{2\pi} \int_{\Im\eta_1=\om_1}\frac{e^{-i(a_1-h)\eta_1}d\eta_1}{\Phi(q,\xi;\eta_1,-i)\eta_1(\eta_1+i)}
\eqa
The first term on the RHS is the integral for the case of a call and put, the remaining two terms can be calculated similarly
to the integrals in formulas for put and call options. If $a_1-h\le 0$, the line $\{\Im \eta_1=\om_1\}$  is deformed into 
a contour of the form $\cL^-$, and if $a_1<h$, then of the form $\cL^+$ or $\cL^{++}$. The new contour must be  strictly below
the contour of integration $\cL^{++}$ w.r.t. $\xi$, and the angles between the asymptotes of the two contour must be non-zero.

Pushing the lines of integration w.r.t. $\eta_1$ up, we obtain a repeated integral which arises when the product of puts
is considered (plus 4 one-dimensional integrals):
\beqa\label{put_put}
&&\frac{1}{(2\pi)^2} \int_{\Im\eta_1=\om_1}
\int_{\Im\eta_2=\om_2}
\frac{e^{-i(a_1-h)\eta_1-i(a_2-h)\eta_2}d\eta_1d\eta_2}{\Phi(q,\xi;\eta_1,\eta_2)\eta_1(\eta_1+i)\eta_2(\eta_2+i)}\\\nonumber
&=&
\frac{1}{(2\pi)^2}\int_{\Im\eta_1=\om'_1}
\int_{\Im\eta_2=\om'_2}
\frac{e^{-i(a_1-h)\eta_1-i(a_2-h)\eta_2}d\eta_1d\eta_2}{\Phi(q,\xi;\eta_1,\eta_2)\eta_1(\eta_1+i)\eta_2(\eta_2+i)}\\\nonumber
&&+\frac{1}{2\pi} \int_{\Im\eta_1=\om_1}\frac{e^{-i(a_1-h)\eta_1}d\eta_1}{\Phi(q,\xi;\eta_1,0)\eta_1(\eta_1+i)}
-\frac{e^{h-a_2}}{2\pi} \int_{\Im\eta_1=\om_1}\frac{e^{-i(a_1-h)\eta_1}d\eta_1}{\Phi(q,\xi;\eta_1,-i)\eta_1(\eta_1+i)}\\\nonumber
&&+\frac{1}{2\pi} \int_{\Im\eta_2=\om'_2}\frac{e^{-i(a_2-h)\eta_2}d\eta_2}{\Phi(q,\xi;0,\eta_2)\eta_2(\eta_2+i)}
-\frac{e^{h-a_1}}{2\pi} \int_{\Im\eta_2=\om'_2}\frac{e^{-i(a_2-h)\eta_2}d\eta_2}{\Phi(q,\xi;-i,\eta_2)\eta_2(\eta_2+i)}.
%\\\nonumber
%&&+\frac{1}{\Phi(q,\xi;0,0)}-\frac{e^{h-a_1}}{\Phi(q,\xi;-i,0)}-\frac{e^{h-a_2}}{\Phi(q,\xi;0,-i)}+\frac{e^{2h-a_1-a_2}}{\Phi(q,\xi;-i,-i)}.
\eqa
Below, we consider the calculation of the repeated integrals above, denote them
$J(\xi;\om_1,\om_2)$, for strikes below or above the barrier. Depending on a case,
it will be convenient to calculate  either $J(\xi;\om_1,\om_2)$ or $J(\xi;\om_1,\om'_2)$ or $J(\xi;\om'_1,\om'_2)$,
and, then, if necessary, use \eqref{call_put} and \eqref{put_put}.  

Since the choice of the parameters of the contours of integration are simpler if the contours of integration
w.r.t. $\eta_1, \eta_2$ are in the lower half-plane, we would recommend the reduction to the case of two calls unless
$\bepq-1$ is much smaller than $-\bemq$.

 \subsubsection{Both strikes are at or above the barrier}\label{twocallabove} 
% In this case, we calculate the repeated integral for the case of the product of two call options.
Since $a_j\ge h$, the contours of integration w.r.t. $\eta_1$ and $\eta_2$ are deformed down,
into contours of the form $\cL^-=\cL(\om_{1j},\om_j,b_j)$, where $\om_j<0$ are different from the ones 
in \eqref{FTsolLFCC2b}. In the process of the deformation,
neither of these two contours may cross the cut $i(-\infty,\lm]$. In addition, the factor $q+\psi(\eta_1)+\psi(\eta_2)$ must remain separated from 0 for all $\eta_1,\eta_2$ on the contours. 
A simple general requirement (sufficient for large $\Re q$; if $\Re q>0$ 
is not large, an additional control of the behaviour of the contours in the process of the deformation is needed)
is as follows.
By Assumption $(X)$, $\psi$ is of the form $\psi(\eta)=-i\mu^0\eta+\psi^0(\eta)$, where $\mu\in\bR$ and $\psi^0(\eta)\sim c^0_\infty |\eta|^\nu e^{i\nu\ga}$,
where $\ga=\mathrm{arg}\,\eta\in (-\pi/2,\pi/2)$, and $c^0_\infty>0$. Hence,
we can choose the sinh-deformation with the desired properties in the following cases
\begin{enumerate}[(1)]
\item
if $\nu\in (1,2]$ or $\mu\le 0$ then $0>\om_j>-\min\{\pi/2,\pi/(2\nu)\}$;
\item
if $\nu=1$ and $\mu>0$, then $0>\om_j>-\arctan(c^0_\infty/\mu)$.
\end{enumerate}
If $\nu\in (0,1)$ and $\mu>0$, then only less efficient conformal deformations and changes of variables
are possible (labeled {\em sub-polynomial} in \cite{ConfAccelerationStable} where efficient methods for evaluation of stable
distributions have been developed). Keeping the notation $\cL^-$ for this case as well, we
obtain
\beqa\label{FTsolLFCC3}
&&\tV_1(G;q,x)\\\nonumber&=&
\frac{K_1K_2}{(2\pi)^3}\int_{\cL^{++}}d\xi\,e^{i(x-h)\xi}\phimq(\xi)\int_{\cL^-}
\int_{\cL^-}
\frac{e^{-i(a_1-h)\eta_1-i(a_2-h)\eta_2}d\eta_1d\eta_2}{\Phi(q,\xi;\eta_1,\eta_2)\eta_1(\eta_1+i)\eta_2(\eta_2+i)},
\\\nonumber&&
+\frac{K_1K_2}{(2\pi)^2}\frac{iqe^{\bemq(x-h)}}{\phipq(-i\bemq)\psi'(-i\bemq)}\int_{\cL^-}
\int_{\cL^-}
\frac{e^{-i(a_1-h)\eta_1-i(a_2-h)\eta_2}d\eta_1d\eta_2}{\Phi(q,-i\bemq;\eta_1,\eta_2)\eta_1(\eta_1+i)\eta_2(\eta_2+i)}.
\eqa
If the characteristic exponent is a rational function, then we may push the contour 
$\cL^{++}$ to infinity upwards, and $\cL^-$ downwards. After all the poles are crossed, instead of the triple integral, 
we will obtain a triple sum expressible in terms of $a_1, a_2, h$,
the parameters of the characteristic exponent and its roots and poles.

The same can be done in the beta-model \cite{beta}; the resulting sums will be infinite, though, and 
one will have to solve a rather non-trivial problem of a sufficiently accurate truncation of these infinite sums.

\subsubsection{The case when one of the strikes is below the barrier}\label{above_below} 
Let $K_2<H\le K_1$. Then $a_2-h<0\le a_1$, and the sinh-deformation of the line of integration w.r.t. $\eta_2$ is impossible,
 hence, we apply the simplified trapezoid rule
to the initial integral w.r.t. $\eta_2$ (flat iFT). The integrands decay very slowly, hence, the number of terms $N_2$ in the simplified trapezoid rule
is too large. 
The number $N_2$ can be significantly decreased using the summation by parts in the infinite trapezoid rule
\[
I=\ze\sum_{j\in \bZ} e^{-ia\ze j}f_j.
\]
if $f'(y)$ decreases faster than $f(y)$ as $y\to\pm\infty$ (as is the case in the setting of
the present paper). Indeed, then  the finite differences $\De f_j=f_{j+1}-f_j$ decay faster 
than $f_j$ as $j\to\pm\infty$ as well. The summation by parts formula is as follows. We choose $\ze>0$ so that
$1-e^{-ia\ze}$ is not close to 0. Then
\[
\ze\sum_{j\in \bZ} e^{-ia\ze j}f_j=\frac{\ze}{1-e^{-ia\ze}}\sum_{j\in \bZ} e^{-ia\ze j}\De f_j.
\]
If each differentiation increases the rate of decay, then the summation by part procedure can be iterated.
In the setting of the present paper, the rate of decay increases by approximately 1 with each differentiation. In
the numerical examples, we apply the summation by parts 3 times, which decreases the number of terms many times,
and makes the number comparable to the number of terms when the sinh-acceleration can be applied.  
\subsubsection{The case when both strikes are below the barrier}\label{both_below} 
In this case, the sinh-acceleration in the integrals w.r.t. $\eta_1$ and $\eta_2$ is impossible, hence, 
we apply the simplified trapezoid rule
to the initial integrals w.r.t. $\eta_2$ and $\eta_1$.
The numbers $N_1, N_2$ can be significantly decreased using the summation by parts in the infinite trapezoid rule
 applied to evaluate
the repeated integrals on the RHS of \eqref{FTsolLFCC3}.

 \section{Semi-static hedging vs Variance minimizing hedging of  down-and-in options: 
 a numerical example and qualitative analysis}\label{numer_din_example}
  
 In this section, we present and discuss in detail several important observations
 and practically important conclusions using an example of a down-and-in call option. The process (KoBoL)
 is the same  
 as in Table~\ref{StVar1_1.2}
, maturity is $T=0.1$, but the strike $K=1.04$ is farther from the barrier than in Table~\ref{StVar1_1.2}
. The reason for that is two-fold: 1) to show
 that if the restrictive formal conditions for the semi-static hedging are satisfied, then
 the semi-static procedure works reasonably well for jump-processes with a rather slowly-decaying jump component
 even if the distance from the barrier to the support of the artificial exotic payoff is sizable (about 4 percent); 2) if this support 
 is too close to the barrier, then the summation-by-part procedure can be insufficiently accurate unless high precision arithmetic is used.
 In the paper, we do calculation with double precision.
 
We consider the standard situation: an agent sells the down-and-in call option, and invests the proceeds into
 the riskless bond. We assume that the spot $S_0=e^{0.04}$ is almost at the strike. That is, the agent makes the bet that  the barrier will not be breached during the lifetime of the option.

 Fig.~\ref{CpdfFtau} shows that the probability of this happy outcome is not very large; even the probability that
 the barrier will be breached before $\tau=0.05$ is about 42\%. 
 
% \subsection{What happens when the barrier is breached during the life time of the option}
 However, with the probability about  50\%, at the time the portfolio
 is breached, the portfolio value is positive. This is partially due to the fact the bond component in the portfolio
 increases fairly fast.\footnote{Recall that
 even a small asymmetry of the jump component requires the riskless rate be sizable in order that
 the semi-static procedure be formally justified.}
 Nevertheless, if the barrier is breached at time close to 0, the loss in the portfolio value can be
 rather sizable (see Fig.~\ref{VCall}). Hence, it is natural for the agent to hedge the bet. 
 
 Assume that the agent uses the semi-static hedging constructed in the paper. The standard static and
 semi-static arguments construct model- and spot-invariant portfolios, which make it impossible to take into
 account that a hedging portfolio needs to be financed. If we consider the portfolio of 3 put options, and ignore
 the riskless bonds borrowed to finance the position, then, at any time $\tau$ the barrier is breached, and at any level
 $S_\tau\le H$, the portfolio value is positive or very close to zero (Fig.~\ref{V0SS3}). Thus, the semi-static hedge seems to work
 very well even if the portfolio consists of only 3 options. Furthermore, if only one option is used, then the portfolio value decreases except in a relatively small region far from the barrier and close to maturity (Fig.~\ref{DifV0}). Since the probability of the option expiring in that region is very small, a naive argument would suggest that increasing the number of options in the hedging portfolio would increased the overall hedging performance of the portfolio. Recall, however, that the agent borrows riskless bonds to finance the put option position. If the barrier is breached, this short position in the riskless bond has to be liquidated alongside the other positions in the portfolio, complicating the overall picture. Fig.~\ref{VSS3} demonstrates that, when $S_\tau$ is close to the barrier (a high probability event conditional on breaching the barrier), the value of the hedging portfolio is negative and large. Thus, the hedge is far from perfect. Furthermore, in Fig.~\ref{DifV} we see that if $\tau>0.037$ and $S_\tau$ is not too far
 from the barrier, then the semi-static portfolio consisting only of a put option whose strike is at the kink of the exotic option has a larger value then the value of the semi-static portfolio
 consisting of 3 put options. 
 The value of the hedging portfolio with one put option is shown in Fig.~\ref{VSS1}.
 Table~\ref{TableBets} shows that, if the barrier is not breached during the lifetime of the options, both portfolios have negative value but the losses on the one with 3 options is twice as large.
 Hence, it is better to use fewer options unless the agent is betting on a very low probability event realizing: if a larger fraction of wealth is invested in options, the cost of these options is higher, eroding the advantage of a more accurate semi-static hedging.

 Fig.~\ref{VP3}
 demonstrates that the variance-minimizing portfolio of the same three put options hedges the risk of the
 down-and-in option, but only partially: if $\tau<0.06$ (approximately), and $S_\tau$ is not far from the barrier,
 the value of the portfolio is negative; for $\tau$ closer to maturity, the portfolio value becomes sizably positive. 
 
 Comparing Fig.~\ref{VFTP2} and Fig.~\ref{VP3}, we observe that the inclusion of the first-touch digital increases the hedging performance
 of the portfolio somewhat. Fig.~\ref{VK1}, \ref{VFT} and \ref{VKeq1} show that portfolios with one put option or one first touch digital perform
 approximately as well as portfolios with 3 options, and the portfolio with the first touch digital as the only hedging instrument is better than portfolios with one
 put option (with strikes either at the barrier or at the kink).
  
  We emphasize once again that, since the hedging portfolio is not costless, proper analysis of the efficiency of hedging should incorporate both the payoff to the short position in the riskless bond and the value of the hedging portfolio at maturity in case the barrier is not breached. In the case of the formally perfect semi-static hedging with 3 put options, 
 {\em the loss of the portfolio when the barrier is not breached} is quite sizable (about 0.6 of the value of the hedged
 option at time 0); in contrast, if only one option is used, then the loss is about 0.3 of the value of the hedged option.
 Instead, if we use a variance-minimizing portfolio with 3 or 1 put options, the gain when the barrier is not breached is
 sizable (about the value of the hedged option); if the first touch digital is used, then the gain is close to 0
 (positive if, in addition to the first touch digital, two put options are used, and negative if only the first touch digital is used).

 To sum up: any realistic hedging portfolio replaces the initial non-hedged bet with another one,
 which can be more risky than the initial bet, {\em once we take the investment in the riskless bond into
 account}. In particular, the semi-static hedging portfolio with 3 put options, which seems perfect  from the point of view of
 the replication of the payoff of the down-and-in call option, ignoring the upfront payments for the hedging instruments, is, in fact, another and  much riskier bet: only if the barrier is breached, and the underlying
 is sizably below the barrier at that moment, the hedging portfolio is in the black; otherwise, it is is the red,
 and, with the large probability, significantly so.  Thus,  semi-static hedging portfolios can 
 be regarded as contrarian bets: if the realization of the underlying is very bad (which occurs with small probability), the portfolio gains a lot, but looses a lot otherwise. The portfolios with first touch digitals have the most concentrated 
 profile of the payoffs. 
 
Table~\ref{TableBets} illustrates the bet structure implicit in different hedging portfolios, showing the approximate payoffs in the vicinity of the barrier when the barrier is breached and the payoff at maturity if the barrier is not breached.\footnote{Payoffs in the vicinity of the barrier when the barrier is breached are shown for several small time intervals. For omitted time intervals, approximate probabilities
 and payoffs can be reconstructed using the interpolation; the probabilities of realizations of $S_\tau$
 farther from the barrier are small because the tails of the L\'evy density decay exponentially, and the rate of decay
 is not small. }
 We see that the replacement of the initial bet (the naked short down-and-in call option)
 with the portfolios based on the semi-static argument may lead to a sizable loss
 with probability more than 90\%. At the same time, there is a small probability
 of a significant gain, if the barrier is breached by a large jump. In our opinion, Table~\ref{TableBets} demonstrates that the most efficient
 hedge is with the first touch digital.
  
  In Table~\ref{TableStDEv}, we show the normalized standard deviation of different portfolios, computed at 1\%-7\% from the barrier. The irregular structure of the payoff to the semi-static hedging portfolio implies that the portfolio volatility must be high,
as demonstrated in Table~\ref{TableStDEv}. The volatilities of other portfolios are of the same order of magnitude, although the
 volatilities of portfolios that contain the first touch digital are sizably smaller, which is an additional indication of the
 advantages of the first touch digitals.
 
 In Table~\ref{TableVarCov}, we show the variance-covariance matrix which is used to construct several hedging portfolios; 
 the reader may check that the matrix is rather close to a matrix of a smaller rank 2-3 rather than 5.
 This explains why the variances of different portfolios that we construct are very close, and why, as far as the hedging
 of small fluctuations is concerned, there is no gain in using several options; and the cost of using several options can be uncomfortably high.

\section{Conclusion}\label{concl}

\subsection{Hedging: results and extensions} 

We developed new methods for constructing static hedging portfolios for European exotics and
variance-minimizing hedging portfolios for European exotics and barrier options in L\'evy models;
in both cases, the calculations are in the dual space. In particular, we constructed approximate static portfolios for
exotic options approximating an exotic payoff with  linear combinations of vanillas in the norm of the 
H\"older space with an appropriate weight; the order of the H\"older space is chosen so that
the space of continuous functions with the same weight is continuously embedded in the weighted H\"older space,
hence, we obtain an approximation in the $C$-norm. The weights are easy to calculate because the weighted H\'older
space is the Hilbert space, and the scalar products of the elements of this space can be easily calculated evaluating integrals
in the dual space. 

We have discussed the limitations of the static hedging/replication of barrier options, and, in applications
to L\'evy models, listed the rather serious  restrictions on the parameters of the model under which the approximate replication of a barrier option with an appropriate European exotic option can be justified.
We explained why in the presence of jumps a perfect semi-static hedging is impossible, and, using an example of the down-and-in call option in KoBoL model, demonstrated that
\begin{enumerate}[(a)]
\item
the formal semi-static procedure of hedging barrier options based on the approximation 
of the latter with exotic European options and the static hedging algorithm of the present paper
produce a good super-replicating portfolio even if only 3 put options are used;
\item
however, if the borrowed riskless bond needed to finance the hedging portfolio is taken into account,
then the hedging portfolio of the short down-and-in call option, 3 put options and riskless bond has the payoff structure
resembling contrarian bets. With a very high probability, the portfolio suffers sizable losses (several dozen percent and more
of the value of the down-and-in option at the initiation). This happens if the barrier is not breached during the lifetime
of the option or, at the moment of breaching, the spot does not jump too far below the barrier. Only if the jump down is large,
the hedging portfolio will turn out the profit, and the profit can be very large;
\item
if only one put option is used for hedging, then the performance of the portfolio significantly improves although
the super-replicating property becomes imperfect, naturally. This observation undermines a general idea behind the
semi-static hedging/replication. An accurate replication needs more options in the hedging/replicating portfolio but the
associated costs may outweigh the formal advantage of the model-free replication/hedging;
\item
in the example that we considered, variance-minimizing hedging portfolios are less risky than the semi-static hedging portfolios,
and the ones with the first touch digitals are the best. This observation implies that, surprisingly enough, for barrier options, 
the variance-minimizing  objective leads to smaller losses even at the scale which is not expected to be characterized by the variance;
\item in the case of the short down-and-out call option, the natural hedging portfolio is a short position in a European call; similarly, the natural hedging portfolio for the down-and-in option is a long position in a European call. Fig.~\ref{doVCall}, \ref{doV0SS3},
\ref{doVSS3}
 show the (normalized) payoffs
of the naked down-and-out option, the semi-static hedging portfolio without taking the bond component into account,
and the semi-static hedging portfolio with the bond component taken into account. While the former suffers losses unless the breach happens exactly at the barrier, the payoff of the latter is 0.3 at maturity
(if no breach happens), and positive in a small neighborhood of the barrier if the breach happens. Thus, if the bond component is taken into account and the non-negligible probability of a large loss is ignored,
the semi-static portfolio almost looks like a good hedge.
\end{enumerate} 
The results the paper suggest that it might be natural to consider semi-static hedging portfolios as separate classes
of derivative securities, with non-trivial payoff structures, which {\em are model-dependent.}  The numerical examples in the paper
demonstrate that the properties of the payoffs of semi-static hedging  portfolios for short down-and-out options are fundamentally different
from the properties of semi-static hedging  portfolios for short down-and-in options. For the process and down-options considered in the paper,
the semi-static hedging portfolios for down-and-out options do have properties close to the properties of good hedging portfolios:
with high probability, the payoff at expiry (or maturity) is positive although with small probability the payoff is negative. However, for the down-and-out options,
the properties are the opposite: small losses with high probability and large gains with small probability (essentially, contrarian bets).
For up-options under the same process, the semi-static hedging portfolios for ``out" options are good but 
the ones for ``in" options are contrarian bets. The properties change with the change of the sign of the drift as well. For barrier options, the picture
becomes even more involved.

%These observations suggest the following extensions of the hedging method of the paper, which we reserve for
%the future.
%\medbreak
To conclude,  the formal semi-static argument is applicable only under rather serious restrictions, cannot be exact in the presence
of jumps, may lead to very risky portfolios, and the hedging errors of the variance-minimizing portfolios can be sizable (although less than the errors of the formal semi-static hedging portfolios). The deficiencies of 
both types of the hedging portfolios stem from the fact that both are static. If we agree that 
a model-independent hedging is seriously flawed (the semi-static hedging disregards the cost of hedging), and
the variance minimizing one does not take into account the payoff of the portfolio at the time of breaching explicitly,
then a natural alternative is a hedging portfolio which is rebalanced after reasonable short time interval so that
the profile of the payoff at an uncertain moment of breaching is approximately equal to the profile at the
moment of the next rebalancing. The portfolio can be calculated using the approximation in the H\"older space norm,
and the following versions seems to be natural:
\begin{enumerate}[I.]
\item
Myopic hedging, when the hedging instruments expire at the moment of the next rebalancing.
\item
Quasi-static hedging, when the first hedging portfolio is constructed using the options of the same maturity
as the barrier option, and the other options are added to the existing portfolio at each rebalancing moment.
\item
Hedging using first touch digitals only. Possible versions: each period, 
buy the digital which expires at the end of the rebalancing period; each period, buy first touch digitals which expire at
the maturity date.
\end{enumerate}
As Fig.1 suggests, and our numerical experiments confirm, the use of the first touch digitals has certain advantages,
 and the first-touch options with the payoffs $(S/H)^\gamma, \gamma>0,$ would be even better hedging instruments.
 
 We leave the study of these versions of hedging to the future.

\subsection{The technique used in the paper to price options with barrier features and its natural extensions}\label{asympform}
The construction of the variance minimizing portfolio for barrier options 
is based on the novel numerical methods for evaluation of the Wiener-Hopf factors and pricing
barrier options that are of a more general interest than applications to hedging. In particular, the Wiener-Hopf factors
can be evaluated with the relative error less than E-12 and barrier options with the relative error less than 
E-08 in a fraction of a msec. and several msec., respectively. The efficient evaluation
of the Wiener-Hopf factors, and, in many cases, the numerical realization of the inverse Fourier transform in
the option pricing formulas are based on the sinh-acceleration method of evaluation of integrals of
wide classes, highly oscillatory ones especially, developed in \cite{SINHregular}. In some cases, the sinh-acceleration method is not applicable. These cases are the ones when the standard Fourier inversion techniques may require the summation
of millions of terms and more; the same problem arises when the Hilbert transform method is applied in L\'evy models
of finite variation. In the present paper, we suggested and successfully applied the summation by parts trick in the infinite
trapezoid rule. In the result, the rate of the convergence of the infinite sum significantly increases, and it suffices to
add thousands of terms
in the infinite trapezoid rule instead of millions. Note that the same summation by parts can be applied when the Hilbert transform is used
to evaluate discrete barrier options, and in other cases.

%This is the case
%in the double-exponential jump diffusion model (DEJD model) used in \cite{lipton-columbia,lipton-risk,kou,KW1,KW2,sepp,lipton-sepp},
 %and its generalization, the hyper-exponential jump-diffusion model
% (HEJD model) outlined in \cite{lipton-risk}, studied in detail in \cite{amer-put-levy}, and used later
% in  \cite{JP, carr-crosby,hejd,MSdouble} and many other papers. In the case of the beta-model \cite{beta}, similar formulas with series %instead of sums
%can be derived; but it is unclear how to truncated repeated series.

Similarly to \cite{paired}, it is possible to derive simplified asymptotic formulas, which are fairly accurate
if the spot is not very close to the barrier and the tails of the jump density decay fast. Similar {\em exact} formulas with
larger number of terms are valid in models with rational exponents $\psi$, if the roots and poles of the characteristic equation $q+\psi(\xi)=0$ can be efficiently calculated, and the number of zeros and poles is not large. This is the case
in the double-exponential jump diffusion model (DEJD model \cite{kou}) 
used in \cite{lipton-columbia,lipton-risk,kou,KW1,KW2,sepp,lipton-sepp},
 and its generalization, the hyper-exponential jump-diffusion model
 (HEJD model) introduced and studied in detail in \cite{amer-put-levy} (and independently outlined in \cite{lipton-risk}), and used later
 in  \cite{JP, carr-crosby,hejd,MSdouble} and many other papers.  In the case of the beta-model \cite{beta}, similar formulas with series instead of finite sums
can be derived; but it is unclear how to truncate repeated series in order to satisfy the desired error tolerance.

The approach of the paper can be applied to construct hedging portfolios for lookbacks \cite{paraLaplace}, American options, barrier options and Asians with discrete monitoring \cite{MarcoDiscBarr}, Bermudas, where the Fourier transform of the option price can be efficiently calculated. 
To this end, one can reformulate the backward induction procedures in the discrete time models or Carr's randomization method
making the calculations in the dual space as in \cite{AsianGammaSIAMFM}, where Asian options were priced making calculations in the dual space.
The double-spiral method introduced in \cite{AsianGammaSIAMFM}, together with the sinh-acceleration, should be used to decrease
the sizes of arrays one works with at each time step. The sinh-acceleration allows one to make calculations fast in 
regime switching models as well (calculation of the Wiener-Hopf factors in the matrix case is possible using the sinh-acceleration
technique),
hence, approximations of models with stochastic interest rates and stochastic volatility can be considered, similarly to 
\cite{stoch-int-rate-CF,amer-reg-sw-SIAM,BLHestonStIR08}, where American options were priced using Carr's randomization in the state space.
The double-barrier options can be treated by reformulating the method of \cite{BLdouble} in the dual space.

\appendix

\section{Standard semi-static hedging under L\'evy processes}\label{SemiStaticLevy}
We consider the down-and-in option (with the value function) $V(G;H; t, X_t)$ with the barrier $H=e^h$. Let $\tau:=\tau_h$ be the first entrance time
by the L\'evy process $X$ with the characteristic exponent $\psi$ into $(-\infty,h]$. If $\tau<T$, then, at time $\tau$,
the option becomes the European option with the payoff $G(X_T)$ at maturity date $T$. The standard semi-static 
hedge is based on the assumption that there exists $\be\in\bR$ such that, for any stopping time $\tau$,
\bbe\label{semistatcondLevy}
\bE_\tau\left[G(X_T)\bfo_{X_T>h}\right]=\bE_\tau\left[e^{\be(X_T-X_\tau)}G(2X_\tau-X_T)\bfo_{2X_\tau-X_T>h}\right].
\ee
If \eqref{semistatcondLevy} holds, we consider the European option $V(G_{ex}; t, X_t)$ of maturity $T$, with the payoff function
\bbe\label{Gex}
G_{ex}(x)=(G(x)+e^{\be(x-h)}G(2x-h))\bfo_{(-\infty,h]})(x).  
\ee
If $\tau> T,$ the down-and-in option and the European option
expire worthless. At time $\tau$, the option values coincide on the strength of \eqref{semistatcondLevy} {\em provided} $X_\tau=h$ because 
\[
\bE_\tau\left[G(X_T)\bfo_{X_T>h}\right]+\bE_\tau\left[G(X_T)\bfo_{X_T\le h}\right]=\bE_\tau[G(X_T)].
\]
But the assumption $X_\tau=h$ means that there are no jumps down. In a moment, we will show that then there are no jumps up as well.
Let $G(x)=(e^x-K)_+$ or, more generally, let $\hG(\xi)$ be well-defined in the half-plane $\{\Im\xi<-1\}$ and decay as $|\xi|^{-2}$ as $\xi\to \infty$
remaining in this half-plane. Let $\psi$ be analytic in a strip $S_{(\lm,\lp)}$, where $\lm<-1$. Take $\om\in (\lm,-1)$,
denote $x=X_\tau$, and
represent the LHS of \eqref{semistatcondLevy} in the form
\bbe\label{LHSstat}
\bE_\tau\left[G(X_T)\bfo_{X_T>h}\right]=\frac{1}{2\pi}\int_{\Im\xi
=\om}e^{ix\xi-(r+\psi(\xi))(T-\tau)}\widehat{G\bfo_{(h,+\infty)}}(\xi)d\xi.
\ee
The RHS of \eqref{semistatcondLevy} can be represented as the repeated integral 
\beqast
&& \frac{1}{2\pi}\int_{\Im\xi=\om'}d\xi\, e^{ix\xi-(r+\psi(\xi))(T-\tau)}\int_{\bR}dy\, e^{-iy\xi}e^{\be(y-x)}G(2x-y)\bfo_{2x-y>h}.
\eqast
if $\om'\in \bR$ can be chosen so that the repeated integral converges (this imposes an additional condition on $\psi$, which
will be made explicit in a moment). Changing the variable $2x-y=y'$, and then $-\xi-i\be=\xi'$, we obtain
\beqast
&&\bE_\tau\left[e^{\be(X_T-X_\tau)}G(2X_\tau-X_T)\bfo_{2X_\tau-X_T>h}\right]\\
&=&\frac{1}{2\pi}\int_{\Im\xi=\om'}d\xi\, e^{ix\xi-(r+\psi(\xi))(T-\tau)}\int_{\bR}dy'\, e^{-i(2x-y')\xi}e^{\be(x-y')}G(y')\bfo_{y'>h}\\
&=&\frac{1}{2\pi}\int_{\Im\xi=\om'}d\xi\, e^{-(r+\psi(\xi))(T-\tau)}\int_{\bR}dy'\, e^{i(x-y')(-\xi-i\be)}G(y')\bfo_{y'>h}\\
&=&\frac{1}{2\pi}\int_{\Im\xi'=-\om'-\be}d\xi'\, e^{-(r+\psi(-\xi'-i\be))(T-\tau)}\int_{\bR}dy'\, e^{i(x-y')\xi'}G(y')\bfo_{y'>h}\\
&=&\frac{1}{2\pi}\int_{\Im\xi=-\om'-\be} e^{ix\xi-(r+\psi(-\xi'-i\be))(T-\tau)}\widehat{G\bfo_{(h,+\infty)}}(\xi)d\xi.
\eqast
We see that $\om'$ and $\be$ must satisfy $-\om'-\be\in (\lm,-1)$, and, comparing with \eqref{LHSstat}, we see that
the characteristic exponent $\psi$ and $\be$ must satisfy
\bbe\label{conspsisemistat}
\psi(\xi)=\psi(-\xi-i\be).
\ee
Given a class of L\'evy processes, the equation \eqref{conspsisemistat} imposes one condition on the diffusion part,
and the second condition on the jump part; hence, the dimension of the admissible parameter space drops by 2 if
there are both diffusion and jump component, and by one if there is only one of these components.

If $X$ is the BM with drift $\mu$ and volatility $\sg$, then \eqref{conspsisemistat}
is equivalent to $\be=-\mu/(2\sg^2)$. Equivalently, $\psi(-i)=0$ (the stock is a martingale), hence, $\de=r$ (the dividend rate equals
the riskless rate).
In the case of the BM with embedded KoBoL component:
\[
\psi(\xi)=\frac{\sg^2}{2}\xi^2-i\mu \xi +c_-\Gamma(-\nu_+)(\lp^{\nu_+}-(\lp+i\xi)^{\nu_+})
+c_+\Gamma(-\nu_-)((-\lm)^{\nu_-}-(-\lm-i\xi)^{\nu_-}),
\]
the conditions become: 1) $c_+=c_-$ (hence, either there are no jumps or there are jumps in both directions),
and $\nu_+=\nu_-$; and 2) $\be=-\mu/(2\sg^2)=-\lp-\lm$. 3) $\psi(-i)+r-\de=0$.
We see that if there is no diffusion components, then the ``drift" $\mu=0$, and if there is a diffusion component,
and $\mu>0$ (resp., $\mu<0$), then $\lp>-\lm$ (resp., $\lp<-\lm$), which means that the density of jumps decays 
slower in the direction of the drift. 

We finish this section with a discussion of a possible size of hedging errors induced by the assumption
that the process does not cross the barrier by a jump when, in fact, it does. In the case of the down-and-in options
the expected size of the overshoot
decreases when $\lp$ increases; in the case of up-and-in options, $-\lm$ increases in absolute value.
Hence, if the diffusion component is sizable: $\sg^2>0$ is  not small, then the condition
$\be=-\mu/(2\sg^2)=-\lp-\lm$ implies a strong symmetry $\lp\approx -\lm$ of the positive and negative jump components.

If $\sg^2$ is small, then a strong asymmetry is possible; but then $\mu$ must be very large in absolute value.
To sum up: the conditions on the parameters of the model which allow one to formally apply the semi-static 
hedging procedure to L\'evy processes with jumps are rather restrictive.

%\section{Wiener-Hopf factorization}\label{WHF}
\section{Additional representations of the Wiener-Hopf factors}\label{WHFadd}
In the case of processes of order $\nu<1$ (that is, processes of finite variation), with non-zero
drift, formulas \eqref{phip1}-\eqref{phim1} for the Wiener-Hopf factors can be made more efficient for computational purposes. Set
\begin{equation}\label{defPsinule}
\Psi(q,\eta)=(1+\psi(\eta)/q)/(1-i\mu\eta/q).
\end{equation}
\mbr\noindent
(IV) If $\nu<1$ and $\mu>0$, then, for the same $\xi$ and $\om_\pm$ as above,
\begin{eqnarray}\label{phip1num1p}
\phi^+_q(\xi)&=&(1-i\mu\xi/q)^{-1}\exp\left[\frac{1}{2\pi i}\int_{\Im\eta=\om_-}
\frac{\xi \ln \Psi(q,\eta))}{\eta(\xi-\eta)}d\eta\right]\\\label{phip1num1m}
\phi^-_q(\xi)&=&\exp\left[-\frac{1}{2\pi i}\int_{\Im\eta=\om_+}\frac{\xi \ln \Psi(q,\eta)}{\eta(\xi-\eta)}d\eta\right].
\end{eqnarray}
\mbr\noindent
(V) If $\nu<1$ and $\mu<0$, then, for the same $\xi$ and $\om_\pm$ as above,
\begin{eqnarray}\label{phip1num2p}
\phi^+_q(\xi)&=&\exp\left[\frac{1}{2\pi i}\int_{\Im\eta=\om_-}
\frac{\xi \ln \Psi(q,\eta)}{\eta(\xi-\eta)}d\eta\right]\\\label{phip1num2m}
\phi^-_q(\xi)&=&(1-i\mu\xi/q)^{-1}\exp\left[-\frac{1}{2\pi i}\int_{\Im\eta=\om_+}\frac{\xi \ln \Psi(q,\eta)}{\eta(\xi-\eta)}d\eta\right].
\end{eqnarray}
Formulas \eqref{phip1num1p}-\eqref{phip1num2m} are more efficient than \eqref{phip1}-\eqref{phim1} because
if either $q$ or $\eta$ (or both) tend to $\infty$ in the complex plane so that the RHS in \eqref{defPsinule} does not cross
$(-\infty,0]$, and $(q-i\mu\eta)^{-1}=O((|q|+|\eta|)^{-1}$, then $\ln \Psi(q,\eta))$ is well-defined and tends to zero 
as $|\eta|^\nu/(|q|+|\eta|)$. Hence, the integrals converge faster.

\mbr\noindent
(VI) Similarly, if the process has the BM component of volatility $\sg$, and $\nu<2$ is the order of the pure jump component,
we introduce 
\begin{equation}\label{defPBMnu}
\Psi(q,\eta)=(1+\psi(\eta)/q)/(1+\eta^2\sg^2/(2q)),
\end{equation}
set $\be^\pm_q=\pm \sqrt{2q}/\sg$,
and represent the Wiener-Hopf factors in the form
\begin{eqnarray}\label{phipBMnu}
\phi^+_q(\xi)&=&\frac{\bepq}{\bepq-i\xi}\left[\frac{1}{2\pi i}\int_{\Im\eta=\om_-}
\frac{\xi \ln \Psi(q,\eta)}{\eta(\xi-\eta)}d\eta\right]\\\label{phimBMnu}
\phi^-_q(\xi)&=&\frac{\bemq}{\bemq+i\xi}\exp\left[-\frac{1}{2\pi i}\int_{\Im\eta=\om_+}\frac{\xi \ln \Psi(q,\eta)}{\eta(\xi-\eta)}d\eta\right].
\end{eqnarray}
\mbr\noindent
(VII) Finally, if the process has the BM component of volatility $\sg$, $\mu$ is the drift, and $\nu<1$ is the order of the pure jump component,
we introduce 
\begin{equation}\label{defPBMnu2}
\Psi(q,\eta)=(q+\psi(\eta))/(q-i\mu\eta+\eta^2\sg^2/2),
\end{equation}
and $\be^\pm_q=(-\mu\pm (\mu^2+2q\sg^2)^{1/2})/\sg^2$ and represent the Wiener-Hopf factors in the form
\eqref{phipBMnu}-\eqref{phimBMnu}. The rate of convergence of the integrals increases further.

 \section{Sinh-acceleration in the Laplace inversion formula}\label{LaplSinh} 
 Let $\sg, \om^{\prime\prime}, b^{\prime\prime}>0$ and $\sg-b^{\prime\prime}\sin\om^{\prime\prime}>0$. Introduce
 the function 
 %\eqref{sinhLapl}.
 %As it was demonstrated in \cite{paired}, in some cases, the Gaver-Stehfest method is unstable for options of moderate maturities,
%and then a somewhat more involved method of \cite{paraLaplace} need to be applied. We make a similar deformation and the
%corresponding changes of variables in the integral on the RHS of \eqref{Vbe4c} and in the outer integral on the RHS of \eqref{Vbe4b} using
%the map of the form 
 \bbe\label{sinhLapl}
 \bC\ni q=\chi_0(\sg,\om^{\prime\prime},b^{\prime\prime},y):= \sg+ib^{\prime\prime}\sinh(i\om^{\prime\prime}+y)\in \bC,
 \ee
 and denote by
 $\cL_0=\cL_0(\sg,\om^{\prime\prime},b^{\prime\prime})$ be the image of $\bR$ under the map $\chi_0(\sg,\om^{\prime\prime},b^{\prime\prime},\cdot)$.
We fix $\om^{\prime\prime}\in (0,\pi/4)$, $k_d\in (0,1)$, and set $d^{\prime\prime}=k_d|\om^{\prime\prime}|$, 
$\ga^{\prime\prime}_\pm=\om^{\prime\prime}\pm d^{\prime\prime}$. If the contour $\cL(\om'_1,\om',b)$ with $\om'>0$
is used to calculate the Wiener-Hopf factors, we set $d'=k_d|\om'|$ and choose $\om'\in (0,\pi/4)$.
Denote by $\cS$ the image of the strip $S_{(-d^{\prime\prime}, d^{\prime\prime})}$ under the map
$y\mapsto \chi_0(\sg,\om^{\prime\prime},b^{\prime\prime},y)$ and by $\cS'$ the image of
the strip $S_{(-d', d')}$ under the map
$y\mapsto \psi(\chi(\om'_1,\om',b',y))$. We choose the parameters of the deformations so that 
the sum $\cS+\cS'=\{q\in \cS, \eta\in \cS'\}$ does not intersect $(-\infty,0]$ in the process of deformation of the two initial lines of integration;
then the initial formulas \eqref{phip1}-\eqref{phim1} can be applied (provided $\xi$ is below $S'$). The other formulas 
for the Wiener-Hopf factors above can be applied under weaker conditions. See \cite{paraLaplace,paired} for details in the similar case
of the fractional-parabolic deformations.

The necessary condition on the pair $\om'$ and $\om^{\prime\prime}$, which ensures that the intersection of $\cS+\cS'$ 
with the exterior of a sufficiently large ball in $\bC$ does not intersect $(-\infty,0]$. At infinity, $\cS$ stabilizes to
the cone $\cC\cup \bar \cC$, where $\bar z$ denotes the complex conjugation, and
\[
\cC:=\{e^{i\phi}\bR_{++}\ |\ \phi\in (-\om^{\prime\prime}-d^{\prime\prime}-\pi/2, -\om^{\prime\prime}+d^{\prime\prime}-\pi/2)\}.
\]
$\cS'$ also stabilizes at infinity to  a cone of the form $\cC'\cup \bar \cC'$, but the description of $\cC'$
 is more involved than that of $\cC$. We need to consider 3 cases:
(1) $\nu\in (1,2]$ or $\nu\in (0,1]$ and $\mu=0$; (2) $\nu=1$ and $\mu\neq 0$; (3) $\nu\in (0,1)$ and $\mu\neq 0$.

\vskip0.1cm\noindent
(1) We use \eqref{aspsiinf} to conclude that 
\[
\cC':=\{e^{i\phi}\bR_{++}\ |\ \phi\in (\nu(\om^{\prime}+d^{\prime}), \nu(\om^{\prime}-d^{\prime}))\}.
\]
Let $\om'\ge 0$, then $(\cC\cup \bar \cC)+(\cC'\cup \bar \cC')$ does not intersect $(-\infty,0]$, if and only if
$\nu(\om^{\prime}+d^{\prime})+\om^{\prime\prime}+d^{\prime\prime}\in (0,\pi/2)$. The case $\om<0$ is by symmetry.
In both cases,
if $k_d<1$ can be chosen arbitrarily close to 1,
then the necessary and sufficient condition on $(\om', \om^{\prime\prime})\in (0,\pi/2)^2$ is $\om^{\prime\prime}+\nu|\om'|<\pi/2$. 
\vskip0.1cm\noindent
(2) Let $\varphi_0=\mathrm{arg}\, (-i\mu+c^0_\infty)=-\arctan (\mu/c^0_\infty)$. Then
\[
\cC':=\{e^{i\phi}\bR_{++}\ |\ \phi\in (\om^{\prime}+d^{\prime}+\varphi_0, \om^{\prime}-d^{\prime}+\varphi_0)\},
\]
and, therefore, the condition is: $|\om'|+d'+\om^{\prime\prime}+d^{\prime\prime}<\pi/2-|\varphi_0|$.
\vskip0.1cm\noindent
(3) Formally, we have the same condition as in (2), with $|\varphi_0|=\pi/2$. Clearly, this condition fails for any positive 
$|\om'|$ and $\om^{\prime\prime}$. Hence, in this case, it is impossible to use the sinh-acceleration w.r.t. $q$ and $\eta$ and apply
\eqref{phip1}-\eqref{phim1}. However, it is possible to choose the deformations so that $S+S'\not\ni 0$. If, in addition, the parameters are chosen so that,
for $q,\eta$ of interest, $|1-\Psi(q,\eta)|<1$, then we can apply the sinh-acceleration in  \eqref{phip1num1p}-\eqref{phip1num2m},
and the sinh-acceleration in the integral for the Laplace inversion.

\section{Gaver-Stehfest method}\label{GaverStehfest}
 If $\tf(q)$ is the Laplace transform
 of $f:\bR_+\to\bR$, then the Gaver-Stehfest approximation to $f$ is given by
 \begin{equation}\label{GS31}
f(M;T)=\frac{\ln(2)}{T}\sum_{k=1}^{2M}\zeta_k(M)\tilde f\left(\frac{k\ln(2)}{T}\right),
\end{equation}
where $M$ is a positive integer,
\begin{equation}\label{GS32}
\zeta_k(M)=(-1)^{M+k}\sum_{j=\lfloor (k+1)/2\rfloor}^{\min\{k,M\}}\frac{j^{M+1}}{M!}\left(\begin{array}{c} M \\ j\end{array}\right)
\left(\begin{array}{c} 2j \\ j\end{array}\right)\left(\begin{array}{c} j \\ k-j\end{array}\right)
\end{equation}
and $\lfloor a \rfloor$ denotes the largest integer that is less than or equal to  $a$. 
%Applying \eqref{GS31}
%and \eqref{GS32} to \eqref{defV}, we obtain the approximation
%\begin{equation}\label{ntpriceGS}
% V(M; h; T,x)=\sum_{k=1}^{2M}\zeta_k(M)\frac{\exp(k\ln(2))}{k}\cE^-_{q(T,k)}\bfo_{(-\infty,h]}(x)
% V(M; h; T,x)=\sum_{k=1}^{2M}\zeta_k(M)k^{-1}\cE^-_{q(T,k)}\bfo_{(-\infty,h]}(x),
%\end{equation}
%where $q(T,k)=k\ln(2)/T$. If the expectation $\cE^-_q\bfo_{(h,+\infty)}(x)$ can be calculated very accurately,
If $\tf(q)$ can be calculated sufficiently accurately,then, in many cases, the Gaver-Stehfest approximation with a moderate number of terms ($M\le 8$) is sufficiently accurate
for practical purposes even if double precision arithmetic is used (sometimes, even $M=9$ can be used). However,
it is possible that larger values of $M$ are needed,
and then high precision arithmetic becomes indispensable. The required system precision is about $2.2*M$, and
about $0.9*M$ significant digits are produced for $f(t)$ with good transforms. ``Good" means that $f$ is of class $C^\infty$, and
the transform's singularities are on the negative real axis. If the transforms are not good, then the number of significant digits may not
be so great and may not be proportional to $M$. See \cite{AbWh06}.

%\subsection{Gaver-Wynn-Rho algorithm}\label{Gaver-Wynn-Rho}
As \cite{AbateValko04} indicates and the numerical experiments in the setting of pricing barrier options and CDSs confirm \cite{paired},
Wynn's rho algorithm is more stable than the Gaver-Stehfest method.
Given a converging sequence $\{f_1, f_2,
\ldots\}$, Wynn's algorithm estimates the limit $f=\lim_{n\to\infty}f_n$ via $\rho^1_{N-1}$, where $N$ is even,
and $\rho^j_k$, $k=-1,0,1,\ldots, N$, $j=1,2,\ldots, N-k+1$, are calculated recursively as follows:
\begin{enumerate}[(i)]
\item
$\rho^j_{-1}=0,\ 1\le j\le N;$
\item
$\rho^j_{0}=f_j,\ 1\le j\le N;$
\item
in the double cycle w.r.t. $k=1,2,\ldots,N$, $j=1,2,\ldots, N-k+1$, calculate
\[
\rho^j_{k}=\rho^{j+1}_{k-2}+k/(\rho^{j+1}_{k-1}-\rho^{j}_{k-1}).
\]
We apply Wynn's algorithm with the Gaver functionals
\[
f_j(T)=\frac{j\ln 2}{T}\left(\frac{2j}{j}\right)\sum_{\ell=0}^j (-1)^j\left(\frac{j}{\ell}\right)\tilde f((j+\ell)\ln 2/T).
\]

 \end{enumerate}

\section{Tables}\label{Tables}
\subsection{Tables ~\ref{StVar1_1.2slow}-\ref{AsymKBL_1.95}: 
Dual-Static vs Variance minimizing hedging of the European option with the
payoff $\cG_{ex}(S)=(S/H)^\be(H^2/S-K_0)$}\label{Eurotables}

\begin{table}

\caption{\small KoBoL 
close to NIG, with an almost symmetric {\em slowly decaying jump density}, and no ``drift": $m_2=0.1$, $\nu=1.2, \lm=-5, \lp=6$,  $\mu=0$, $\sg=0$, 
$c=0.1670$ (rounded), $\be=-1$, $r=0.100$ (rounded).
 }
{\tiny
\begin{tabular}{c|ccc|c||ccccc|c}
\hline\hline
$\begin{array}{cc}
%\nu=1.95 \\
T=0.01
\end{array}$ & & $\#K=3$ & & & & & $\#K=5$ & & & \\
\hline
 Static   & $n_1$ & $n_2$ & $n_3$ & $nStd$ & $n_1$ & $n_2$ & $n_3$ & $n_4$ & $n_4$ & $nStd$ \\
 $s=0.5$  & 1.061 &	-0.794 &	1.756 &	0.26-2.64 &
 1.061 &	-0.084 &	0.168 &	-0.604 &	1.836 &	0.26-2.65\\
 $s=0.55$  & 1.061 &	-0.745 &	1.683 &	0.33-2.92 & 1.061 &-0.070 & 	0.164 &
 	-0.570 &	1.765 &	0.13-0.65\\\hline
	
 $VM1$ &  & & & & & & & & & \\
 $x'=-0.03$ & 1.061 &	-0.102&	0.880&	0.300 & 1.061  & 0.1448&	0.204&	-0.916&	2.000&	0.249\\	
$x'=-0.01$ & 1.061 &-0.236&	1.149&	0.566& 	1.061 &	0.1556&	0.268&	-1.390&	2.516&	0.480\\
$x'=0.00$ & 1.061 &	-0.368&	1.348&	0.902 & 1.061 &	0.156&	0.314&	-1.629&	2.760&	0.774\\
$x'=0.01$ & 1.061 &	-0.513&	1.555&	1.449 & 1.061 &0.1503&	0.367&	-1.863&	2.994&	1.255\\
$x'=0.03$ & 1.061 &	-0.800&	1.946&	2.851 &1.061 & 0.1333&	0.479&	-2.310&	3.431&	2.503\\\hline

$VM2$ &  & & & & & & & & & \\
 $x'=-0.03$ & 1.200&	-0.304&	0.948&	0.230 & 1.106&	0.077&	0.237&	-0.931&	2.004&	0.249\\
$x'=-0.01$ &	1.161&	-0.419&	1.238 &	0.564 & 1.102&	0.075&	0.328&	-1.418&	2.524&	0.480\\
$x'=0.00$ &1.1621&	-0.573&	1.459&	0.900 & 1.099&	0.072&	0.381&	-1.658&	2.767&	0.774\\
$x'=0.01$ & 1.181&	-0.769&	1.696&	1.446 & 1.099&	0.0619&	0.440&	-1.892&	3.001&	1.255\\
	$x'=0.03$ & 1.245&	-1.197&	2.166&	2.85 & 1.103&	0.033&	0.559&	-2.339&	3.438&	2.503\\
		\hline\hline
$\begin{array}{cc}
%\nu=1.95 \\
T=0.1
\end{array}$ & &  & & & & &  & & & \\
\hline
Static & & & &0.57-1.02 & & & & & & 0.51-0.93
\\\hline
 $VM1$ &  & & & & & & & & & \\
$x'=-0.03$ & 1.061 &-1.394&	2.535&	0.567 & 1.061 &	0.079&	0.645&	-2.370&	3.219&	0.506\\
$x'=0.00$ &1.061 &	-1.293&	2.456&	0.741 & 1.061   &	0.090&	0.623&	-2.439&	3.358&	0.659\\
$x'=0.03$ &1.061    &	-1.309&	2.517&	1.003& 1.061   &0.092&	0.643&	-2.640&	3.614&	0.893\\\hline
$VM2$ &  & & & & & & & & & \\
 $x'=-0.03$ & 1.510&	-2.244&	2.941&	0.566 & 1.124&	-0.051&	0.730&	-2.392&	3.223&	0.506\\
$x'=0.00$ &	1.456&	-2.065&	2.839&	0.740& 1.120&	-0.035&	0.708 &	-2.460&	3.363&	0.659\\
$x'=0.03$ & 1.438&	-2.067&	2.905&	1.001 & 1.118&	-0.033&	0.732&	-2.666&	3.619&	0.893\\
\hline\hline
$\begin{array}{cc}
%\nu=1.95 \\
T=0.5
\end{array}$ & &  & & & & &  & & & \\
\hline
Static & & & &1.02-1.19 & & & & & & 0.93-1.09
\\\hline
 $VM1$ &  & & & & & & & & & \\
$x'=-0.03$ & 1.061 &--7.000&	9.246&	0.870 &1.061 &-0.288&	2.836&	-11.023&	11.155&	0.811 \\
$x'=0.00$ &1.061 &	-6.393&	8.563&	0.943 &	1.061 &	-0.245&	2.606&	-10.193&	10.446&	0.878\\	
$x'=0.03$ & 1.061 &	-5.890&	7.997&	1.030& 1.061 &-0.211&	2.417&	-9.512&	9.863&	0.957\\\hline

$VM2$ &  & & & & & & & & & \\
 $x'=-0.03$ & 2.927&	-10.687&	11.074&	0.869 & 1.223&	-0.643&	3.080&	-11.084&	11.166&	0.811\\
 $x'=0.00$ & 2.759&	-9.763&	10.243&	0.942& 1.211&	-0.575&	2.834&	-10.251&	10.456&	0.878\\
 $x'=0.03$ & 2.619&	-8.999&	9.555&	1.029 & 1.201&	-0.519&	2.632&	-9.567&	9.873&	0.957 \\\hline
	\end{tabular}
}
\label{StVar1_1.2slow}
 \end{table}

\begin{table}
\caption{\small KoBoL 
close to NIG, with an almost symmetric jump density, and no ``drift": $m_2=0.1$, $\nu=1.2, \lm=-11, \lp=12$,  $\mu=0$, $\sg=0$, 
$c=0.3026$ (rounded), $\be=-1$, $r=0.100$ (rounded).
 }
{\tiny
\begin{tabular}{c|ccc|c||ccccc|c}
\hline\hline
$\begin{array}{cc}
%\nu=1.95 \\
T=0.01
\end{array}$ & & $\#K=3$ & & & & & $\#K=5$ & & & \\
\hline
 Static   & $n_1$ & $n_2$ & $n_3$ & $nStd$ & $n_1$ & $n_2$ & $n_3$ & $n_4$ & $n_4$ & $nStd$ \\
 $s=0.5$  & 1.061 &	-0.794 &	1.756 &	0.22-0.91 &
 1.061 &	-0.084 &	0.1684 &	-0.604 &	1.836 &	0.14-0.67\\
 $s=0.55$  & 1.061 &	-0.745 &	1.683 &	0.21-0.88 & 1.061 &-0.070 & 	0.164 &
 	-0.570 &	1.765 &	0.13-0.65\\\hline
	
 $VM1$ &  & & & & & & & & & \\
 $x'=-0.03$ & 1.061 &	0.024&	0.487&	0.104 & 1.061  & 0.158&	0.115 &	-0.192&	0.840&	0.068\\	
$x'=-0.01$ & 1.061 &0.023&	0.535&	0.163& 	1.061 &	0.170&	0.111&	-0.295&	0.987&	0.110\\
$x'=0.00$ & 1.061 &	0.002&	0.582&	0.226 & 1.061 &	0.175&	0.114&	-0.350&	1.058&	0.155\\
$x'=0.01$ & 1.061 &	-0.028&	0.638&	0.326 & 1.061 &	0.178&	0.121&	-0.407&	1.127&	0.226\\
$x'=0.03$ & 1.061 &	-0.098&	0.755&	0.652 &1.061 & 0.176&	0.145&	-0.520&	1.254&	0.465\\\hline

$VM2$ &  & & & & & & & & & \\
 $x'=-0.03$ & 1.146&	-0.109&	0.538&	0.103 & 1.102&	0.092&	0.150&	-0.205&	0.844&	0.068\\
$x'=-0.01$ &	1.128&	-0.101&	0.598&	0.161 & 1.097&	0.094&	0.166&	-0.318&	0.994&	0.109\\
$x'=0.00$ &1.124&	-0.126&	0.653&	0.223 & 1.098&	0.096&	0.176&	-0.377&	1.066&	0.153\\
$x'=0.01$ & 1.126&	-0.165&	0.716&	0.322 & 1.098&	0.096&	0.187&	-0.435&	1.134&	0.225\\
	$x'=0.03$ & 1.137&	-0.264&	0.850&	0.647 & 1.098&	0.0920&	0.213&	-0.546&	1.260&	0.464\\
		\hline\hline
$\begin{array}{cc}
%\nu=1.95 \\
T=0.1
\end{array}$ & &  & & & & &  & & & \\
\hline
Static & & & &0.24-0.36 & & & & & & 0.20-0.29
\\\hline
 $VM1$ &  & & & & & & & & & \\
$x'=-0.03$ & 1.061 &-0.876&	1.764&	0.239 & 1.061 &	0.121&	0.400&	-1.304&	1.958&	0.188\\
$x'=0.00$ &1.061 &	-0.705&	1.556&	0.280 & 1.061   &	0.133&	0.345&	-1.140&	1.815&	0.217\\
$x'=0.03$ &1.061    &	-0.596&	1.422&	0.340& 1.061   &0.142&	0.311&	-1.044&	1.736&	0.262\\\hline
$VM2$ &  & & & & & & & & & \\
 $x'=-0.03$ & 1.371&	-1.470&	2.053&	0.238 & 1.113&	0.013&	0.471&	-1.322&	1.961&	0.188\\
$x'=0.00$ &	1.317&	-1.208&	1.807&	0.279& 1.109&	0.031&	0.415&	-1.159&	1.819&	0.217\\
$x'=0.03$ & 1.282&	-1.040&	1.650&	0.339 & 1.1067&	0.042&	0.381&	-1.064&	1.740&	0.261\\
\hline\hline
$\begin{array}{cc}
%\nu=1.95 \\
T=0.5
\end{array}$ & &  & & & & &  & & & \\
\hline
Static & & & &0.72-0.78 & & & & & & 0.63-0.68
\\\hline
 $VM1$ &  & & & & & & & & & \\
$x'=-0.03$ & 1.061 &-5.943&	7.964&	0.567 &1.061 &-0.206&	2.353&	-9.077&	9.315&	0.507 \\
$x'=0.00$ &1.061 &	-5.298&	7.223&	0.596 &	1.061 &	-0.162&	2.107&	-8.155&	8.498&	0.530\\	
$x'=0.03$ & 1.061 &	-4.745&	6.582&	0.629& 1.061 &-0.125&	1.897&	-7.360&	7.789&	0.556
	\\
		\hline

$VM2$ &  & & & & & & & & & \\
 $x'=-0.03$ & 2.657&	-9.101&	9.532&	0.566 & 1.201&	-0.513&	2.564&	-9.130&	9.324&	0.507\\
 $x'=0.00$ & 2.483&	-8.126&	8.634&	0.594& 1.189&	-0.443&	2.302&	-8.204&	8.506&	0.530\\
 $x'=0.03$ & 2.335&	-7.289&7	7.858&	0.628 & 1.178&	-0.384&	2.077&	-7.406&	7.797&	0.556 \\\hline
	\end{tabular}
}
\label{StVar1_1.2}
 \end{table}

\begin{table}

\caption{\small KoBoL 
close to BM, with an almost symmetric jump density, and no ``drift": $m_2=0.1$, $\nu=1.95, \lm=-11, \lp=12$,  $\mu=0$, $\sg=0$, 
$c=0.0029$ (rounded), $\be=-1$, $r=0.100$ (rounded).
}
{\tiny
\begin{tabular}{c|ccc|c||ccccc|c}
\hline\hline
$\begin{array}{cc}
%\nu=1.95 \\
T=0.01
\end{array}$ & & $\#K=3$ & & & & & $\#K=5$ & & & \\
\hline
 Static   & $n_1$ & $n_2$ & $n_3$ & $nStd$ & $n_1$ & $n_2$ & $n_3$ & $n_4$ & $n_4$ & $nStd$ \\
 $s=0.5$  & 1.061 &	-0.794 &	1.756 &	0.20-	0.78 & 
 1.061 &	-0.084 &	0.168 &	-0.604 &	1.836 &	0.16-	0.36\\
 $s=0.55$  & 1.061 &	-0.745 &	1.683 &	0.19-0.75 & 1.061 &	-0.070 &	0.164 &	-0.570 &	
 1.766 &	0.15-	0.34\\\hline
 $VM1$ &  & & & & & & & & & \\
 $x'=-0.03$ & 1.061 &	0.117&	0.226 &	0.028 & 1.061 &	0.168 &	0.073 &	0.084 &	0.276 &	0.015\\
	$x'=-0.01$ & 1.061 &	0.143 &	0.182 &	0.037 & 1.061 &	0.176&	0.060&	0.079&	0.354 &	0.022\\
$x'=0.00$ &	1.061 &	0.151&	0.173&	0.046 & 1.061 &	0.180&	0.054&	0.072&	0.412&	0.029\\
$x'=0.01$ &      1.061 &	0.155 &	0.176 &	0.062 & 1.061 &	0.184&	0.048&	0.061&	0.475&	0.039\\
$x'=0.03$ &      1.061 &	0.152&	0.228&	0.129  & 1.061 &	0.190&	0.040&	0.004&	0.614&	0.081\\\hline

$VM2$ &  & & & & & & & & & \\
 $x'=-0.03$ & 1.114&	0.028&	0.267&	0.027 & 1.099&  0.101&	0.111&	0.072&	0.280&	0.014\\
$x'=-0.01$ &	1.105&	0.059&	0.228 &	0.034 & 1.098&	0.103&	0.110&	0.059&	0.361&	0.019\\
$x'=0.00$ & 1.103&	0.066&	0.224 &	0.042 & 1.097 &	0.104 &	0.110&	0.048&	0.421&	0.024\\
$x'=0.01$ & 1.102&	0.068&	0.232&	0.056 & 1.097&	0.105 &	0.110&	0.031&	0.486&	0.032\\
	$x'=0.03$ & 1.101&	0.060&	0.291&	0.118 & 1.095&	0.107&	0.113&	-0.035&	0.629&	0.068\\
		\hline\hline
$\begin{array}{cc}
%\nu=1.95 \\
T=0.1
\end{array}$ & &  & & & & &  & & & \\
\hline
Static & & & &0.17-0.27 & & & & & & 0.15-0.25
\\\hline
 $VM1$ &  & & & & & & & & & \\
$x'=-0.03$ & 1.061 & -0.655&	1.444&	0.148 & 1.061   &0.138&	0.298&	-0.875&	1.455&	0.103\\
$x'=0.00$ &1.061 &	-0.466&	1.193&	0.158 & 1.061 &	0.151&	0.237&	-0.657&	1.230&	0.106\\
$x'=0.03$ &1.061    &-0.327&	1.000&	0.173 & 1.061  &	0.160&	0.192&	-0.496&	1.056&	0.112\\\hline

$VM2$ &  & & & & & & & & & \\
 $x'=-0.03$ & 1.313&	-1.141&	1.682&	0.147 & 1.108& 	0.040&	0.364&	-0.892&	1.458&	0.103\\
$x'=0.00$ &1.258&	-0.856&	1.388&	0.157 & 1.105&	0.058&	0.300&	-0.674&	1.233&	0.106\\
$x'=0.03$ & 1.220&	-0.646&	1.164 &	0.172 & 1.102&	0.070 &	0.254&	-0.513&	1.060&	0.112\\
\hline\hline
$\begin{array}{cc}
%\nu=1.95 \\
T=0.5
\end{array}$ & &  & & & & &  & & & \\
\hline
Static & & & &0.65-0.68 & & & & & & 0.55-0.57
\\\hline
 $VM1$ &  & & & & & & & & & \\
$x'=-0.03$ & 1.061 &-5.609&	7.560&	0.495 &1.061 &	-0.180&	2.202&	-8.471&	8.741&	0.436 \\
$x'=0.00$ &1.061 &	-4.956&	6.804&	0.514 &	1.061 &	-0.136&	1.953&	-7.526&	7.897&	0.449\\	
	$x'=0.03$ & 1.061 &	-4.391&	6.144&	0.536 &1.061 &-0.098&	1.734&	-6.705&	7.157&	0.465\\
		\hline
		
$VM2$ &  & & & & & & & & & \\
 $x'=-0.03$ & 2.571&	-8.600&	9.046&	0.494 & 1.194&	-0.473&	2.404&	-8.521&	8.750&	0.436\\
 $x'=0.00$ & 2.397&-7.614&	8.131&	0.513 & 1.182&	-0.403 &	2.137&	-7.573&	7.905&	0.449\\
 $x'=0.03$ & 2.247&	-6.760&	7.333&	0.535 & 1.171&	-0.342&	1.907&	-6.748&	7.165&	0.465 \\\hline
	\end{tabular}
}
\label{StVar1_1.95}
 \end{table}

\begin{table}
\caption{\small BM with an embedded KoBoL
close to NIG, an almost symmetric jump density: $m_2=0.15$, $\sg^2=0.05$,
$\mu=0.05$, $\nu=1.2, \lm=-12, \lp=12.5$,  
$c=0.31865$ (rounded), $\be=-0.5$, $r=5.0*10^{-6}$ (rounded).
 }
{\tiny
\begin{tabular}{c|ccc|c||ccccc|c}
\hline\hline
$\begin{array}{cc}
%\nu=1.95 \\
T=0.01
\end{array}$ & & $\#K=3$ & & & & & $\#K=5$ & & & \\
\hline
 Static   & $n_1$ & $n_2$ & $n_3$ & $nStd$ & $n_1$ & $n_2$ & $n_3$ & $n_4$ & $n_4$ & $nStd$ \\
 $s=0.5$ & 1.030 &	-1.520	& 2.232&	0.34-	1.12 &
 1.030 &	-0.326 &	0.121 &	-0.968 &	1.983 &	0.27-	0.70\\
 $s=0.55$  & 1.030 &	-1.282 &	1.921 &	0.29-	0.97 & 1.030&
 	-0.267 &	0.101 &	-0.821 &	1.717 &	0.23-0.59\\\hline
	
 $VM1$ &  & & & & & & & & & \\
 $x'=-0.03$ & 1.030 &	0.005&	0.099&	0.020 & 1.030 & 0.039&	0.022&	-0.019 &	0.129&	0.012\\	
$x'=-0.01$ & 1.030 &0.011&	0.098&0.027& 1.030 & 0.040&	0.020&	-0.026 &	0.146&	0.016\\
$x'=0.00$ & 1.030 &	0.012&	0.101&	0.033 & 1.030 &	0.041&	0.019&	-0.032&	0.156&	0.019\\
$x'=0.01$ & 1.030 &	0.011&	0.105&	0.041 & 1.030&0.042&	0.019&	-0.040&	0.167&	0.025\\
$x'=0.03$ & 1.030&	0.005&	0.119&	0.072 & 1.030&	0.043&	0.020&	-0.057&	0.189&	0.044\\\hline

$VM2$ &  & & & & & & & & & \\
 $x'=-0.03$ & 1.049&	-0.027&	0.114&	0.020 & 1.040&	0.022&	0.032&	-0.021&	0.130&	0.012\\
$x'=-0.01$ &	1.046&	-0.019&	0.113&	0.026   & 1.039&	0.023&	0.032&	-0.030&	0.147&	0.015\\	
$x'=0.00$ & 1.045&	-0.018&	0.116&	0.032 & 1.039 &	0.023&	0.032&	-0.037&	0.157&	0.019\\
$x'=0.01$ & 1.045&	-0.019&	0.122&	0.040 & 1.039 &	0.023&	0.033&	-0.046&	0.169&	0.024\\
	$x'=0.03$ & 1.045&	-0.026&	0.137&	0.071 & 1.039 &	0.023&	0.036&	-0.064&	0.191&	0.044\\
		\hline\hline
$\begin{array}{cc}
%\nu=1.95 \\
T=0.1
\end{array}$ & &  & & & & &  & & & \\
\hline
Static & & & &0.30-0.45 & & & & & & 0.25-0.36
\\\hline
 $VM1$ &  & & & & & & & & & \\
$x'=-0.03$ & 1.030 &-0.237&	0.437&	0.048 & 1.030 &	0.028&	0.098&	-0.323&	0.440&	0.036\\
$x'=0.00$ &1.030 &	-0.194&	0.383&	0.052 & 1.030 &	0.031&	0.084&	-0.274&	0.395&	0.039\\
$x'=0.03$ &1.030&-0.160&	0.341&	0.059 & 1.030 &	0.033&	0.073&	-0.237&	0.361&	0.043\\\hline

$VM2$ &  & & & & & & & & & \\
 $x'=-0.03$ & 1.111&	-0.394&	0.513&	0.048& 1.043&0.001&	0.116&	-0.327&	0.441&	0.036\\
$x'=0.00$ &	1.099&	-0.328&	0.449&	0.052& 1.042&	0.005&	0.101&	-0.279&	0.396&	0.039\\
$x'=0.03$ & 1.088&	-0.277&	0.400&1	0.058 & 1.041&	0.008&	0.090&	-0.242&	0.361&	0.043\\
\hline\hline
$\begin{array}{cc}
%\nu=1.95 \\
T=0.5
\end{array}$ & &  & & & & &  & & & \\
\hline
Static & & & &0.22-0.32 & & & & & & 0.21-0.30
\\\hline
 $VM1$ &  & & & & & & & & & \\
$x'=-0.03$ & 1.030 &	-1.467&	1.886&	0.109 &1.030 &	-0.050&	0.558&	-2.136&	2.097&	0.096 \\
$x'=0.00$ &1.030 &-1.344&	1.749&	0.113&	1.030 &	-0.042&	0.512&	-1.964&	1.948&	0.098\\	
	$x'=0.03$ & 1.030 &	-1.233&	1.624&	0.116  &	1.030 &	-0.035&	0.471&	-1.808&	1.813&	0.101\\
		\hline

$VM2$ &  & & & & & & & & & \\
 $x'=-0.03$ & 1.426&	-2.248&	2.272&	0.109 & 1.064 &	-0.124&	0.609&	-2.149&	2.099&	0.096\\
 $x'=0.00$ & 1.393&	-2.061&	2.104&	0.112 & 1.062&	-0.111&	0.560&	-1.976&	1.950&	0.098\\
 $x'=0.03$ & 1.363&1	-1.894&	1.953&	0.116  & 1.060&	-0.100&	0.516&	-1.820&	1.815&	0.101\\\hline
	\end{tabular}
}
\label{BMKBL_1.2}
 \end{table}
 
 \begin{table}
\caption{\small BM with an embedded KoBoL model
close to BM, an almost symmetric jump density: $m_2=0.15$, $\sg^2=0.05$,
$\mu=0.05$, $\nu=1.95, \lm=-12, \lp=12.5$,  
$c=0.0029$ (rounded), $\be=-0.5$, $r=1.8*10^{-7}$ (rounded).
 }
{\tiny
\begin{tabular}{c|ccc|c||ccccc|c}
\hline\hline
$\begin{array}{cc}
%\nu=1.95 \\
T=0.01
\end{array}$ & & $\#K=3$ & & & & & $\#K=5$ & & & \\
\hline
 Static   & $n_1$ & $n_2$ & $n_3$ & $nStd$ & $n_1$ & $n_2$ & $n_3$ & $n_4$ & $n_4$ & $nStd$ \\
 $s=0.5$  & 1.030&	-1.520 &	2.232&	0.32-	1.13 &
 1.030 &	-0.326 &	0.121 &	-0.968 &	1.983 &	0.28-	0.57\\
 $s=0.55$  & 1.030 &	-1.282 &	1.921 &	0.27-	0.96 & 1.030&
 	-0.267 &	0.101 &	-0.821 &	1.717 &	0.24-0.48\\\hline
	
 $VM1$ &  & & & & & & & & & \\
 $x'=-0.03$ & 1.030 &	0.023&	0.061&	0.008  & 1.030
 & 0.040&	0.015&	0.017&	0.054&	0.003\\	
 
$x'=-0.01$ & 1.030 &0.031&	0.048&	0.008& 	1.039&	0.024&	0.025&	0.014&	0.055&	0.003\\
$x'=0.00$ & 1.030 &	0.033&	0.043&	0.009 & 1.030 &	0.042&	0.011&	0.020&	0.057&	0.005\\
$x'=0.01$ & 1.030 &	0.035&	0.040&	0.011 & 1.030&	0.043&	0.010&	0.020&	0.062&	0.006\\
$x'=0.03$ & 1.030&	0.037&	0.037 &	0.017 &1.030 &	0.044&	0.008&	0.017&	0.080&	0.011\\\hline

$VM2$ &  & & & & & & & & & \\
 $x'=-0.03$ & 1.044&	-0.001&	0.072&	0.007&  1.039 &	0.024 &	0.025 &	0.014 &	0.055 &	0.003\\
$x'=-0.01$ &	1.041&	0.009&	0.059&	0.008 & 1.039 &	0.024&	0.024&	0.015&	0.055&	0.003\\	
$x'=0.00$ & 1.041&	0.012&	0.055&	0.009 & 1.039&	0.0243&	0.024&	0.015 &	0.058&	0.004\\
$x'=0.01$ & 1.040&	0.014&	0.052&	0.010  &	1.039 &	0.025 &	0.024&	0.014&	0.064&	0.005\\
	$x'=0.03$ & 	1.040&	0.016&	0.051&	0.015 & 1.039 &	0.025&	0.024&	0.010&	0.082&	0.008\\
		\hline\hline
$\begin{array}{cc}
%\nu=1.95 \\
T=0.1
\end{array}$ & &  & & & & &  & & & \\
\hline
Static & & & &0.30-0.44& & & & & & 0.28-0.38
\\\hline
 $VM1$ &  & & & & & & & & & \\
$x'=-0.03$ & 1.030 &-0.212&	0.403&	0.040& 1.030 &	0.030&	0.087&	-0.278&	0.392&	0.028\\
$x'=0.00$ &1.030 &	-0.166&	0.345&	0.042 & 1.030 &0.033&	0.072&	-0.224&	0.340&	0.029\\
$x'=0.03$ &1.030&-0.130&	0.298&	0.044 & 1.030 &	0.035&	0.060&	-0.181&	0.297&	0.030\\\hline

$VM2$ &  & & & & & & & & & \\
 $x'=-0.03$ & 1.1047 &	-0.356 &	0.473&	0.039 & 1.042&	0.004&	0.104&	-0.282&	0.393&	0.028\\
$x'=0.00$ &	1.092&	-0.287&	0.405 &	0.041 & 1.041&	0.008&	0.088&	-0.228&	0.341&	0.029\\
$x'=0.03$ & 1.082&	-0.232&	0.350&	0.044 & 1.041&	0.012&	0.076&	-0.185&	0.298&	0.030\\
\hline\hline
$\begin{array}{cc}
%\nu=1.95 \\
T=0.5
\end{array}$ & &  & & & & &  & & & \\
\hline
Static & & & &0.22-0.32 & & & & & & 0.21-0.29
\\\hline
 $VM1$ &  & & & & & & & & & \\
$x'=-0.03$ & 1.030 &	-1.441&	1.856&	0.105 &1.030 &	-0.048&	0.547&	-2.092&	2.057&	0.092 \\
$x'=0.00$ &1.030 &	-1.317&	1.717&	0.108 &	1.030 &	-0.040&	0.500&	-1.918&	1.906&	0.094\\	
	$x'=0.03$ & 1.030 &	-1.205&	1.591&	0.111 &	1.030 &	-0.033&	0.459&	-1.760&	1.769&	0.096\\
		\hline

$VM2$ &  & & & & & & & & & \\
 $x'=-0.03$ & 1.419&	-2.208&	2.235&	0.105 & 1.063&	-0.121&	0.597&	-2.105&	2.059&	0.092\\
 $x'=0.00$ & 1.386&	-2.021&	2.066&	0.108 &  1.061&	-0.108&	0.548&	-1.929&	1.908&	0.094\\
 $x'=0.03$ & 1.356&	-1.852&	1.913&	0.111 & 1.059&	-0.096&	0.503&	-1.771&	1.771&	0.096\\\hline
	\end{tabular}
}
\label{BMKBL_1.95}
 \end{table}

\begin{table}
\caption{\small KoBoL 
close to NIG, with a sizably asymmetric jump density, and positive ``drift": $m_2=0.1$, $\nu=1.2, \lm=-12, \lp=8$,  $\mu=0.15$, $\sg=0$, 
$c=0.26312$ (rounded), $\be=4$, $r=0.0035$ (rounded).
}
{\tiny
\begin{tabular}{c|ccc|c||ccccc|c}
\hline\hline
$\begin{array}{cc}
%\nu=1.95 \\
T=0.01
\end{array}$ & & $\#K=3$ & & & & & $\#K=5$ & & & \\
\hline
 Static   & $n_1$ & $n_2$ & $n_3$ & $nStd$ & $n_1$ & $n_2$ & $n_3$ & $n_4$ & $n_4$ & $nStd$ \\
 $s=0.5$ & 0.961 &	-0.880 &	-0.071 &	0.47-	2.50 
  & 0.961 &	-0.883 &	0.005 &	0.001 &	-0.077	&0.45-2.36 \\
 $s=0.55$  & 0.961 &	-0.906 &	-0.041 &	0.48-2.53 & 0.961 &	-0.907 &	0.004 &	0.000 &	-0.045 &	0.46-2.44\\\hline

 $VM1$ &  & & & & & & & & & \\
 $x'=-0.03$ & 0.961 &	-0.088&	-0.396&	0.088
  & 0.961 &	-0.196&	-0.087&	0.075&	-0.482&	0.040\\

	$x'=-0.01$ & 0.961 &	-0.088&	-0.427&	0.134 & 0.961 &	-0.215&	-0.067&	0.123&	-0.550&	0.066\\
$x'=0.00$ &	0.961 &	-0.072&	-0.460&	0.186& 0.961&	-0.223&	-0.062&	0.149&	-0.582&	0.094\\
$x'=0.01$ & 0.961 &	-0.048&	-0.498&	0.267 & 0.961 &	-0.226&	-0.064&	0.178&	-0.612&	0.136\\
$x'=0.03$ & 0.961 &	0.003&	-0.574&	0.509 & 0.961 &	-0.227&	-0.076&	0.233&	-0.667&	0.265\\\hline

$VM2$ &  & & & & & & & & & \\
 $x'=-0.03$ & 0.873&	0.046&	-0.445&	0.086 & 0.909&	-0.113&	-0.129&	0.092&	-0.487&	0.039\\
$x'=-0.01$ &	0.889&	0.045&	-0.493&	0.130 & 0.912&	-0.118&	-0.138&	0.153&	-0.559&	0.061\\
$x'=0.00$ & 0.893&	0.067&	-0.533&	0.179 & 0.914&	-0.121&	-0.142&	0.182 &	-0.591&	0.088\\
$x'=0.01$ & 0.892&1	0.098&	-0.579&	0.259& 0.916&	-0.123&	-0.147&	0.211&	-0.621&	0.129\\
	$x'=0.03$ & 0.882&	0.175&	-0.670&	0.497 & 0.916&	-0.120&	-0.162&	0.264&	-0.675&	0.258\\
		\hline\hline
$\begin{array}{cc}
%\nu=1.95 \\
T=0.1
\end{array}$ & &  & & & & &  & & & \\
\hline
Static & & & &0.75-1.16 & & & & & & 0.70-1.12
\\\hline
 $VM1$ &  & & & & & & & & & \\
$x'=-0.03$ & 0.961 &0.341&	-0.985&	0.180 & 0.961 & -0.198&	-0.177&	0.468&	-0.849&	0.108\\
$x'=0.00$ &0.961 &	0.278&	-0.912&	0.210 & 0.961 &	-0.204&	-0.159&	0.429&	-0.820&	0.124\\
$x'=0.03$ &0.961 &	0.240&	-0.870&	0.254 &0.961 & -0.208&-0.148&	0.410&	-0.811&	0.147\\\hline

$VM2$ &  & & & & & & & & & \\
 $x'=-0.03$ & 0.769 &	0.706&	-1.1596&	0.178 & 0.908&	-0.087&	-0.250&	0.487&	-0.852&	0.108\\
$x'=0.00$ &	0.791&	0.609&	-1.075&	0.208&0.909&	-0.093&	-0.234&	0.449&	-0.824&	0.123\\
$x'=0.03$ & 0.806&	0.551&	-1.028&	0.251 & 0.910&	-0.097&	-0.225&	0.432&	-0.815&	0.147\\
\hline\hline
$\begin{array}{cc}
%\nu=1.95 \\
T=0.5
\end{array}$ & &  & & & & &  & & & \\
\hline
Static & & & &0.73-0.90 & & & & & & 0.70-0.87
\\\hline
 $VM1$ &  & & & & & & & & & \\
$x'=-0.03$ & 0.961&	1.340&	-2.197&	0.273 & 0.961 & -0.150&	-0.478&	1.655&	-1.952&	0.188\\
$x'=0.00$ &0.961 &	1.246&	-2.092&	0.287 &	0.961 &	-0.156&	-0.449&	1.555&	-1.866&	0.196\\	
	$x'=0.03$ & 0.961 &	1.1619&	-1.996&0.303 & 0.961&-0.161&	-0.423&	1.464&	-1.788&	0.205\\\hline

$VM2$ &  & & & & & & & & & \\
 $x'=-0.03$ & 0.515&	2.216&	-2.628&	0.271&  0.894&	-0.005&	-0.577&	1.680&	-1.956&	0.188\\
 $x'=0.00$ & 0.542&	2.074&	-2.501&	0.285 & 0.896&	-0.013&	-0.547&	1.579&	-1.870&	0.196\\
 $x'=0.03$ & 0.566&	1.947&	-2.387&	0.301 & 0.897&	-0.021&-0.519&	1.487&	-1.792&	0.205\\\hline
		
	\end{tabular}
}
 \label{AsymKBL_1.2}

 \end{table}

 \begin{table}

\caption{\small KoBoL 
close to BM, with a sizably asymmetric jump density, and positive ``drift": $m_2=0.1$, $\nu=1.95, \lm=-12, \lp=8$,  $\mu=0.15$, $\sg=0$, 
$c=0.00288$ (rounded), $\be=4$, $r=0.000$ (rounded).
}
{\tiny
\begin{tabular}{c|ccc|c||ccccc|c}
\hline\hline
$\begin{array}{cc}
%\nu=1.95 \\
T=0.01
\end{array}$ & & $\#K=3$ & & & & & $\#K=5$ & & & \\
\hline
 Static   & $n_1$ & $n_2$ & $n_3$ & $nStd$ & $n_1$ & $n_2$ & $n_3$ & $n_4$ & $n_4$ & $nStd$ \\
 $s=0.5$  & 0.961&	-0.867&	-0.080&	0.50-	1.12 & 
 0.961&	-0.870&	0.006&	0.001&	-0.086&	0.46-1.07\\
 $s=0.55$  & 0.961 &	-0.896 &	-0.049&	0.51-	1.15 & 0.961&	-0.898&	0.004&	0.001&-0.054&	0.48-	1.12\\\hline

 $VM1$ &  & & & & & & & & & \\
 $x'=-0.03$ & 0.961 &	-0.163&	-0.192&	0.032 & 0.961 &	-0.207&	-0.059& -0.081&	-0.194&	0.016\\
	$x'=-0.01$ & 0.961 &-0.186&	-0.157&	0.043 & 0.961 &	-0.217&	-0.042&	-0.082&	-0.242&	0.025\\
$x'=0.00$ &	0.961 &	-0.193&	-0.152&	0.054 & 0.961&-0.223&	-0.033&	-0.082&	-0.270&	0.033\\
$x'=0.01$ & 0.961 &	-0.197&	-0.158&	0.071 & 0.961 &	-0.228&	-0.024&	-0.078&	-0.298&	0.044\\
$x'=0.03$ & 0.961 &	-0.195&	-0.208&	0.140 & 0.961 &	-0.236&	-0.010&	-0.047&	-0.367&	0.086\\\hline

$VM2$ &  & & & & & & & & & \\
 $x'=-0.03$ & 0.899&	-0.058&	-0.239&	0.029 & 	0.915 &	-0.120&	-0.107&	-0.065&	-0.199&	0.012\\
$x'=-0.01$ &	0.906&	-0.082&	-0.214&	0.036 & 0.913&	-0.123&	-0.106&	-0.057&	-0.251&	0.015\\
$x'=0.00$ & 0.910&	-0.089 &	-0.214&	0.044 & 0.914&	-0.124&	-0.105&	-0.051&	-0.281&	0.0181\\
$x'=0.01$ & 0.910&	-0.087&	-0.227&	0.057 & 0.915&	-0.126&	-0.105&	-0.041&	-0.312&	0.023\\
	$x'=0.03$ & 0.911&	-0.079&	-0.284&	0.113&	0.916&	-0.128&	-0.106&	0.002&	-0.385&	0.046\\
		\hline\hline
$\begin{array}{cc}
%\nu=1.95 \\
T=0.1
\end{array}$ & &  & & & & &  & & & \\
\hline
Static & & & &0.73-1.08 & & & & & & 0.70-1.04
\\\hline
 $VM1$ &  & & & & & & & & & \\
$x'=-0.03$ & 0.961 &	0.232&	-0.825&	0.124 & 0.961 &	-0.205&	-0.138&	0.300&	-0.651&	0.068\\
$x'=0.00$ &0.961 &	0.147&	-0.714&	0.132 & 0.961 &	-0.211&	-0.114&	0.224&	-0.575&	0.069\\
$x'=0.03$ &0.961 &	0.080&	-0.624&	0.144 & 0.961 &	-0.216&	-0.095&	0.165&	-0.515&	0.073\\\hline

$VM2$ &  & & & & & & & & & \\
 $x'=-0.03$ & 0.797& 	0.547&	-0.978&	0.122 & 	0.910&	-0.097&	-0.209&	0.312&	-0.654&	0.068\\
$x'=0.00$ &	0.822&	0.420&	-0.849&	0.130 & 0.911&-0.104&	-0.186&	0.243&	-0.579&	0.069\\
$x'=0.03$ & 0.841&	0.320&	-0.746&	0.142 & 0.912&	-0.109&	-0.169&	0.185&	-0.519&	0.072\\
\hline\hline
$\begin{array}{cc}
%\nu=1.95 \\
T=0.5
\end{array}$ & &  & & & & &  & & & \\
\hline
Static & & & &0.72-0.88 & & & & & & 0.70-0.85
\\\hline
 $VM1$ &  & & & & & & & & & \\
$x'=-0.03$ & 0.961&	1.292&	-2.137&	0.251 &	0.961 &	-0.153&	-0.460&	1.582&	-1.881&	0.172\\
$x'=0.00$ &0.961 &	1.187&	-2.018&	0.260 &	0.961 &	-0.160&	-0.428&	1.467&	-1.780&	0.177\\	
	$x'=0.03$ & 0.961 &	1.091&	-1.908&	0.271 & 0.961&	-0.165&	-0.398&	1.3560&	-1.686&	0.182\\
		\hline
$VM2$ &  & & & & & & & & & \\
 $x'=-0.03$ &0.527&	2.145&	-2.558&	0.249 & 0.895&	-0.009&	-0.559&	1.607&	-1.885&	0.172\\
 $x'=0.00$ & 0.556&	1.988&	-2.415&	0.258 & 0.897&	-0.018&	-0.525&	1.491&	-1.784&	0.177\\
 $x'=0.03$ & 0.582&	1.843&	-2.282&	0.269 & 0.898&	-0.027&	-0.494&	1.384&	-1.690&	0.182 \\\hline
	\end{tabular}
}
 \label{AsymKBL_1.95}
 \end{table}
 
%\newpage
\subsection{Tables for the Wiener-Hopf factors (rounded) and prices of barrier options}\label{tablesWHFbarrier}

 \begin{table}

\caption{\small Wiener-Hopf factors $\phi^\pm_q(\xi)$, $q=1.1$, $\xi=\sinh(-0.15i+y)$, $|y|\le 15$,
for KoBoL 
close to NIG, with an almost symmetric jump density, and no ``drift": $m_2=0.1$, $\nu=1.2, \lm=-11, \lp=12$,  $\mu=0$, $\sg=0$, 
$c=0.3026$ (rounded)}
{\tiny
\begin{tabular}{c|ccccc}
\hline\hline
$y$ & -15 &	-10 &	-5 &	0 &	5  \\
\hline
$\Re\phipq(\xi)$& -3.573612313E-07 &	-4.67991818674E-05 &	-3.8732615824580E-03
&	1.03663360053561&	-3.8732615824580E-03\\
$\Im\phipq(\xi)$& -3.3276442219E-06 & 4.358762795728E-04 & -5.8173831101994E-02 & 2E-19 & 5.8173831101994E-02\\\hline
$\Re\phimq(\xi)$& 7.6320768431674E-07& 1.0003700599534E-04& 1.68715279339566E-02 &
0.972251835483482 & 1.68715279339566E-02\\
$\Im\phimq(\xi)$& 4.06659819933492E-06& 5.3263366781895E-04&6.88181160078408E-02& -6E-19 &-6.88181160078408E-02\\
\hline
Err15 & 1.14E-16 &	9.86E-17 &	2.71E-16 &	2.22E-16 &	3.19E-16 \\
Err10 & 6.31E-11 &	5.57E-11	&4.99E-11	&1.79E-12 &	4.99E-11\\\hline
\end{tabular}
}
\begin{flushleft}{\tiny
$\phipq(\xi)$ is calculated applying the sinh-acceleration with the parameters $\om_1=0$, $b=	3.37$, $\om=0.79$
to the integral in \eqref{phip1}.\\
$\phimq(\xi)$ is calculated applying the sinh-acceleration with the parameters $\om_1=-0.26$, $b=	3.96$, $\om=-0.86$
to the integral in \eqref{phim1}.\\
The parameters are chosen using the general recommendations for the sinh-acceleration method.\\
The mesh and number of points are chosen using the general recommendations for the given error tolerance\\
If $\eps=10^{-15}$ 
($\ze_-=0.0969$, $N_-=385$ for $\phipq(\xi)$, and $\ze_+=0.107$, $N_+=350$ for $\phimq(\xi)$); if
$\eps=10^{-10}$ 
($\ze_-=0.1410$, $N_-=175$ for $\phipq(\xi)$, and $\ze_+=0.1559$, $N_+=159$ for $\phimq(\xi)$).\\
The last two lines are absolute differences between $\phipq(\xi)$ calculated using \eqref{phip1} and 
$\phipq(\xi)$ calculated using \eqref{phim1} and the Wiener-Hopf identity.
CPU time for the calculation at 30 points, the average over 10k runs: 1.63 and	1.48 msec if $\eps=10^{-15}$,
and 1.15 and	1.06 msec if $\eps=10^{-10}$.
}
\end{flushleft}
\label{TableWHF}
 \end{table}
 
 \begin{table}
 \caption{\small Pricing no-touch and first-touch options (down case) using Gaver-Wynn-Rho method with $M=8$ and sinh-acceleration.
 Benchmark prices  
 and errors for different choices of the parameters of the sinh-acceleration are rounded. The underlying: KoBoL model as in Table 2.
 Parameters $m_2=0.1$, $\nu=1.2, \lm=-11, \lp=12$,  $\mu=0$, $\sg=0$, 
$c=0.3026$ (rounded), $\be=-1$, $r=0.100$ (rounded). Time to maturity $T=0.1$, barrier $H=1$, $S$  the spot.
 }
 
 {\tiny
\begin{tabular}{c|ccccccc}
\hline\hline
$\ln (S/H)$ & 0.01 &0.03 &0.05 &	0.07 &	0.09 &0.11 & 0.13  \\
\hline
$V_{nt}$ & 0.159611796 &	0.349192704 &	0.50261129 &	0.628008257 &	0.727388428 &	0.803409906 &	0.859701515\\
$V_{ft}$ & 0.837533576 &	0.645372195 &	0.490377871 &	0.363952758 &	0.263898333 &	0.187440524 &	0.130864377\\\hline
$(A)$& 
1.01E-08 &	-4.65E-10 &	-1.45E-08 &	6.09E-09 &	7.80E-09 &	-6.29E-09 &	3.80E-09 \\
& -6.32E-09	 & -2.76E-09 &	8.74E-10 &	-3.19E-09 &	-6.68E-09 &	-1.41E-08 &	-1.30E-09\\\hline
$(B)$& 2.08E-05 &	-7.47E-06 &	-2.68E-05 &	-3.36E-05	&-1.89E-04 &	8.85E-05	&5.24E-06\\
& -2.22E-05	 &7.70E-06 &	2.93E-05 &	4.94E-05 &	2.38E-04 &	-8.50E-05	&-5.71E-06\\\hline

\end{tabular}
}
\begin{flushleft}{\tiny
Curve used in the Fourier inversion $\cL_+:=\cL(0,1.1, 9.6)$. 
Curve used to calculate $\phimq(\xi), \xi\in \cL_+$: $\cL_-:=\cL(-1,-1.1, 8)$.  $N_\pm, \ze_\pm$ are chosen using
the universal simplified recommendation in \cite{SINHregular} for $\eps=10^{-10}$ (A) and 
$\eps=10^{-6}$ (B), with $\ze_\pm$ 30\% smaller than recommended. Recommendation based on 
$\eps=10^{-15}$ does not lead to a sizable improvement. Hence, the errors of GWR method are of the order of $10^{-10}-10^{-8}$.\\
(A):  $\ze_-=0.1559, N_-=139, \ze_+=0.1559, N_+=39$; (B): $\ze_-=0.2449, N_-=51$, $\ze_+=	0.2449$, $N_+=23$
\\ CPU time for calculation of $V_{nt}$ and $V_{ft}$ at seven spots: 8.28 msec (A) and 23.7 msec (B). Averages over
1000 runs.
}
\end{flushleft}
\label{TableNT}
 \end{table}
 
 \begin{table}
 \caption{\small Pricing down-and-out call options using Gaver-Wynn-Rho method with $M=8$ and sinh-acceleration.
 Benchmark prices  
 and errors for different choices of the parameters of the sinh-acceleration are rounded. The underlying: KoBoL model as in Table 2.
 Parameters $m_2=0.1$, $\nu=1.2, \lm=-11, \lp=12$,  $\mu=0$, $\sg=0$, 
$c=0.3026$ (rounded), $\be=-1$, $r=0.100$ (rounded). Strike $K=1.04, 1.1$, 
time to maturity $T=0.1, 0.5$, barrier $H=1$, $S$  the spot.
 }
 
 {\tiny
\begin{tabular}{c|ccccccc}
\hline\hline
$\ln (S/H)$ & 0.01 &0.03 &0.05 &	0.07 &	0.09 &0.11 & 0.13  \\
\hline
& &&$K=1.04$ &$T=0.1$&&& \\
$V_{call}$ & 0.013292606 &	0.029690703 &	0.045214455 &	0.0613371 &	0.078535419 &	0.096973726 &	0.116654007\\
$Err$ & 4.1E-10 &	3.4E-10 &	-7.7E-10 &	-3.5E-10 &	2.5E-09 &	1.3E-08 &	-2.2E-11\\\hline
&&&$K=1.04$ & $T=0.5$& && \\
$V_{call}$ & 0.021574742 &	0.047816254 &	0.071305316 &	0.093913695 &	0.11618586 &	0.138385615 &	0.1606790678\\
$Err$ & 6.7E-08 &	4.2E-08 &	-1.1E-08 &	-1.9E-08 &	-4.1E-10 &	-9.4E-10 &	-1.4E-09\\\hline
&&&$K=1.1$& $T=0.1$& && \\
$V_{call}$ & 0.006541799 &	0.014745612 &	0.023095024 &	0.032627133 &	0.043822976 &	0.056937082 &	0.07205657\\
$Err$ & 8.9E-11 &	1.3E-10 &	1.9E-10 &	7.35E-11 &	2.4E-10 &	-1.0E-08	& 4.43E-09\\\hline
&&&$K=1.1$& $T=0.5$& && \\
$V_{call}$ & 0.017419401&	0.03862725 &	0.057705653 &	0.076216112 &	0.094638672 &	0.113222686 &	0.1321308674\\
$Err$ & 2.3E-09 &	1.4E-09 &1.3E-09 &	1.0E-09 &	6.3E-10 &	4.74E-10	&3.5E-10\\\hline
\end{tabular}
}
\begin{flushleft}{\tiny
$K=1.04, T=0.1$: meshes: $\ze_\pm=0.1339$; numbers of terms: $N_-=24,58$ (for the realization of $\Pi_+$ and 1D integral in the case
$\ln(S/K)<0$), $N_+=30.47$ (for the iFT and 1D integral in the case $\log(S/K)\ge 0$), $N_-=114$ and $N_+=113$ (for calculation of the Wiener-Hopf factors; CPU time 47.8 msec (average over 1000 runs).\\
I$K=1.04, T=0.5$:  meshes  $\ze_\pm=0.1607$; $N_-=38,63$ (for the realization of $\Pi_+$ and 1D integral in the case
$\ln(S/K)<0$), $N_+=36, 61$ (for the iFT and 1D integral in the case $\log(S/K)\ge 0$), $N_-=176$ and $N_+=174$ (for calculation of the Wiener-Hopf factors; CPU time 74.5 msec (average over 1000 runs).
\\
$K=1.1, T=0.1$: meshes: $\ze_\pm=0.1753$; numbers of terms: $N_-=23,52$ (for the realization of $\Pi_+$ and 1D integral in the case
$\ln(S/K)<0$), $N_+=33, 51$ (for the iFT and 1D integral in the case $\log(S/K)\ge 0$), $N_-=124$ and $N_+=1237$ (for calculation of the Wiener-Hopf factors; CPU time 54.2 msec (average over 1000 runs).\\
$K=1.1, T=0.5$: meshes: $\ze_\pm=0.1378$; numbers of terms: $N_-=37, 73$ (for the realization of $\Pi_+$ and 1D integral in the case
$\ln(S/K)<0$), $N_+=48,72$ (for the iFT and 1D integral in the case $\log(S/K)\ge 0$), $N_-=205$ and $N_+=203$ (for calculation of the Wiener-Hopf factors; CPU time 87.6 msec (average over 1000 runs).
}
\end{flushleft}
\label{TableDOCall}

 \end{table}
 
 \begin{table}
 \caption{\small Prices $V_{call}(x), x=\ln(S/H)$ and normalized prices 
 $normV_{call}(x)=V_{call}(x)x^{-\nu/2}$ of the down-and-out call options  close to the barrier. 
 Parameters are the same as in Table 10.
 }
 
 {\tiny
\begin{tabular}{c|ccccccc}
\hline\hline
$\ln (S/H)$ & 0.0005 &	0.001 &	0.0015 &	0.002 &	0.0025 &	0.003 &	0.0035  \\
\hline
& &&$K=1.04$ &$T=0.1$&&& \\
$V_{call}$ & 0.001957852 &	0.002998404 &	0.003858716 &	0.004623158 &	0.005325463 &	0.005983272 &	0.006607213\\
$normV_{call}$ & 0.187239703 &	0.189186516 &	0.19089213 &	0.192451269 &	0.193906949 &	0.195283701 &	0.196597273\\\hline
&&&$K=1.04$ & $T=0.5$& && \\
$V_{call}$ & 0.003165882 &	0.004850658 &	0.006244885 &	0.007484685 &	0.00862443 &	0.009692532	& 0.010706088\\
$normV_{call}$ & 0.302769882	& 0.306055829 &	0.308936848 &	0.311569941 &	0.314026582	&0.316347549&0.318559081\\\hline
&&&$K=1.1$& $T=0.1$& && \\
$V_{call}$ & 0.000969788	& 0.001484144 &	0.001908772 &	0.002285612 &	0.002631453 &	0.002955087 &	0.0032618217\\
$normV_{call}$ & 0.092745937 &	0.093643172 &	0.094427692 &	0.095144674 &	0.095814566 &	0.096448958 &	0.0970553081\\\hline
&&&$K=1.1$& $T=0.5$& && \\
$V_{call}$ & 0.002557148 &	0.003917862 &	0.005043789 &	0.006044921 &	0.006965194 &	0.007827566 &	0.008645856\\
$normV_{call}$ & 0.24455353 &	0.247200368 &	0.249518145	&0.251635921 &	0.253611648 &	0.255478276 &	0.257256996\\\hline
\end{tabular}
}
\label{TableDOCallClose}
 \end{table}
\subsection{Tables for Section \ref{numer_din_example}}

 \begin{table}
 \caption{\small Hedging portfolios as bets: payoffs and probabilities (rounded)}
  {\tiny
 \begin{tabular}{l|ccccccc|c}

\hline\hline
$\tau$ 
& 0.000-0.005	& 0.0015-0.02 &	0.035-0.04 &	0.045-0.05 &	0.06-0.065 &	0.075--0.08 &	0.09-0.095 &	$>0.1$\\\hline
 Prob &0.0349 &	0.0554 &	0.0324 &	0.0250 &	0.0179 &	0.0136 &	0.0107 &	0.4339 \\\hline
 $HPCall$ &  -2.0340 &	-1.5958 &	-0.9729 &	-0.6423 &	-0.1206 &1	0.4218 &	0.9027 &	1.0101  \\\hline
SS3 & -0.6275 &	-0.6628 &	-0.7030 &	-0.7149&	-0.7126 &	-0.6780&	-0.6298 &	-0.6195 \\
SS1 & -0.9378 &	-0.8109 &	-0.6506 &	-0.5755 &	-0.4719&-0.3840&	-0.3215& -0.3082\\\hline
HP3 &-2.0075 &	-1.5800 &	-0.9705 &	-0.6459 &	-0.1322 &	0.4027 &	0.8746 &	0.9795\\
$HP(K_1)$ & -1.9986 &	-1.5704 &	-0.9625 &	-0.6401 &	-0.1320 &	0.3958 &	0.8632 &	0.9674\\
$HP(1)$ &  -1.8971 &	-1.4841 &	-0.9003 &	-0.5923 &	-0.1100 &	0.3807 &	0.7799 &	1.0085\\\hline
$HFTP2$ & -1.3220 &	-0.8995 &	-0.2980 &	0.0220 &	0.5287 &	1.0582 &	1.5298  &	0.0862\\
$HFT$ & -1.145 &	-0.7086 &	-0.0881 &	0.2413 &	0.7613 &	1.3020 &	1.7812&		-0.1526
\\\hline

\end{tabular}
  }  

\begin{flushleft}{\tiny
Bins: time intervals when the barrier is breached or not at all (the last bin)\\
$Prob$: probability that the barrier is breached during the time interval or not breached at all (the last column). Since not all events are shown, the probabilities do not sum up to 1.\\
The other rows: approximate payoffs of the initial bet $HPCall$ and values of hedging portfolios, 
in units of $V_{d.in.call}$ at initiation, if $S_\tau$ is close
to the barrier. \\
The last  column: the values at maturity, if the barrier has not been breached.\\
$SS3$: semistatic portfolio of 3 put options.\\
$SS1$: semistatic portfolio of 1 put option.\\
$HP3$: variance-minimizing portfolio constructed using put options
with strikes $K_j=1/1.04-(j-1)*0.02, j=1,2,3$)\\
$HFTP2$: variance-minimizing portfolio constructed using the first touch digital
and put options
with strikes $K_1, K_3$.\\
$HFT$ and $HP(K)$: variance-minimizing portfolio constructed using the put option
with strike $K$.
}
\end{flushleft}
\label{TableBets}
\end{table}

 \begin{table}
 \caption{\small Semistatic and mean-variance hedging of a down-and-in call option $V_{d.in.call}(T,K_0;S)$, 
 using European put options $V_{put}(T,K_j;S)$ and the first touch digital.
 Standard deviations $nStd$ and weights $w$ of options in hedging portfolios are in units of
 the price of the down-and-in call option.
 The underlying: KoBoL,
 parameters $m_2=0.1$, $\nu=1.2, \lm=-11, \lp=12$,  $\mu=0$, $\sg=0$, 
$c=0.3026$ (rounded),  $r=0.100$ (rounded). Time to maturity $T=0.1$, strike $K_0=1.04$,
barrier $H=1$, $S$  the spot.   
 }
 
  {\tiny
\begin{tabular}{l|ccccccc}
\hline\hline
$\ln (S/H)$ & 0.01& 0.02 &0.03 &0.04 &	0.05 &	0.06 & 0.07  \\
\hline\hline
$nStd$ & & & & & & &  \\\hline
$d.-in\ call$ & 2.144 &	2.585 &	3.069&	3.612 &	4.226 &	4.925 &	5.721\\\hline
$SS3$ & 8.173 &	9.067 &	10.32 &	11.84 &	13.58 &	15.67 &	18.05\\

$HP3$ & 2.068&	2.541 &	3.024 &
	3.558 &	4.146 &	4.843 &	5.619\\
$HFTP2$  & 2.027&	2.471 &	2.938 &	3.461 &	4.049&
	4.732 &	5.507\\
$HFT$ &  2.067 &	2.481 &	2.944 &	3.467 &	4.062 &	4.742 &	5.521\\
$HP(K_1) $ & 2.129 &	2.574 &	3.058 &	3.599 &	4.211 &	4.907 &	5.700\\
$HP(1) $ &2.144 &	2.585 &	3.069&	3.612 &	4.226 &	4.925 &	5.721
\\\hline\hline
Weights & \\\hline
$SS3$ &\\
$w(K_1)$&1.034 &	1.141 &	1.227&	1.305 &	1.379 &	1.452 &	1.524\\
$w(K_2)$&  -0.658&	-0.720 &	-0.768 &	-0.811&	-0.852 &	-0.891&	-0.931\\
$w(K_3)$& 0.926&	1.006 &	1.066&	1.119 &	1.170 &	1.219 &	1.269 \\\hline
 $HP3$\\
$w(K_1)$&0.080 &	0.088 &	0.095 &	0.101 &	0.106 &	0.112 &	0.117\\
$w(K_2)$& -0.160&	-0.175 &	-0.186 &	-0.197 &	-0.207 &	-0.216 &	-0.226\\
$w(K_3)$& 0.105 &	0.114 &	0.121 &	0.127 &	0.132 &	0.138 &	0.143\\\hline
$HPFTP2$\\
$w(FT)$&0.627 &	0.713 &	0.794 &	0.874 &	0.954 &	1.034 &	1.115\\
$w(K_1)$& 0.0025 &	0.0027&	0.0029 &	0.0031 &	0.0032&	0.0034 &	0.0035\\
$w(K_3)$&0.0313&	0.0340 &	0.0361 &	0.0379 &	0.0396 &	0.0412 &	0.0429\\
\hline
$w(FT)$ & 1.300 &	1.210 &	1.174 &	1.151 &	1.1500 &	1.1424 &	1.1388\\\hline
$w(K_1)$ &  0.0456 &	0.0397 &	0.0402 &	0.0422 &	0.0446 &	0.0472 &	0.0497\\\hline
$w(1)$ & 0.170 &	0.149 &	0.166 &	0.1878 &	0.208 &	0.225 &	0.238
\\\hline
\end{tabular}
}

%\begin{flushleft}{\tiny
%$SS3$: semistatic portfolio of 3 put options.\\
%$HP3$: variance-minimizing portfolio constructed using put options
%with strikes $K_j=1/1.04-(j-1)*0.02, j=1,2,3$ (strikes of the put options
%in the semi-static portfolio in Table 2)\\
%$HFTP2$: variance-minimizing portfolio constructed using the first touch digital
%and put options
%with strikes $K_1, K_3$.\\
%$HFT$ and $HP(K)$: variance-minimizing portfolio constructed using the put option
%with strike $K$.
%}
%\end{flushleft}
\label{TableStDEv}
 \end{table}
 
  \begin{table}
 \caption{\small Variance-covariance matrix (in units of 0.001), almost at the strike: $\log(S/H)=0.04$.
 $T=0.1$, barrier $H=1$}
 
 \begin{tabular}{l|ccccccc}
\hline\hline
& $Call(1.04)$ & $FT$ & $Put(K_1)$ & $Put(K_2)$ & $Put(K_3)$\\\hline
 $Call(1.04)$ &0.900 &	4.133 &	0.178 &	0.0138 &	0.0723 \\
$FT$ &  4.133 &	243.9 &	3.336&	2.207&	1.236         \\
$Put(K_1)$ & 0.178 &	3.336 &	4.895 &	2.560 &	1.699 \\
$P(K_2)$ & 0.0138&	2.207 &	2.560 &	1.844 &	1.047   \\
$P(K_3)$ &0.0723 &	1.236 &	1.699&	1.047 &	0.798 \\\hline

\end{tabular}

\begin{flushleft}{\tiny
$Call(1.04)$: down-and-in call with strike $1.04$; $FT$: first-touch digital; $Put(K_j)$: put option with strike $K_j$.

}
\end{flushleft}
\label{TableVarCov}
\end{table}

\newpage
 \section{Figures}

  \begin{figure}
\scalebox{0.5}
{\includegraphics{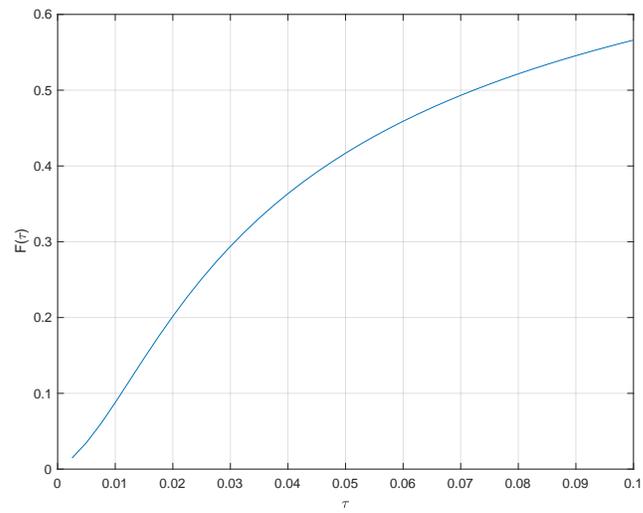}}
\caption{\small Cumulative probability distribution of the time the barrier is breached; $S_0/H=e^{0.04}$.}
\label{CpdfFtau}
\end{figure}
 
 \begin{figure}
\scalebox{0.5}
{\includegraphics{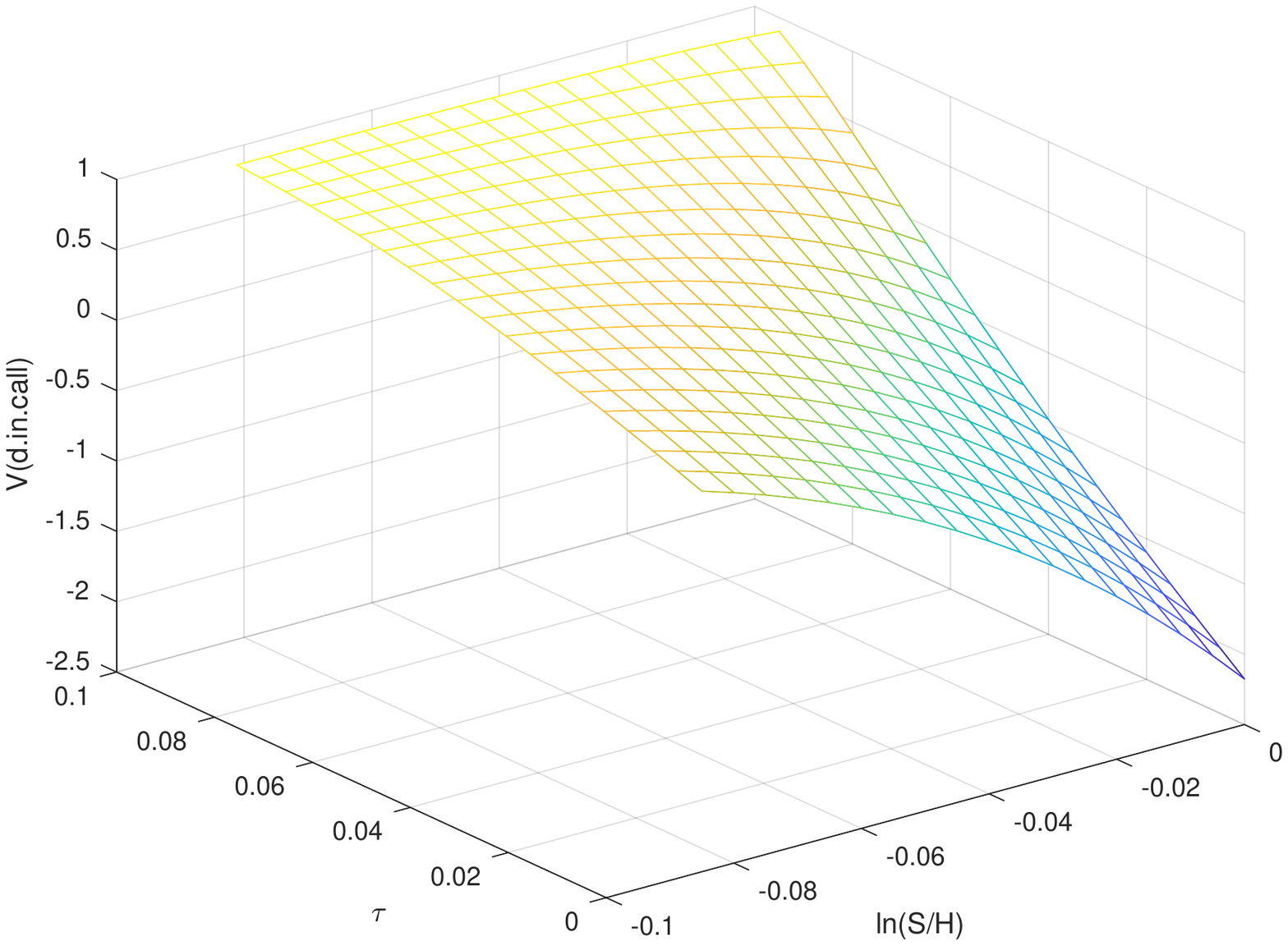}}
\caption{\small The value of the portfolio of 1 short down-and-in call option of maturity 
$T=0.1$ and strike $K=1.04$, and the riskless bonds with $B_0=V_{d.in.call}(T;K; 0, S_0)$, in units of $V_{d.in.call}(T;K; 0, S_0)$, where $S_0=e^{0.04}$,
if the barrier $H=1$ is breached at time $\tau$, and $S_\tau=S$ (no hedging).}
\label{VCall}
\end{figure}

\begin{figure}
\scalebox{0.45}
{\includegraphics{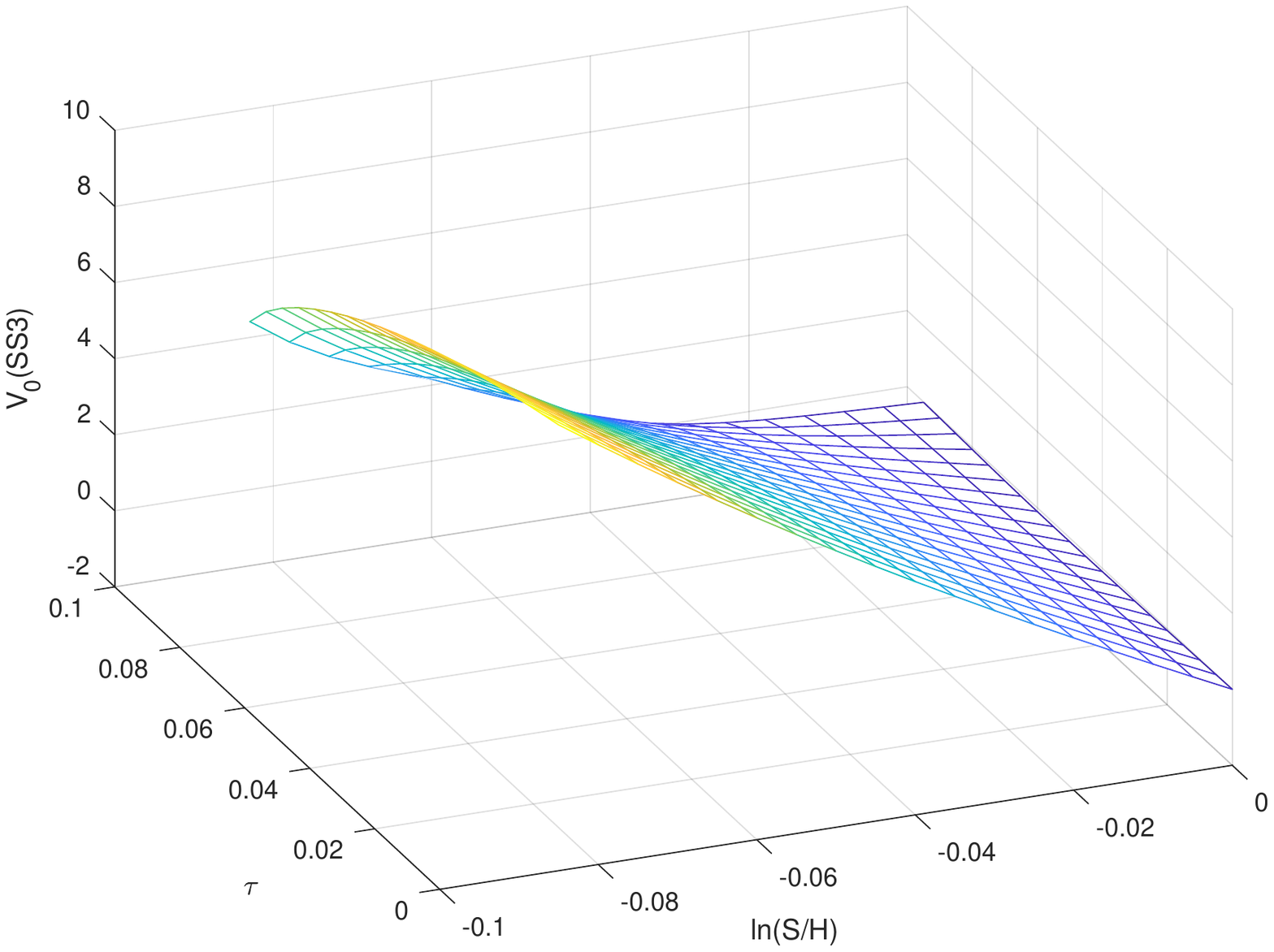}}
\caption{\small The value of the approximate semi-static hedging portfolio for 1 short down-and-in call option
of maturity $T=0.1$ and strike $K=1.04$,  in units of $V_{d.in.call}(T;K; 0, S_0)$, where $S_0=e^{0.04}$,
if the barrier $H=1$ is breached at time $\tau$, and $S_\tau=S$, {\em without counting the riskless bond component}.
Hedging instruments: put options of the same maturity and strikes $K_1=1/1.04, K_2=K_1-0.02, K_3=K_1-0.04$,
weights $w=[1.305,
-0.811,
1.119]$.}
\label{V0SS3}

\end{figure}

\begin{figure}
\scalebox{0.45}
{\includegraphics{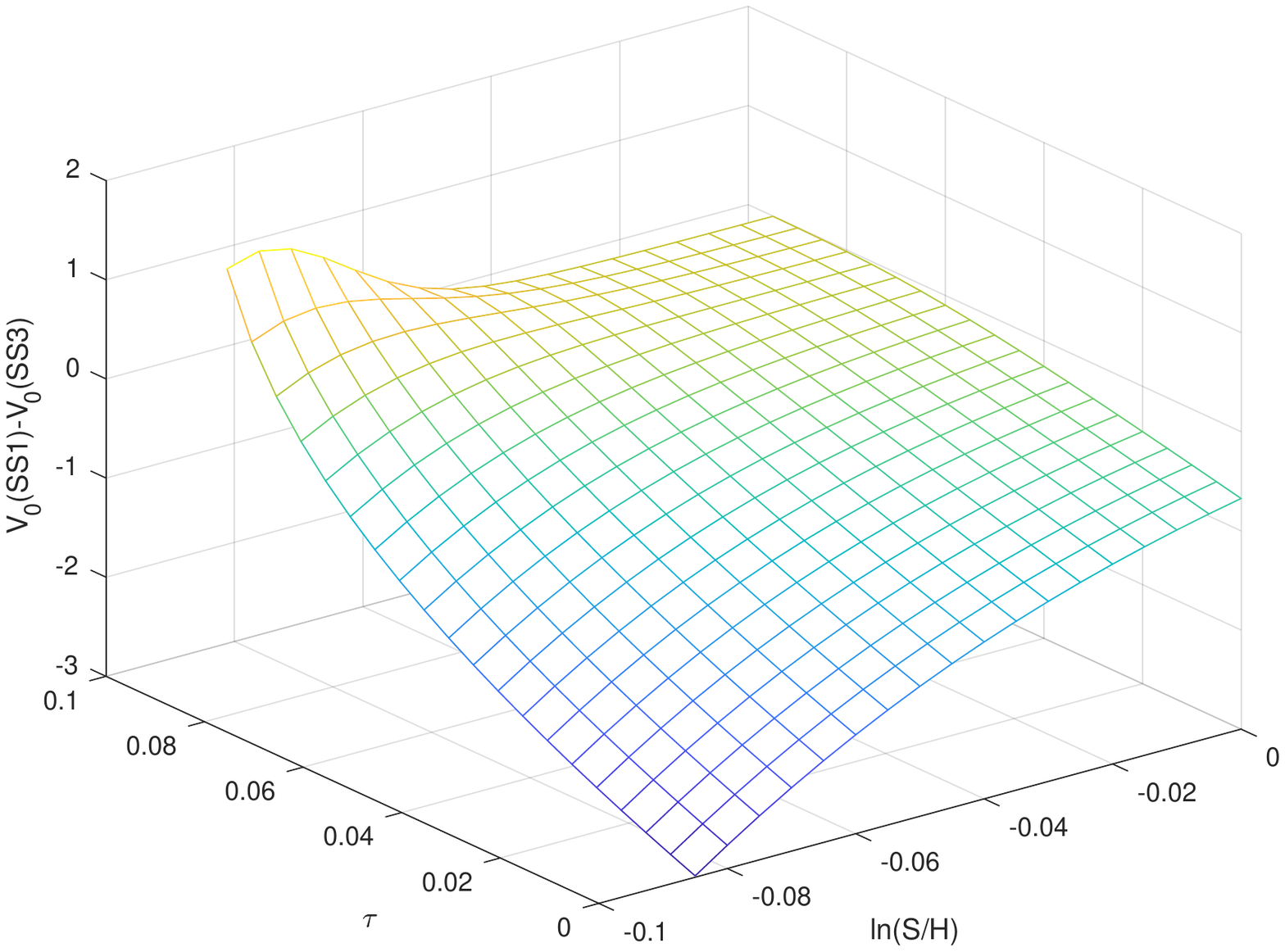}}
\caption{\small The difference $V_0(SS1)-V_0(SS3)$ of the values of the 
approximate semi-static hedging portfolios for 1 short down-and-in call option
of maturity $T=0.1$ and strike $K=1.04$,  in units of $V_{d.in.call}(T;K; 0, S_0)$, where $S_0=e^{0.04}$,
if the barrier $H=1$ is breached at time $\tau$, and $S_\tau=S$, {\em without counting the riskless bond component}.
Hedging instruments: put options of the same maturity and strikes $K_1=1/1.04, K_2=K_1-0.02, K_3=K_1-0.04$,
weights $w=[1.305,
-0.811,
1.119]$ (portfoilio $V_0(SS3)$. $V_0(SS1)$: only the option of maturity $K_1$ is used, weight $w=1.305$.}
\label{DifV0}

\end{figure}

\begin{figure}
\scalebox{0.45}
{\includegraphics{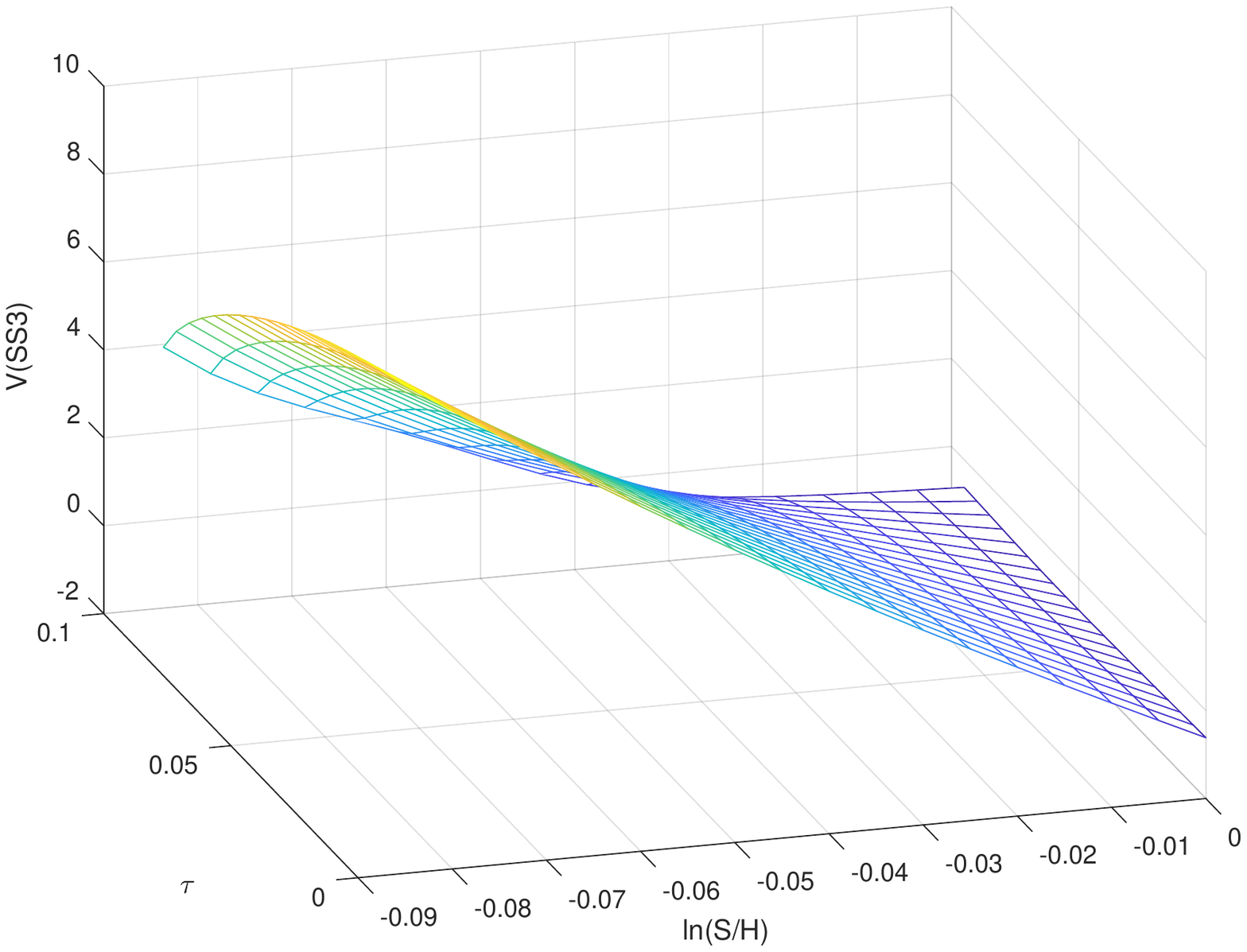}}
\caption{\small The value of the approximate semi-static hedging portfolio for 1 short down-and-in call option
of maturity $T=0.1$ and strike $K=1.04$,  in units of $V_{d.in.call}(T;K; 0, S_0)$, where $S_0=e^{0.04}$,
if the barrier $H=1$ is breached at time $\tau$, and $S_\tau=S$. The riskless bond component is taken into account.
Hedging instruments: put options of the same maturity and strikes $K_1=1/1.04, K_2=K_1-0.02, K_3=K_1-0.04$,
weights $w=[1.305,
-0.811,
1.119]$.}
\label{VSS3}

\end{figure}\begin{figure}
\scalebox{0.45}
{\includegraphics{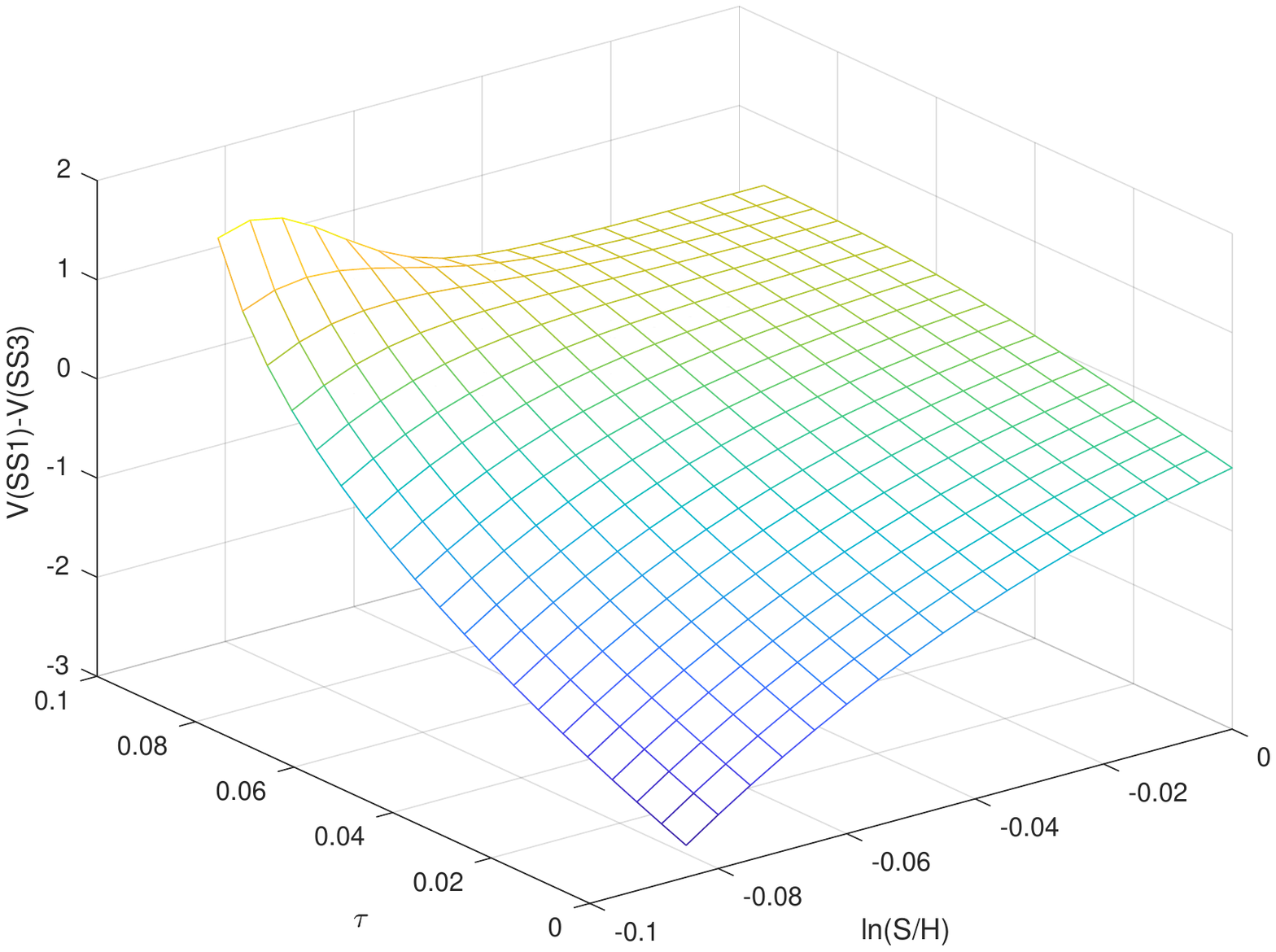}}
\caption{\small The difference $V(SS1)-V(SS3)$ of the values of the 
approximate semi-static hedging portfolios for 1 short down-and-in call option
of maturity $T=0.1$ and strike $K=1.04$,  in units of $V_{d.in.call}(T;K; 0, S_0)$, where $S_0=e^{0.04}$,
if the barrier $H=1$ is breached at time $\tau$, and $S_\tau=S$. The riskless bond component is taken into account.
Hedging instruments: put options of the same maturity and strikes $K_1=1/1.04, K_2=K_1-0.02, K_3=K_1-0.04$,
weights $w=[1.305,
-0.811,
1.119]$ (portfoilio $V(SS3)$. $V(SS1)$: only the option of maturity $K_1$ is used, weight $w=1.305$.}
\label{DifV}
\end{figure}

\begin{figure}
\scalebox{0.45}
{\includegraphics{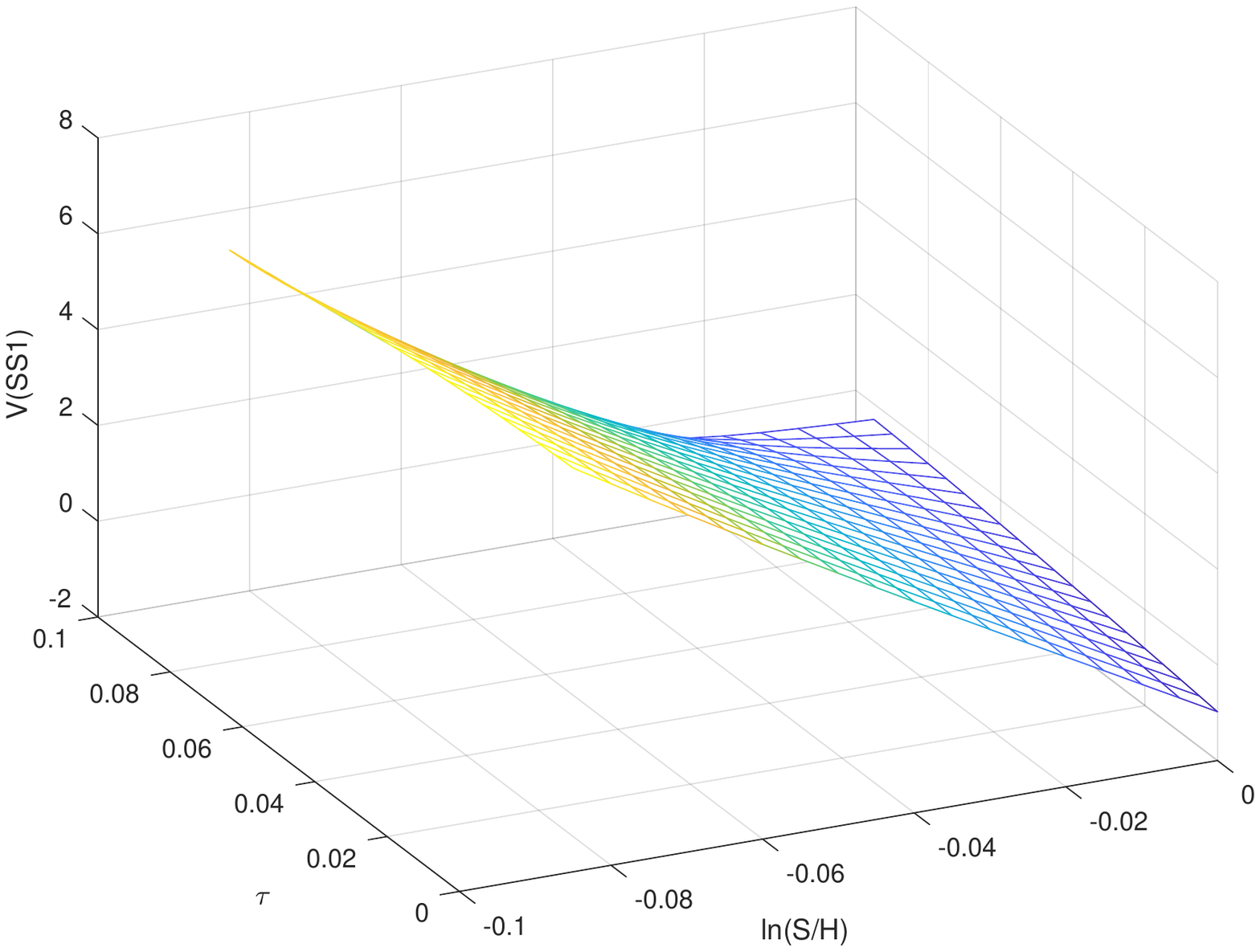}}
\caption{\small The value of the approximate semi-static hedging portfolio for 1 short down-and-in call option
of maturity $T=0.1$ and strike $K=1.04$,  in units of $V_{d.in.call}(T;K; 0, S_0)$, where $S_0=e^{0.04}$,
if the barrier $H=1$ is breached at time $\tau$, and $S_\tau=S$. The riskless bond component is taken into account.
Hedging instruments: put option of the same maturity and strike $K_1=1/1.04$,
weight $w=1.305$.}
\label{VSS1}
\end{figure}

\begin{figure}
\scalebox{0.45}
{\includegraphics{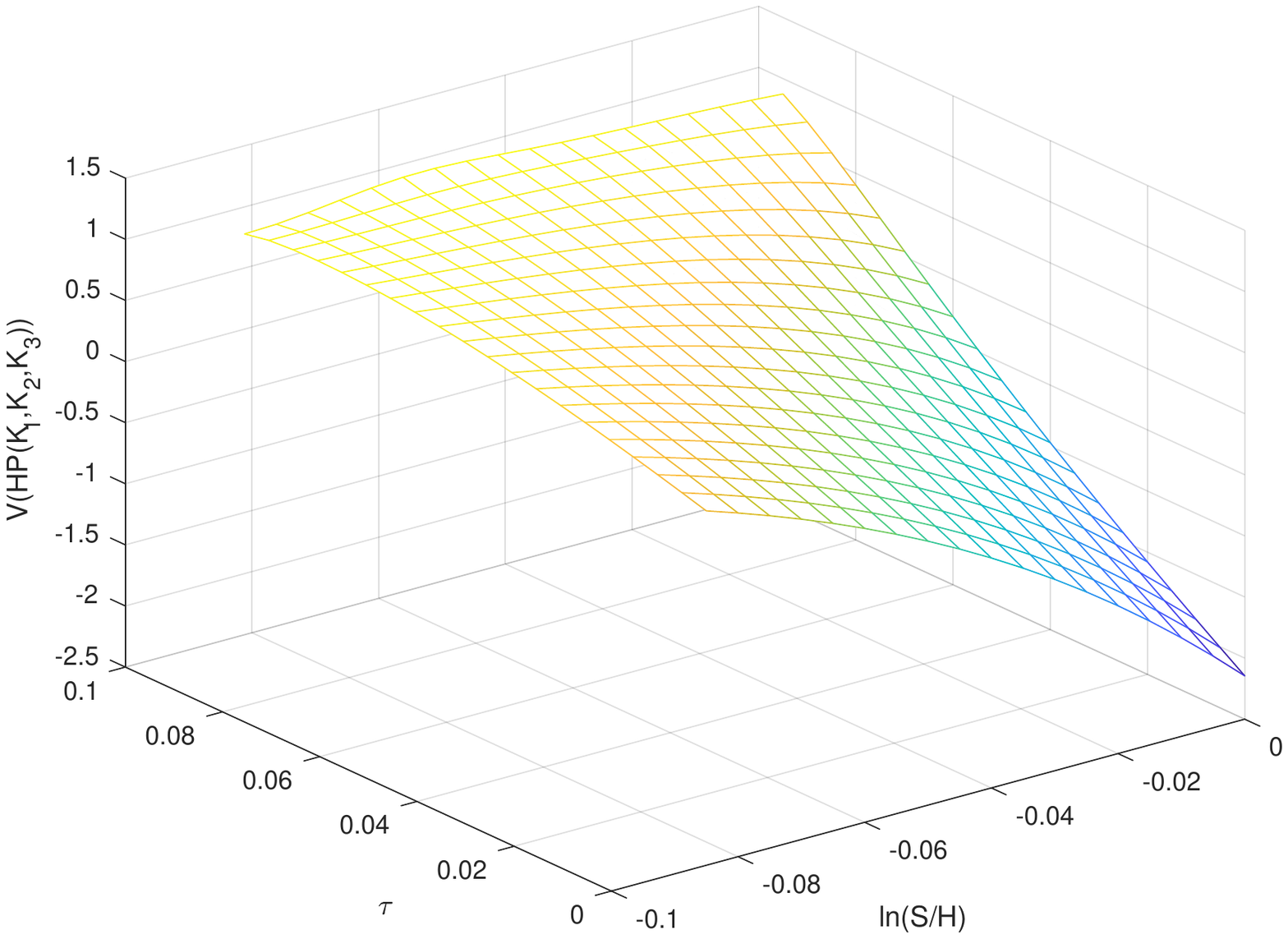}}
\caption{\small The value of the variance-minimizing hedging portfolio for 1 short down-and-in call option
of maturity $T=0.1$ and strike $K=1.04$,  in units of $V_{d.in.call}(T;K; 0, S_0)$, where $S_0=e^{0.04}$,
if the barrier $H=1$ is breached at time $\tau$, and $S_\tau=S$. The riskless bond component is taken into account.
Hedging instruments: put options of the same maturity and strikes $K_1=1/1.04, K_2=K_1-0.02, K_3=K_1-0.04$,
weights $w=[0.101, -0.197, 
0.127]$.}
\label{VP3}
\end{figure}

\begin{figure}
\scalebox{0.45}
{\includegraphics{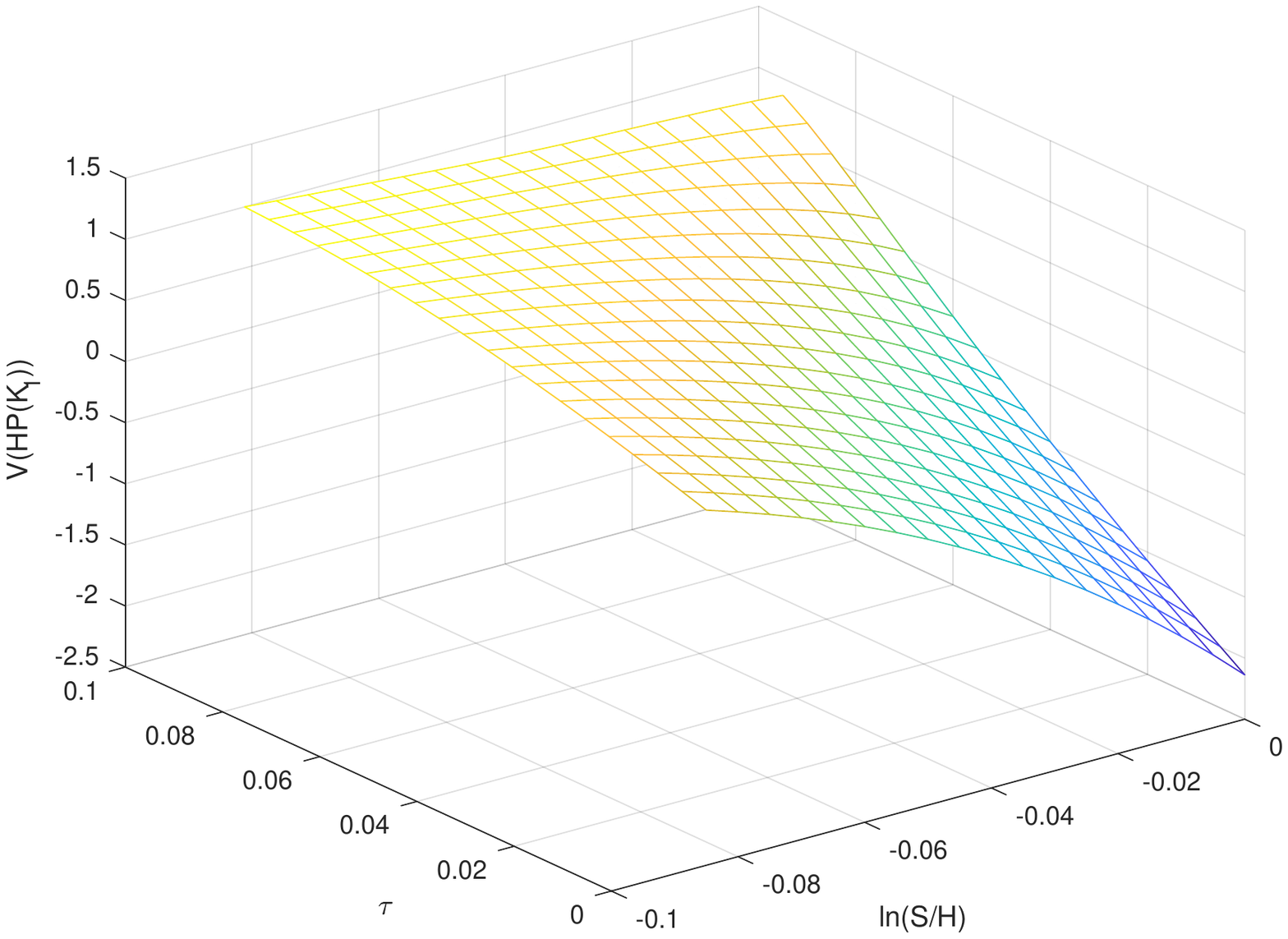}}
\caption{\small The value of the variance-minimizing hedging portfolio for 1 short down-and-in call option
of maturity $T=0.1$ and strike $K=1.04$,  in units of $V_{d.in.call}(T;K; 0, S_0)$, where $S_0=e^{0.04}$,
if the barrier $H=1$ is breached at time $\tau$, and $S_\tau=S$. The riskless bond component is taken into account.
Hedging instruments: put option of the same maturity and strike $K_1=1/1.04$,
weight $w_1=0.042$.}
\label{VK1}
\end{figure}

\begin{figure}
\scalebox{0.5}
{\includegraphics{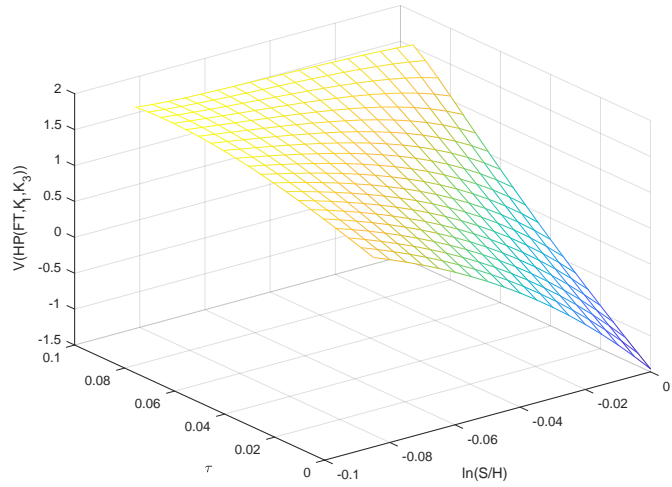}}
\caption{\small The value of the variance-minimizing hedging portfolio for 1 short down-and-in call option
of maturity $T=0.1$ and strike $K=1.04$,  in units of $V_{d.in.call}(T;K; 0, S_0)$, where $S_0=e^{0.04}$,
if the barrier $H=1$ is breached at time $\tau$, and $S_\tau=S$. The riskless bond component is taken into account.
Hedging instruments: first touch digital and put options of the same maturity and strikes $K_1=1/1.04, K_3=K_1-0.04$,
weights $w=[0.874,
0.003,
0.038]$.}
\label{VFTP2}

\end{figure}

\begin{figure}
\scalebox{0.5}
{\includegraphics{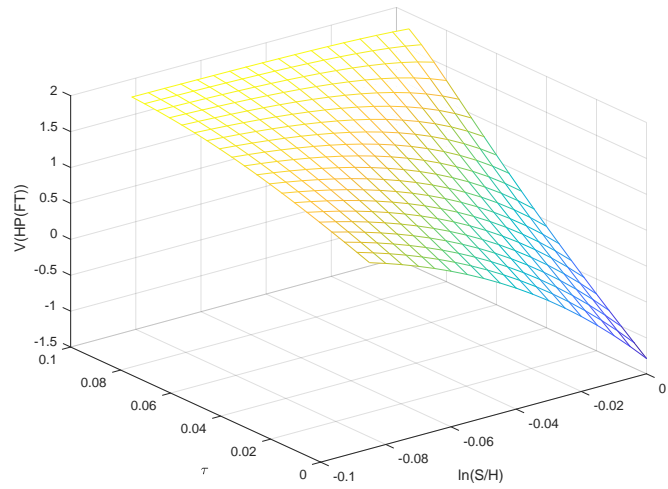}}
\caption{\small The value of the variance-minimizing hedging portfolio for 1 short down-and-in call option
of maturity $T=0.1$ and strike $K=1.04$,  in units of $V_{d.in.call}(T;K; 0, S_0)$, where $S_0=e^{0.04}$,
if the barrier $H=1$ is breached at time $\tau$, and $S_\tau=S$. The riskless bond component is taken into account.
Hedging instruments: first touch digital of the same maturity, 
weight $w=1.151$.}
\label{VFT}
\end{figure}

\begin{figure}
\scalebox{0.5}
{\includegraphics{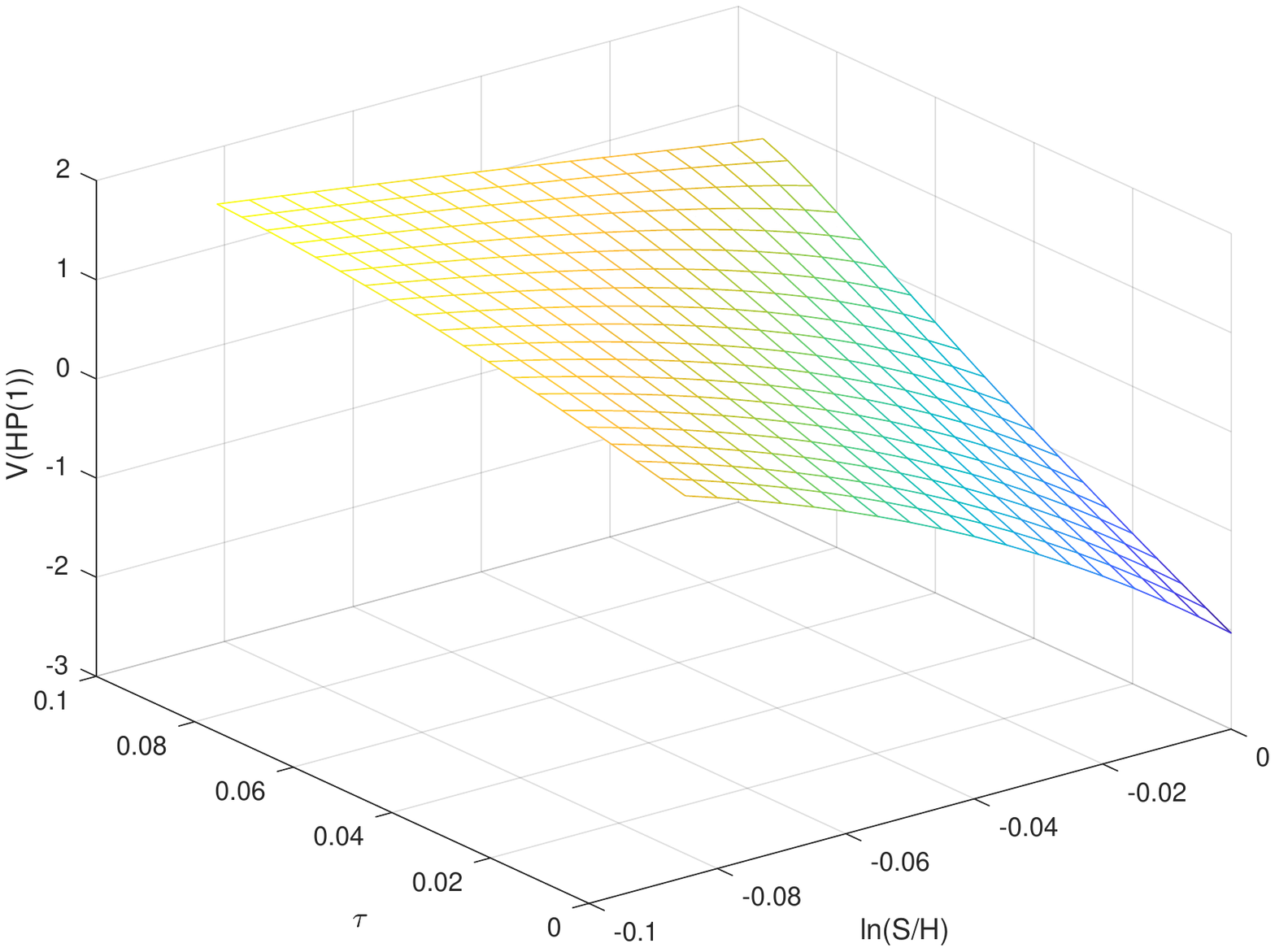}}
%{\includegraphics{FigLpLm.eps}}
\caption{\small The value of the variance-minimizing hedging portfolio for 1 short down-and-in call option
of maturity $T=0.1$ and strike $K=1.04$,  in units of $V_{d.in.call}(T;K; 0, S_0)$, where $S_0=e^{0.04}$,
if the barrier $H=1$ is breached at time $\tau$, and $S_\tau=S$. The riskless bond component is taken into account.
Hedging instruments: put option of the same maturity and strike $K_1=1$,
weight $w_1=0.188$.}
\label{VKeq1}
\end{figure}

\begin{figure}
\scalebox{0.5}
{\includegraphics{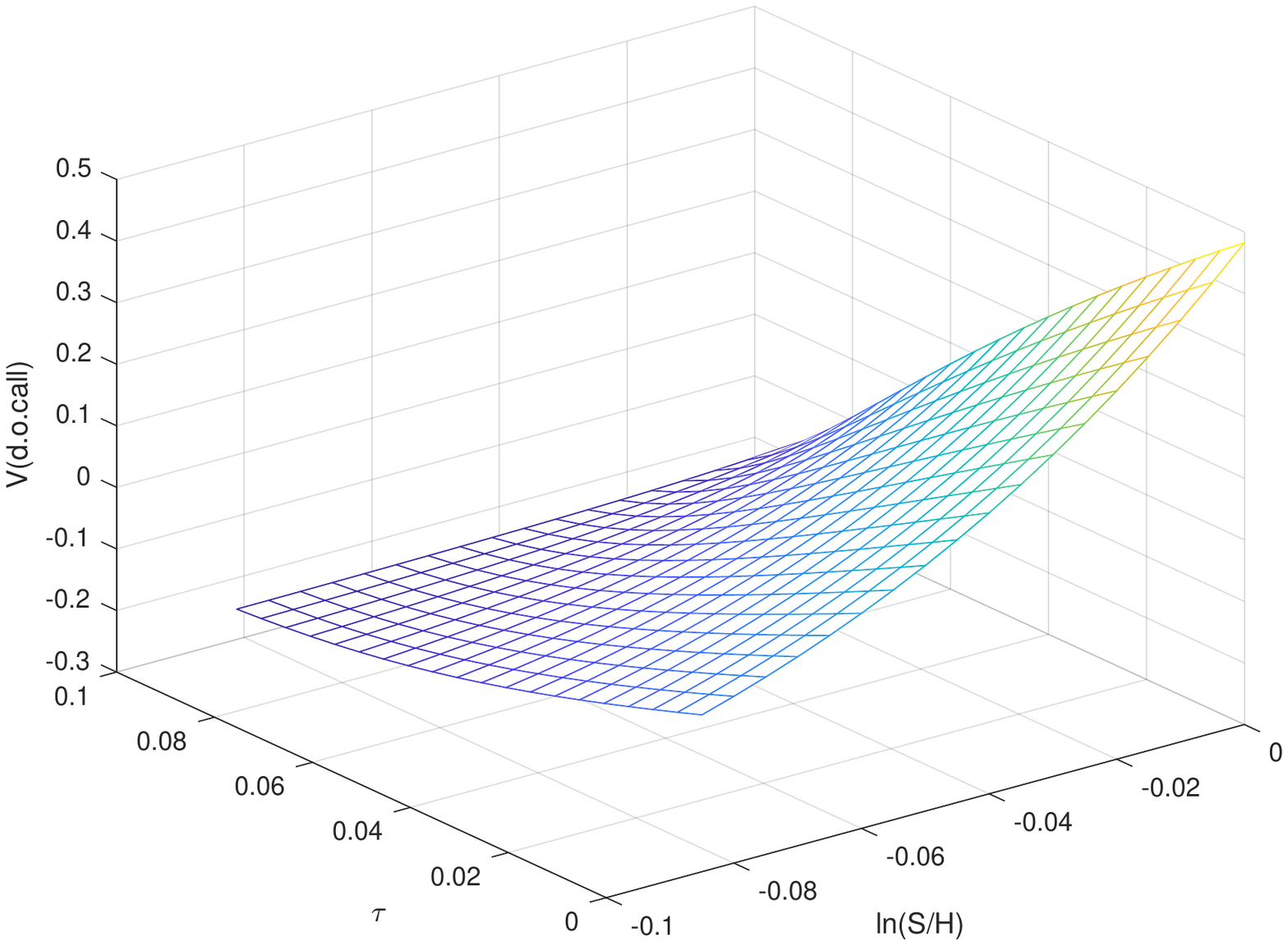}}
\caption{\small The value of the portfolio of 1 short down-and-out call option of maturity 
$T=0.1$ and strike $K=1.04$, and the riskless bonds with $B_0=V_{d.o.call}(T;K; 0, S_0)$, in units of $V_{d.o.call}(T;K; 0, S_0)$, where $S_0=e^{0.04}$,
if the barrier $H=1$ is breached at time $\tau$, and $S_\tau=S$ (no hedging).}
\label{doVCall}
\end{figure}

\begin{figure}
\scalebox{0.5}
{\includegraphics{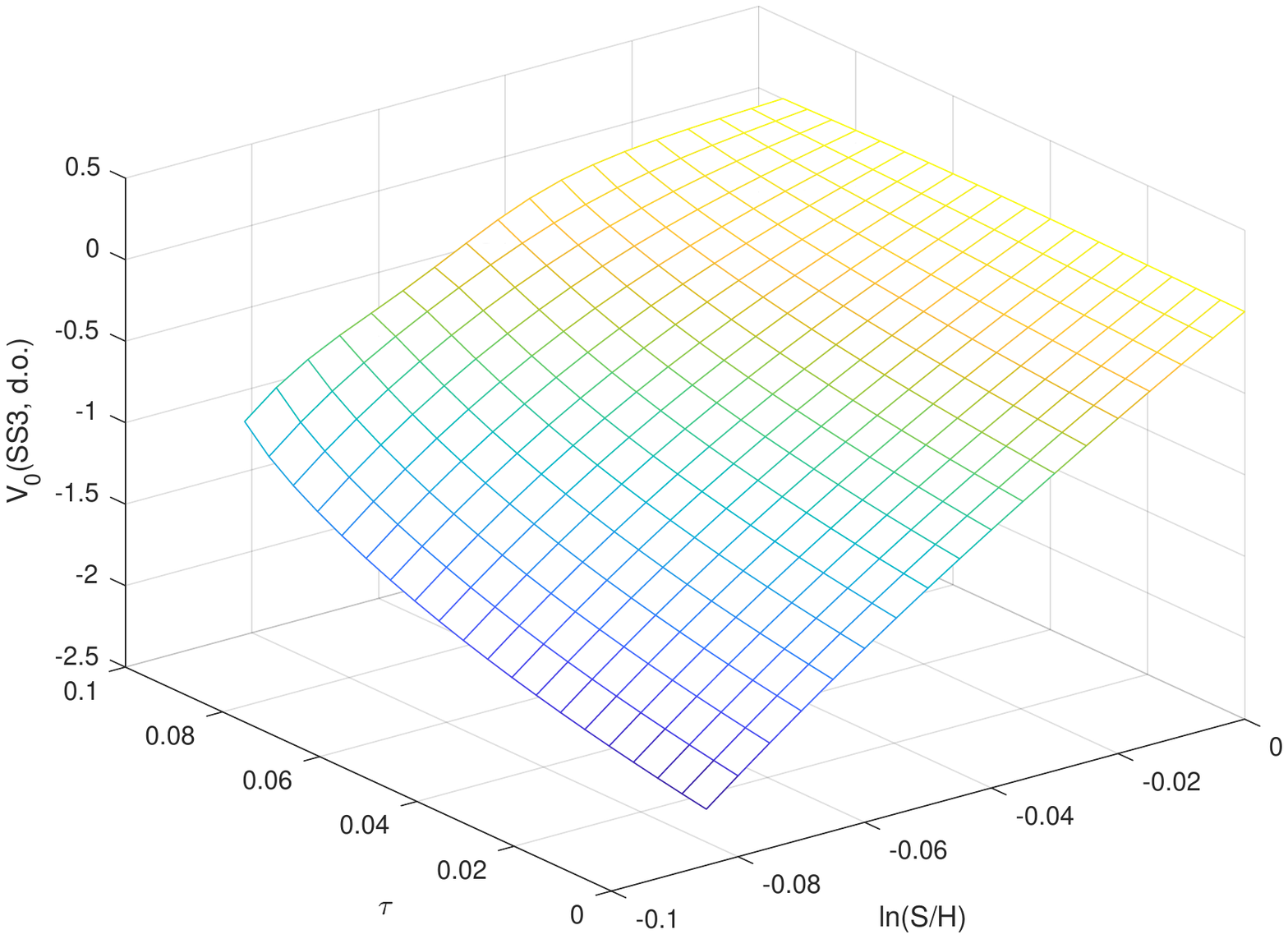}}
\caption{\small The value of the approximate semi-static hedging portfolio for 1 short down-and-out call option
of maturity $T=0.1$ and strike $K=1.04$,  in units of $V_{d.o.call}(T;K; 0, S_0)$, where $S_0=e^{0.04}$,
if the barrier $H=1$ is breached at time $\tau$, and $S_\tau=S$, {\em without counting the riskless bond component}.
Hedging instruments: put options of the same maturity and strikes $K_1=1/1.04, K_2=K_1-0.02, K_3=K_1-0.04$,
weights $w=[-1.305,
0.811,
-1.119]$.}
\label{doV0SS3}
\end{figure}

\begin{figure}
\scalebox{0.5}
{\includegraphics{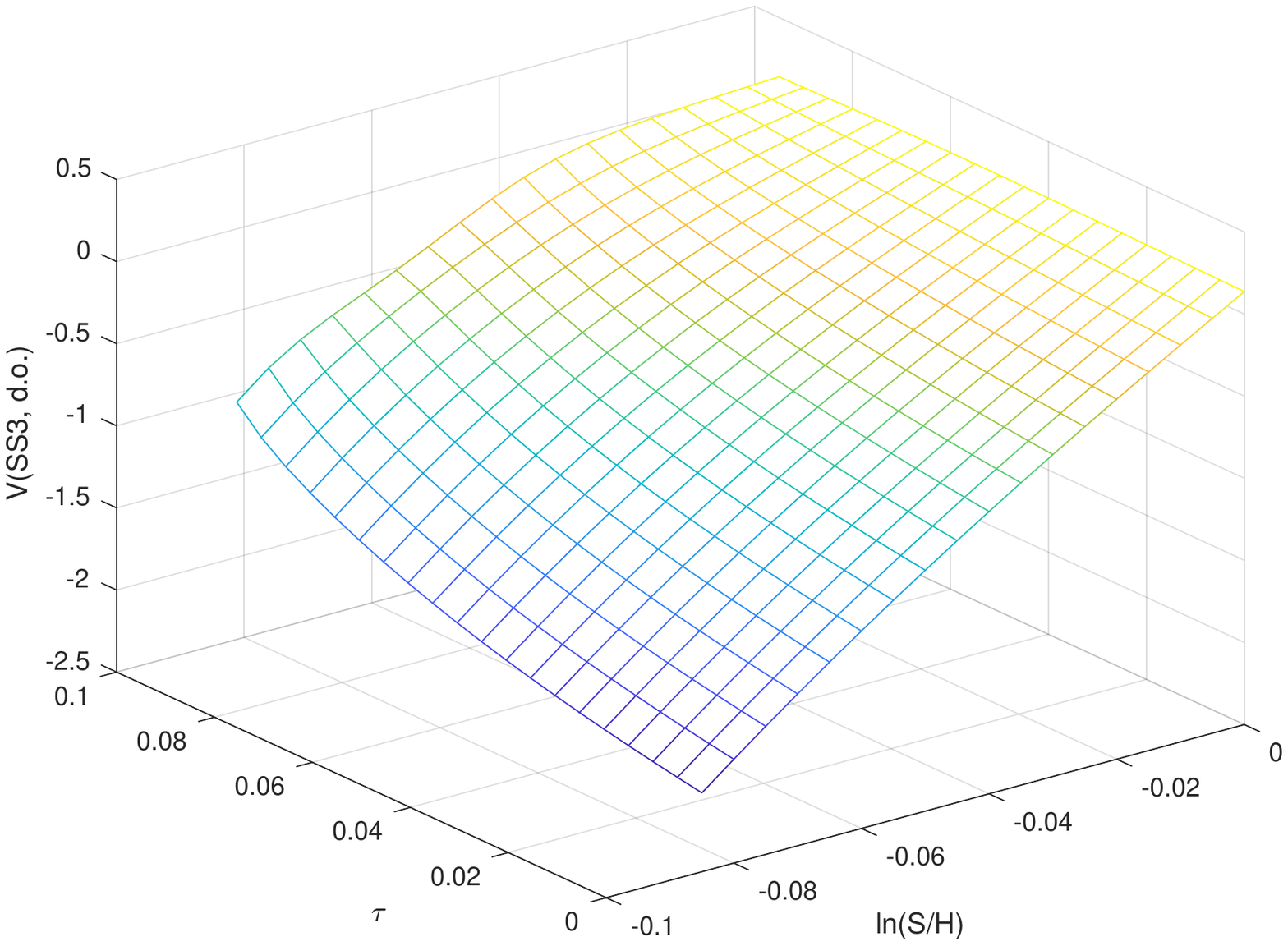}}
\caption{\small The value of the approximate semi-static hedging portfolio for 1 short down-and-out call option
of maturity $T=0.1$ and strike $K=1.04$,  in units of $V_{d.o.call}(T;K; 0, S_0)$, where $S_0=e^{0.04}$,
if the barrier $H=1$ is breached at time $\tau$, and $S_\tau=S$. The riskless bond component is taken into account.
Hedging instruments: put options of the same maturity and strikes $K_1=1/1.04, K_2=K_1-0.02, K_3=K_1-0.04$,
weights $w=[-1.305,
0.811,
-1.119]$.}
\label{doVSS3}
\end{figure}

\end{document}